\documentclass[oneside,english]{amsart}
\usepackage[T1]{fontenc}
\usepackage[latin9]{inputenc}
\usepackage{geometry}
\usepackage{amstext}
\usepackage{amsthm}
\usepackage{amssymb}
\usepackage{graphicx}
\usepackage{esint}
\PassOptionsToPackage{normalem}{ulem}
\usepackage{ulem}

\makeatletter
\numberwithin{equation}{section}
\numberwithin{figure}{section}
\theoremstyle{plain}
\newtheorem{thm}{\protect\theoremname}
\theoremstyle{plain}
\newtheorem{prop}[thm]{\protect\propositionname}
\theoremstyle{remark}
\newtheorem{rem}[thm]{\protect\remarkname}
\theoremstyle{plain}
\newtheorem{cor}[thm]{\protect\corollaryname}
\theoremstyle{plain}
\newtheorem{lem}[thm]{\protect\lemmaname}

\usepackage{hyperref}
\newcommand{\dd}{\mathrm{d}}
\newtheorem{assm}[thm]{Assumption}

\makeatother

\usepackage{babel}
\providecommand{\corollaryname}{Corollary}
\providecommand{\lemmaname}{Lemma}
\providecommand{\propositionname}{Proposition}
\providecommand{\remarkname}{Remark}
\providecommand{\theoremname}{Theorem}

\begin{document}
\title[A generalised time-evolution model for contact problems with wear]{A generalised time-evolution model for\\ contact problems with wear
and its analysis}
\author{Dmitry Ponomarev$^{1,2,3}$}
\begin{abstract}
In this paper, we revisit some classical and recent works on modelling
sliding contact with wear and propose their generalisation. Namely,
we upgrade the relation between the pressure and the wear rate by
incorporating some non-local time-dependence. To this effect, we use
a combination of fractional calculus and relaxation effects. Moreover,
we consider a possibility when the load is not constant in time. The
proposed model is analysed and solved. The results are illustrated
numerically and comparison with similar models is discussed.
\end{abstract}

\maketitle

\section{Introduction\label{sec:intro}}

\footnotetext[1]{FACTAS team, Centre Inria d'Universit{\'e} C{\^o}te d'Azur, France}\footnotetext[2]{St. Petersburg Department of Steklov Mathematical Institute of Russian Academy of Sciences, Russia}\footnotetext[3]{Contact: dmitry.ponomarev@inria.fr}
Wear of material is a complicated process which involves different
phenomena such as abrasion, adhesion and crack formation. Wear processes
have been studied for decades and numerous empirical models have been
proposed in attempt to fit experimental data for particular settings.
One classical model of wear is due to Archard \cite{Arch}. According
to this model, the wear rate is proportional to the load with a power-law
dependence. In many settings, this reduces to a linear relation between
the wear rate and the pressure (see e.g. \cite[Sect. 17.2--17.3]{Pop})
leading to several effectively solvable tribological models for sliding
of an indented punch (stamp), see e.g. \cite{ArgChai1,ArgChai2,ArgFad,Fepp,Kov,Zhu}.
It seems natural to explore a larger class of such models based on
a more general relation between the wear rate and the contact pressure.
A qualitative upgrade of this relation could be achieved by accounting
for non-local (temporal) dependencies. Investigation of how a new
wear model affects the solution of a problem would give
a hope to extend the applicability of original models to contexts
with different wear mechanisms and a broader set of materials.

Here, we consider the following instance of the two-dimensional punch
problem: a rigid wearable punch, subject to a given normal (vertical)
load, slides on a thick elastic layer or a half-plane, with a prescribed
speed. The sliding speed is taken to be constant, but the normal load
may be time-dependent. The contact area is assumed to be fixed (which
is normally expected if the load is sufficiently large, see e.g. \cite{Kom}).

Since the layer is homogeneous and material wear occurs on the interface,
the problem can be effectively described by one-dimensional integral
equations. On this level, the presence of wear in the problem manifests
itself as an additional term in the integral equation for the pressure.
This term stems from the linear Archard's wear law and brings a temporal
dependence to the problem.

In the present work, we upgrade the wear term
so that it corresponds to a more general differential relation between
the wear rate and the contact pressure. This more general relation
is meant to have two features (and combinations thereof). First, it
incorporates a wear-relaxation effect which is consistent with typical
observations that wear is the most intense in the beginning of the
sliding process. Second, the time-derivative in this relation may
be changed to a fractional order (in a Riemann-Liouville sense) which
adds another non-locality and is motivated by recent success of fractional
calculus in mechanics (viscoelasticity) and other applied contexts
\cite[Ch. 9--10]{GorKilMaiRog}. Importance of non-local relations
between wear and contact pressure is also stressed in recent work
\cite{ArgChai3}.

As we shall see, newly introduced parameters can
affect qualitative behaviour of the solution. Namely, depending on
a choice of the parameter values, the solution may exhibit exponential
or algebraic decay in time that would range between monotone and arbitrary
oscillatory. Therefore, such a generalised model is highly desirable
and is expected to be useful for fitting experimentally observable
data. On the other hand, this model is almost explicitly solvable,
meaning that the solution can be written up in a closed form in terms
of some auxiliary functions which could be precomputed (numerically
or asymptotically). In particular, it is rather straightforward to
analyse long-time behaviour of this model, namely, its convergence
speed to the stationary pressure distribution. Moreover, computing
such a pressure profile itself is a problem of an essential practical
interest. 

The outline of the paper is as follows. We formulate a new model in
Section \ref{sec:model}. In Section \ref{sec:model_deriv}, we derive
the proposed form of the wear relation and recast the model in form
of a single integral equation. Then, in Section \ref{sec:sol}, we
deal with the solution of the general model and its analysis. In particular,
in Subsection \ref{subsec:sol_exist}, under appropriate conditions,
we deduce the existence of a unique solution and provide its explicit
form in terms of spectral functions of a pertinent integral operator.
This is detailed even more for a concrete choice of the kernel function
in Subsection \ref{subsec:concr_K}. Furthermore, in Section \ref{sec:sol_analys},
we show that the solution of the proposed model has the anticipated
behaviour consistent with previous models. A collection of mathematical
results needed for Sections \ref{sec:sol}--\ref{sec:sol_analys}
is outsourced to Appendix. Section \ref{sec:numerics} is dedicated
to the numerical solution of the proposed model: its validation and
exploration. Finally, we conclude with Section \ref{sec:discuss}
where we summarise and discuss the results outlining potential further
work.

\section{The model\label{sec:model}}

According to \cite{AleksKov1,ArgChai1,ArgChai2,ArgFad,Gal2,Kom,Vor}, a
balance of displacements at the contact interface (see Figure \ref{fig:stamp}) yields an integral
equation which involves two unknowns: the contact pressure $p\left(x,t\right)$
and the indentation $\delta\left(t\right)$. This equation can be
formulated as follows
\begin{equation}
\eta\left(p\left(x,t\right)-p\left(x,0\right)\right)+\int_{-a}^{a}K\left(x-\xi\right)\left(p\left(\xi,t\right)-p\left(\xi,0\right)\right)\dd\xi+w\left[p\right]\left(x,t\right)=\delta\left(t\right)-\delta\left(0\right),\hspace{1em}x\in\left(-a,a\right),\hspace{1em}t>0.\label{eq:p_delta_eq}
\end{equation}
Furthermore, (\ref{eq:p_delta_eq}) should be complemented by the
equilibrium condition 
\begin{equation}
\int_{-a}^{a}p\left(x,t\right)\dd x=P\left(t\right),\hspace{1em}t\geq0,\label{eq:equil_eq}
\end{equation}
and the integral equation for the initial pressure distribution $p\left(x,0\right)$
\begin{equation}
\eta p\left(x,0\right)+\int_{-a}^{a}K\left(x-\xi\right)p\left(\xi,0\right)\dd\xi=\delta\left(0\right)-\Delta\left(x\right),\hspace{1em}x\in\left(-a,a\right).\label{eq:p0_int_eq}
\end{equation}
Here, $K\left(x\right)$ is a given function which links the contact pressure to the vertical displacement, $\eta\geq0$ is a given constant, $\left(-a,a\right)$ is the contact interval, $P\left(t\right)>0$
is the contact load, $w\left[p\right]\left(x,t\right)$ is the wear
term which is a given mapping of the contact pressure $p\left(x,t\right)$
and the punch profile $\Delta\left(x\right)$ is a known function.

The main feature of our model is a new form of the wear term $w\left[p\right]\left(x,t\right)$
entering (\ref{eq:p_delta_eq}). This wear term is given by 
\begin{equation}
w\left[p\right]\left(x,t\right)=-\nu\mu^{1/\alpha-1}\int_{0}^{t}\mathcal{E}_{\alpha}\left(\mu^{1/\alpha}\left(t-\tau\right)\right)p\left(x,\tau\right)\dd\tau,\label{eq:w_p_term}
\end{equation}
where $\mathcal{E}_{\alpha}$ is an auxiliary function defined by
(\ref{eq:cal_E_a}) and $\mu>0$, $\alpha\in\left(0,2\right)$ are given constants.

\begin{figure}
\begin{centering}
\includegraphics[scale=0.6]{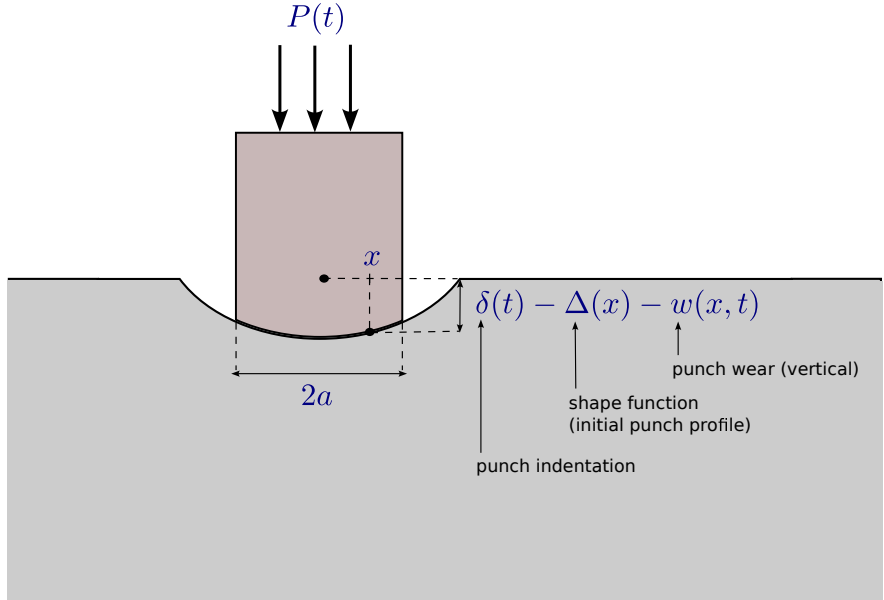}
\par\end{centering}
\caption{\label{fig:stamp} Geometry of the problem. A balance of displacements under the punch is governed by the equation: $\eta p\left(x,t\right)+\int_{-a}^{a}K\left(x-\xi\right)p\left(\xi,t\right)\dd\xi=\delta\left(t\right)-\Delta\left(x\right)-w\left(x,t\right)$, $x\in\left(-a,a\right)$, $t\geq 0$.}
\end{figure}

\subsection*{Relation to other models}

The considered model is an amalgam as it attempts to treat different
settings at once as we shall discuss here.

First of all, we perform most of the analysis for a rather general
form of the kernel function $K\left(x\right)$ entering (\ref{eq:p_delta_eq})
and (\ref{eq:p0_int_eq}) while only making some assumptions on it.
This allows dealing with an elastic half-plane or an elastic layer
subject to different boundary conditions at its foundation, see e.g.
\cite[Par. 11]{AlblKuip1,AlblKuip2,Vor}. In Subsection \ref{subsec:concr_K},
we also provide a more detailed study for a concrete kernel function
pertinent to the half-plane geometry.

Works \cite{ArgFad,ArgChai1,ArgChai2,Kom} involve counterparts of
(\ref{eq:p_delta_eq}) and (\ref{eq:p0_int_eq}) with $\eta=0$. The
case $\eta>0$ accounts for the surface roughness effect or coating
\cite{AleksKov2,Kov} which is due to an additional deformation proportional
to the local pressure \cite[Ch. 2 Par. 8]{Shtaer}. We note that
some aspects of the solution construction and its analysis differ
depending on whether $\eta=0$ or $\eta>0$, and this will be reflected
in our formulations of theorems and propositions given in Sections
\ref{sec:sol}--\ref{sec:sol_analys}.

In most settings, $P\left(t\right)\equiv P_{0}$ was taken in (\ref{eq:equil_eq})
for some constant $P_{0}>0$. A possibility of a non-constant load
is mentioned in \cite{ArgChai1} but not studied. We study the more
general load since, on the one hand, it appears to be of a direct
physical interest and, on the other hand, the presented solution of
the model can be adapted for this case, even though some computational
details become more cumbersome. Despite this generality, we perform
a more concrete analysis (see Section \ref{sec:sol_analys}) for two
particular type of loads: the most important case of a constant load
and the case of a transitional load (when the load may vary in time
but is eventually constant). 

We shall now briefly discuss some relevant works that would motivate
our new model based on a peculiar choice of the wear term. \\
As discussed in Section \ref{sec:intro}, a simple but powerful version
of the wear term is derived from the linear Archard's law. According
to this relation, the wear is linearly proportional to the total sliding
distance and the contact pressure. Consequently, we have
\begin{equation}
\partial_{t}w\left(x,t\right)=\nu p\left(x,t\right)\label{eq:w_Archard_diff}
\end{equation}
for some constant $\nu>0$ related to the material and the sliding
speed. Hence, since at the initial moment $t=0$ the worn material
is absent, i.e. $w\left(x,0\right)\equiv0$ for $x\in\left(-a,a\right)$,
(\ref{eq:w_Archard_diff}) is equivalent to
\begin{equation}
w\left[p\right]\left(x,t\right)\equiv w\left(x,t\right)=\nu\int_{0}^{t}p\left(x,\tau\right)\dd\tau.\label{eq:w_Archard_int}
\end{equation}
Recently, in \cite{ArgChai1}, the authors considered a generalised
version of (\ref{eq:w_Archard_int}), namely,
\[
w\left[p\right]\left(x,t\right)=k_{0}\left(x\right)\int_{0}^{t}p\left(x,\tau\right)\dd\tau
\]
for some suitable function $k_{0}\left(x\right)$. Another generalisation was proposed in \cite{ArgChai4} which, in its simple form, corresponds to
\[
w\left[p\right]\left(x,t\right)=\int_{0}^{t}k_{1}\left(\tau\right)p\left(x,\tau\right)\dd\tau,
\]
but can also include additional dependencies of $k_1$.

Numerous models with nonlinear wear have been considered, in this
case we have (see e.g. \cite{Zhu} and \cite[eq. (6.82)]{Goryach})
\[
w\left[p\right]\left(x,t\right)=\nu\int_{0}^{t}p^{\gamma}\left(x,\tau\right)\dd\tau,
\]
with some nonlinearity parameter $\gamma>0$.

Finally, we mention that, in addition to the wear term, the friction
effects were also considered in \cite{ArgChai2}, by incorporating
in (\ref{eq:p_delta_eq}) another term proportional to $\int_{-a}^{x}p\left(\xi,t\right)\dd\xi$.

\section{Derivation of the model\label{sec:model_deriv}}

While the generalisation of the previous setting by including an extra
term proportional to the solution and allowing load to vary in time
is natural, choosing a complicated looking wear term is a nontrivial
feature of the present model that we shall now discuss. 

\subsection{Towards the wear term (\ref{eq:w_p_term})}

We consider a generalisation of (\ref{eq:w_Archard_diff})--(\ref{eq:w_Archard_int})
that can be obtained in two independent steps. \\
First, we add to the right-hand side of (\ref{eq:w_Archard_diff})
a term $-\mu w\left(x,t\right)$ with some constant $\mu\geq0$. This
would result in (\ref{eq:w_Archard_int}) being replaced by
\begin{equation}
w\left[p\right]\left(x,t\right)=\nu\int_{0}^{t}e^{-\mu\left(t-\tau\right)}p\left(x,\tau\right)\dd\tau,\label{eq:w_nonfrac}
\end{equation}
where we took into account the condition $w\left(x,0\right)=0$ representing
the initial absence of the worn material. We note that models with
such a hereditary wear term has already been considered, see e.g.
\cite{YevtPyr}.\\
At the second step, we replace the differentiation in time in the
left-hand side of (\ref{eq:w_Archard_diff}) with a more general notion
of the derivative. In particular, we use the fractional Riemann-Liouville
derivative of order $\alpha\in\left(0,1\right]$ (for the values $\alpha>1$ see remark at the end of the paragraph) that we denote $D_{t}^{\alpha}$,
namely,
\begin{equation}\label{eq:der_RL_def}
D_{t}^{\alpha}w\left(x,t\right):=\frac{1}{\Gamma\left(1-\alpha\right)}\frac{\partial}{\partial t}\int_{0}^{t}\frac{w\left(x,\tau\right)}{\left(t-\tau\right)^{\alpha}}\dd\tau,
\end{equation}
where $\Gamma$ denotes Gamma function defined in Appendix. Note that,
due to the condition $w\left(x,0\right)=0$, this derivative would
also coincide with the fractional Caputo derivative (see e.g. \cite[eq. (1.19)]{GorMain}).

In other words, we choose to generalise (\ref{eq:w_Archard_diff})
to the following non-homogeneous fractional order ODE
\begin{equation}
D_{t}^{\alpha}w\left(x,t\right)=-\mu w\left(x,t\right)+\nu p\left(x,t\right).\label{eq:w_gener}
\end{equation}

We are now going to obtain the integral form of (\ref{eq:w_gener})
and hence the corresponding relation $w\left[p\right]\left(x,t\right)$.
To this end, we set $\tilde{t}:=\mu^{1/\alpha}t$, $\tilde{w}\left(x,\tilde{t}\right):=w\left(x,\tilde{t}/\mu^{1/\alpha}\right)$,
$\tilde{p}\left(x,\tilde{t}\right):=p\left(x,\tilde{t}/\mu^{1/\alpha}\right)$.
Then, relation (\ref{eq:w_gener}) rewrites as
\begin{equation}
D_{\tilde{t}}^{\alpha}\tilde{w}\left(x,\tilde{t}\right)=D_{\tilde{t}}^{\alpha}\left(\tilde{w}\left(x,\tilde{t}\right)-\tilde{w}\left(x,0\right)\right)=-\tilde{w}\left(x,\tilde{t}\right)+\frac{\nu}{\mu}\tilde{p}\left(x,\tilde{t}\right),\label{eq:w_tilde_eq}
\end{equation}
where we used that $\tilde{w}\left(x,0\right)=w\left(x,0\right)=0$.
A fractional differential equation in the form (\ref{eq:w_tilde_eq})
is solved in \cite[eq. (3.13)]{GorMain} using Laplace transforms.
Taking into account that $\tilde{w}\left(x,0\right)=0$ and $\alpha\in\left(0,1\right]$,
the obtained solution is 
\[
\tilde{w}\left(x,\tilde{t}\right)=-\frac{\nu}{\mu}\int_{0}^{\tilde{t}}\tilde{p}\left(x,\tilde{t}-\tilde{\tau}\right)\mathcal{E}_{\alpha}\left(\tilde{\tau}\right)\dd\tilde{\tau}.
\]
Therefore,
\begin{equation}
w\left[p\right]\left(x,t\right)=\tilde{w}\left(x,\mu^{1/\alpha}t\right)=-\nu\mu^{1/\alpha-1}\int_{0}^{t}p\left(x,t-\tau\right)\mathcal{E}_{\alpha}\left(\mu^{1/\alpha}\tau\right)\dd\tau,\label{eq:w_p_final}
\end{equation}
and thus (\ref{eq:w_p_term}) follows. 
Observe that the assumed range $\alpha\in\left(0,1\right]$ can be extended to values $\alpha>1$ as relation (\ref{eq:w_p_term}) remain meaningful (note, however, that definition (\ref{eq:der_RL_def}) should be replaced with a more general one \cite[eq. (1.13)]{GorMain}). 

Equations (\ref{eq:p_delta_eq})--(\ref{eq:p0_int_eq}) together
with the wear term given by (\ref{eq:w_p_term}) constitute the general
model proposed in the present work. 
\subsection{Consistency}

It is noteworthy that when $\mu=0$ in (\ref{eq:w_gener}), we take
the limit $\mu\rightarrow0$ in (\ref{eq:w_p_term}). In doing so,
we use the asymptotics (\ref{eq:cal_E_a_small}) to obtain
\begin{equation}
w\left[p\right]\left(x,t\right)=\frac{\alpha\nu}{\Gamma\left(1+\alpha\right)}\int_{0}^{t}p\left(x,t-\tau\right)\tau^{\alpha-1}\dd\tau=\frac{\nu}{\Gamma\left(\alpha\right)}\int_{0}^{t}\frac{p\left(x,\tau\right)}{\left(t-\tau\right)^{1-\alpha}}\dd\tau.\label{eq:RL_int}
\end{equation}
The right-hand side of (\ref{eq:RL_int}) contains the Riemann-Liouville
integral of order $\alpha$ which is consistent with the inversion
of the respective fractional derivative appearing in (\ref{eq:w_gener})
(see \cite[eqs. (1.29), (1.2)]{GorMain}).

On the other hand, by taking $\alpha=1$ in (\ref{eq:w_p_term}),
we have $\mathcal{E}_{1}\left(z\right)=-e^{-z}$, and hence we recover
the non-fractional order model corresponding to (\ref{eq:w_nonfrac}). 

Hence, our present model can be succinctly characterised as a model
with the wear term interpolating between the exponential relaxation
and a new fractional order model.

\subsection{The integral equation}

Equations (\ref{eq:p_delta_eq})--(\ref{eq:p0_int_eq}) with the
wear term (\ref{eq:w_p_term}) can be effectively reduced to a single
equation. To achieve that, we integrate (\ref{eq:p_delta_eq}) in
the $x$ variable over $\left(-a,a\right)$ and use (\ref{eq:equil_eq})
and (\ref{eq:w_p_term}). This yields
\begin{equation}
\eta\left(P\left(t\right)-P\left(0\right)\right)+\int_{-a}^{a}K_{1}\left(\xi\right)\left(p\left(\xi,t\right)-p\left(\xi,0\right)\right)\dd\xi-\nu\mu^{1/\alpha-1}\int_{0}^{t}\mathcal{E}_{\alpha}\left(\mu^{1/\alpha}\left(t-\tau\right)\right)P\left(\tau\right)\dd\tau=2a\left(\delta\left(t\right)-\delta\left(0\right)\right),\hspace{1em}t>0,\label{eq:p_delta_eq_int}
\end{equation}
where we introduced the function $K_{1}\left(\xi\right):=\int_{-a}^{a}K\left(\zeta-\xi\right)\dd\zeta$.

Dividing (\ref{eq:p_delta_eq_int}) over $2a$ and subtracting the
resulting equation from (\ref{eq:p_delta_eq}), we eliminate $\delta\left(t\right)-\delta\left(0\right)$
and obtain, for $x\in\left(-a,a\right)$, $t>0$,
\begin{align}
\eta\left(p\left(x,t\right)-p\left(x,0\right)-\frac{P\left(t\right)-P\left(0\right)}{2a}\right)+\int_{-a}^{a}\left[K\left(x-\xi\right)-\frac{1}{2a}K_{1}\left(\xi\right)\right]\left(p\left(\xi,t\right)-p\left(\xi,0\right)\right)\dd\xi=\label{eq:p_eq}\\
=\frac{\nu}{\mu^{1-1/\alpha}}\int_{0}^{t}\mathcal{E}_{\alpha}\left(\mu^{1/\alpha}\left(t-\tau\right)\right)\left(p\left(x,\tau\right)-\frac{P\left(\tau\right)}{2a}\right)\dd\tau.\nonumber 
\end{align}
This is an integral equation featuring only one unknown $p\left(x,t\right)$.
Note that $p\left(x,0\right)$ can be assumed to be known from solving
(\ref{eq:p0_int_eq}) with the unknown $\delta\left(0\right)$ that
is found \textit{a posteriori} from condition (\ref{eq:equil_eq})
imposed at $t=0$. As we shall further discuss in Subsection \ref{subsec:sol_exist}
(see Corollary \ref{cor:eta_pos}), the assumption on the knowledge
of $p\left(x,0\right)$ does not restrict generality as long as $\eta>0$. 

The obtained relation (\ref{eq:p_eq}) is an integral equation of
the mixed type in two variables: it contains both Fredholm and Volterra
integral operators in the spatial and temporal parts, respectively. 

\section{Solution by separation of variables\label{sec:sol}}

Since temporal and spatial operators appearing in equation (\ref{eq:p_eq})
are not intertwined, it is natural to attempt solving the problem
with the method of separation of variables. Therefore, we first focus
on the spatial part of the problem and study the appropriate functional
setting for its solution. 

\subsection{\label{subsec:prelims} Notation, assumptions and their implications}

We start by fixing the notation. We shall use $L^{p}\left(-a,a\right)$
to denote the space of all real-valued functions on $\left(-a,a\right)$
whose $p$-th power is Lebesgue integrable on the interval $\left(-a,a\right)$.
In particular, $L^{2}\left(-a,a\right)$ is a Hilbert space endowed
with the inner product and norm defined as 
\begin{equation}
\left\langle f,g\right\rangle :=\int_{-a}^{a}f\left(x\right)g\left(x\right)\dd x,\hspace{1em}\left\Vert f\right\Vert :=\sqrt{\left\langle f,f\right\rangle },\hspace{1em}f,g\in L^{2}\left(-a,a\right).\label{eq:inner_prod}
\end{equation}
We will also need the space $L_{0}^{2}\left(-a,a\right)$ a subspace
of $L^{2}\left(-a,a\right)$ that consists of all square-integrable
functions with vanishing mean value on $\left(-a,a\right)$, i.e.
\begin{equation}
L_{0}^{2}\left(-a,a\right):=\left\{ f\in L^{2}\left(-a,a\right):\;\int_{-a}^{a}f\left(x\right)\dd x=0\right\} .\label{eq:L0_def}
\end{equation}
Note that $L_{0}^{2}\left(-a,a\right)$ is a closed subspace of $L^{2}\left(-a,a\right)$,
and hence is also a Hilbert space. Given $T>0$, let us denote $C_{b}\left(\left[0,T\right]\right)$,
$C_{b}\left(\mathbb{R}_{+}\right)$ the spaces of bounded continuous
functions on the closed interval $\left[0,T\right]$ and the positive
half-line $\mathbb{R}_{+}:=\left[0,\infty\right)$. 

\begin{assm}[Parity and real-valuedness]\label{assm:K_par}
Suppose that $K\left(x\right)$ is a real-valued even function on $\left(-a,a\right)$:
$$K\left(-x\right)=K\left(x\right)\in\mathbb{R},\hspace{1em} x\in\left(-a,a\right).$$
\end{assm}

\begin{assm}[Hilbert-Schmidt regularity]\label{assm:K_reg}
Suppose that $K\left(x\right)$ is sufficiently regular, namely, we assume validity of the following integrability condition:
$$\int_{-a}^a\int_{-a}^a\left[K\left(x-\xi\right)\right]^2 \dd x \dd \xi< \infty.$$
\end{assm}

\begin{assm}[Positive semidefiniteness]\label{assm:K_pos}
We suppose that, for any function $f\in L^{2}\left(-a,a\right)$, the following condition holds true:
$$\int_{-a}^{a}f\left(x\right)\int_{-a}^{a}K\left(x-\xi\right)f\left(\xi\right)\dd\xi\dd x\geq0.$$
\end{assm}

Because of Assumption \ref{assm:K_reg}, the condition of Lemma \ref{lem:HS_kern}
is satisfied, and therefore the integral transformation given by
\begin{equation}
f\mapsto\mathcal{K}\left[f\right]\left(x\right):=\int_{-a}^{a}K\left(x-\xi\right)f\left(\xi\right)\dd\xi\label{eq:K_op_def}
\end{equation}
defines a compact linear operator $\mathcal{K}:$ $L^{2}\left(-a,a\right)\rightarrow L^{2}\left(-a,a\right)$.
Moreover, due to Assumption \ref{assm:K_par}, this operator is self-adjoint.
Consequently, by the spectral theory for compact self-adjoint operators
(Lemma \ref{lem:A_spec}), it follows that, for any $f\in L^{2}\left(-a,a\right)$,
we can write
\begin{equation}
f\left(x\right)=\sum_{n=1}^{\infty}f_{n}\varphi_{n}\left(x\right)+f_{\perp}\left(x\right),\label{eq:f_decomp1}
\end{equation}
where $f_{\perp}\left(x\right)\in\text{Ker }\mathcal{K}$ and $\varphi_{n}\left(x\right)$,
$n\geq1$, are normalised eigenfunctions of the operator $\mathcal{K}$
with non-zero eigenvalues. In other words, we have 
\begin{equation}
\int_{-a}^{a}K\left(x,\xi\right)f_{\perp}\left(\xi\right)\dd\xi=0,\hspace{1em}x\in\left(-a,a\right),\label{eq:K_null_sp}
\end{equation}
and
\begin{equation}
\int_{-a}^{a}K\left(x,\xi\right)\varphi_{n}\left(\xi\right)\dd\xi=\lambda_{n}\varphi_{n}\left(x\right),\hspace{1em}x\in\left(-a,a\right),\label{eq:K_spectral_pbm}
\end{equation}
with some $\lambda_{n}\neq0$, $n\geq1$, which are eigenvalues (in
general, repeated sequentially according to their multiplicity).

Setting
\begin{equation}
K_{2}\left(x,\xi\right):=K\left(x-\xi\right)-\frac{1}{2a}K_{1}\left(x\right)-\frac{1}{2a}K_{1}\left(\xi\right),\hspace{1em}\hspace{1em}K_{1}\left(x\right):=\int_{-a}^{a}K\left(\zeta-x\right)\dd\zeta,\label{eq:K2K1_def}
\end{equation}
we observe that
\begin{equation}
\int_{-a}^{a}K_{2}\left(x,\xi\right)\dd x=-\frac{1}{2a}\int_{-a}^{a}K_{1}\left(x\right)\dd x=:c_{0},\label{eq:c0_def}
\end{equation}
The right-hand side of (\ref{eq:c0_def}) is constant (independent
of $\xi$), and hence, we obtain a relation
\begin{equation}
\int_{-a}^{a}\int_{-a}^{a}K_{2}\left(x,\xi\right)\dd\xi\dd x=-\int_{-a}^{a}K_{1}\left(x\right)\dd x.\label{eq:K2_K1_rel}
\end{equation}
Moreover, for any $f\in L_{0}^{2}\left(-a,a\right)$, we have

\[
\int_{-a}^{a}\int_{-a}^{a}K_{2}\left(x,\xi\right)f\left(\xi\right)\dd\xi\dd x=0.
\]
Therefore, using Lemma \ref{lem:HS_kern}, we can define $\mathcal{K}_{2}:$
$L_{0}^{2}\left(-a,a\right)\rightarrow L_{0}^{2}\left(-a,a\right)$,
a compact linear operator that corresponds to the integral transformation
with the kernel function $K_{2}\left(x,\xi\right)$:
\begin{equation}
f\mapsto\mathcal{K}_{2}\left[f\right]\left(x\right):=\int_{-a}^{a}K_{2}\left(x,\xi\right)f\left(\xi\right)\dd\xi.\label{eq:K2_op_def}
\end{equation}
By the symmetry of $K_{2}\left(x,\xi\right)$, the operator $\mathcal{K}_{2}$
is self-adjoint and hence, using Lemma \ref{lem:A_spec}, we deduce
that, for any $f\in L_{0}^{2}\left(-a,a\right)$, we can write
\begin{equation}
f\left(x\right)=\sum_{n=1}^{\infty}\widetilde{f}_{n}\phi_{n}\left(x\right)+\widetilde{f}_{\perp}\left(x\right),\label{eq:f_decomp2}
\end{equation}
where $\widetilde{f}_{n}:=\left\langle f,\phi_{n}\right\rangle $,
$\widetilde{f}_{\perp}:=f-\sum_{n=1}^{\infty}\widetilde{f}_{n}\phi_{n}\in\text{Ker }\mathcal{K}_{2}$,
and $\phi_{n}$, $n\geq1$, are normalised (i.e. $\left\Vert \phi_{n}\right\Vert =1$,
$n\geq1$) eigenfunctions of the operator $\mathcal{K}_{2}$ with
non-zero eigenvalues. That is, we have 
\begin{equation}
\int_{-a}^{a}K_{2}\left(x,\xi\right)\widetilde{f}_{\perp}\left(\xi\right)\dd\xi=0,\hspace{1em}x\in\left(-a,a\right),\label{eq:K2_null_sp}
\end{equation}
and
\begin{equation}
\int_{-a}^{a}K_{2}\left(x,\xi\right)\phi_{n}\left(\xi\right)\dd\xi=\sigma_{n}\phi_{n}\left(x\right),\hspace{1em}x\in\left(-a,a\right),\label{eq:K2_spectral_pbm}
\end{equation}
with some $\sigma_{n}\neq0$, $n\geq1$, which are eigenvalues, repeated
consequently if multiple.

Note that $\varphi_{k}\perp\varphi_{n}$ (meaning that $\left\langle \varphi_{k},\varphi_{n}\right\rangle =0$)
for any $k\neq n\geq1$, and similarly, $\phi_{k}\perp\phi_{n}$,
$k\neq n\geq1$. This orthogonality is automatic when eigenfunctions
correspond to different eigenvalues. When belonging to the same eigensubspace,
we assume that they have already been orthogonalised (e.g. by the
Gram-Schmidt procedure). Moreover, since each $\phi_{n}\in L_{0}^{2}\left(-a,a\right)$,
we have
\begin{equation}
\left\langle 1,\phi_{n}\right\rangle =\int_{-a}^{a}\phi_{n}\left(x\right)\dd x=0,\hspace{1em}n\geq1.\label{eq:phi_n_zeromean}
\end{equation}

We shall now show deduce more information about eigenvalues of the
operators $\mathcal{K}$ and $\mathcal{K}_{2}$.
\begin{prop}
\label{prop:K2_spec_pos} The operator $\mathcal{K}_{2}:\,L_{0}^{2}\left(-a,a\right)\rightarrow L_{0}^{2}\left(-a,a\right)$
defined by (\ref{eq:K2_op_def}) is positive semidefinite, and thus
$\sigma_{n}\geq0$.
\end{prop}

\begin{proof}
First of all, note that according to (\ref{eq:K2K1_def}), for any
$f\in L_{0}^{2}\left(-a,a\right)$, we have
\begin{equation}
\left\langle \mathcal{K}\left[f\right]-\mathcal{K}_{2}\left[f\right],f\right\rangle =0.\label{eq:K2_min_K}
\end{equation}

Let us assume that there is at least one eigenvalue $\sigma_{1}<0$
(if there are few, we assume that $\sigma_{1}$ is the largest in
absolute value) and we shall derive a contradiction.

Observe that, by the variational characterisation of the smallest
negative eigenvalue (Rayleigh principle) given by Lemma \ref{lem:Rayl-Hilb},
we have 
\begin{equation}
\sigma_{1}=\min_{f\in L_{0}^{2}(-a,a)}\frac{\left<\mathcal{K}_{2}[f],f\right>}{\left\Vert f\right\Vert ^{2}}=\min_{f\in L_{0}^{2}(-a,a)}\frac{\left<\mathcal{K}[f],f\right>}{\left\Vert f\right\Vert ^{2}},\label{eq:lmb1_neg}
\end{equation}
where the second equality is due to (\ref{eq:K2_min_K}).

By positive semidefiniteness of $\mathcal{K}$, we have $\left<\mathcal{K}[f],f\right>\geq0$,
and hence equation (\ref{eq:lmb1_neg}) shows that $\sigma_{1}\geq0$
which contradicts the assumption $\sigma_{1}<0$. Therefore, all eigenvalues
$\sigma_{n}$, $n\geq1$, are non-negative, and hence $\mathcal{K}_{2}$
is a positive semidefinite operator, that is, we have 
\[
\left<\mathcal{K}_{2}[f],f\right>\geq0
\]
for any $f\in L_{0}^{2}(-a,a)$.
\end{proof}
\begin{prop}
\label{prop:K2_K_link} Eigenvalues $\lambda_{n}$ and $\sigma_{n}$,
$n\geq1$, defined by (\ref{eq:K_spectral_pbm}) and (\ref{eq:K2_spectral_pbm}),
respectively, satisfy the following inequalities:
\begin{equation}
\lambda_{n+1}\leq\sigma_{n}\leq\lambda_{n},\hspace{1em}\hspace{1em}n\geq1.\label{eq:sigm_bnds}
\end{equation}
\end{prop}

\begin{proof}
By Proposition \ref{prop:K2_spec_pos}, we have $\sigma_{n}\geq0$,
$n\geq1$. Furthermore, the Weyl-Courant-Fischer min-max principle
for characterisation of positive eigenvalues (Lemma \ref{lem:Cour-Fisch})
gives, for $n\geq2$, 
\begin{align}
\sigma_{n}= & \min_{q_{1},\ldots,q_{n-1}\in L^{2}(-a,a)}\max_{\substack{f\in L_{0}^{2}(-a,a)\\
f\perp\left(q_{1},\ldots,q_{n-1}\right)
}
}\frac{\left<\mathcal{K}_{2}[f],f\right>}{\left\Vert f\right\Vert ^{2}}\leq\max_{\substack{f\in L_{0}^{2}(-a,a)\\
f\perp\left(\varphi_{1},\ldots,\varphi_{n-1}\right)
}
}\frac{\left<\mathcal{K}_{2}[f],f\right>}{\left\Vert f\right\Vert ^{2}}\label{eq:sigm_n_ubnd}\\
= & \max_{\substack{f\in L_{0}^{2}(-a,a)\\
f\perp\left(\varphi_{1},\ldots,\varphi_{n-1}\right)
}
}\frac{\left<\mathcal{K}[f],f\right>}{\left\Vert f\right\Vert ^{2}}\leq\max_{\substack{f\in L^{2}(-a,a)\\
f\perp\left(\varphi_{1},\ldots,\varphi_{n-1}\right)
}
}\frac{\left<\mathcal{K}[f],f\right>}{\left\Vert f\right\Vert ^{2}}=\lambda_{n},\nonumber 
\end{align}
where we used (\ref{eq:K2_min_K}) and, in the last equality on the
second line, we employed the Rayleigh's variational characterisation
for positive eigenvalues (Lemma \ref{lem:Rayl-Hilb}). For $n=1$,
the situation is even simpler: 
\begin{align}
\sigma_{1}= & \max_{f\in L_{0}^{2}(-a,a)}\frac{\left<\mathcal{K}_{2}[f],f\right>}{\left\Vert f\right\Vert ^{2}}=\max_{f\in L_{0}^{2}(-a,a)}\frac{\left<\mathcal{K}[f],f\right>}{\left\Vert f\right\Vert ^{2}}\leq\max_{f\in L^{2}(-a,a)}\frac{\left<\mathcal{K}[f],f\right>}{\left\Vert f\right\Vert ^{2}}=\lambda_{1}.\label{eq:sigm_1_ubnd}
\end{align}
We can also obtain the lower bound estimate for eigenvalues $\sigma_{n}$.
To this effect, we apply the Weyl-Courant-Fischer min-max principle
twice: for $\lambda_{n}$ and $\sigma_{n-1}$. Namely, we have, for
$n\geq3$, 
\begin{align}
\lambda_{n}= & \min_{q_{1},\ldots,q_{n-1}\in L^{2}(-a,a)}\max_{\substack{f\in L^{2}(-a,a)\\
f\perp\left(q_{1},\ldots,q_{n-1}\right)
}
}\frac{\left<\mathcal{K}[f],f\right>}{\left\Vert f\right\Vert ^{2}}\leq\min_{q_{1},\ldots,q_{n-2}\in L^{2}(-a,a)}\max_{\substack{f\in L^{2}(-a,a)\\
f\perp\left(q_{1},\ldots,q_{n-2},1\right)
}
}\frac{\left<\mathcal{K}[f],f\right>}{\left\Vert f\right\Vert ^{2}}\label{eq:sigm_n_m1_lbnd}\\
= & \min_{q_{1},\ldots,q_{n-2}\in L^{2}(-a,a)}\max_{\substack{f\in L_{0}^{2}(-a,a)\\
f\perp\left(q_{1},\ldots,q_{n-2}\right)
}
}\frac{\left<\mathcal{K}[f],f\right>}{\left\Vert f\right\Vert ^{2}}=\min_{q_{1},\ldots,q_{n-2}\in L^{2}(-a,a)}\max_{\substack{f\in L_{0}^{2}(-a,a)\\
f\perp\left(q_{1},\ldots,q_{n-2}\right)
}
}\frac{\left<\mathcal{K}_{2}[f],f\right>}{\left\Vert f\right\Vert ^{2}}=\sigma_{n-1}.\nonumber 
\end{align}
Similarly, we can also obtain: 
\begin{align}
\lambda_{2}=\min_{q_{1}\in L^{2}(-a,a)}\max_{\substack{f\in L^{2}(-a,a)\\
f\perp q_{1}
}
}\frac{\left<\mathcal{K}[f],f\right>}{\left\Vert f\right\Vert ^{2}}\leq\max_{f\in L_{0}^{2}(-a,a)}\frac{\left<\mathcal{K}[f],f\right>}{\left\Vert f\right\Vert ^{2}}=\max_{f\in L_{0}^{2}(-a,a)}\frac{\left<\mathcal{K}_{2}[f],f\right>}{\left\Vert f\right\Vert ^{2}}=\sigma_{1}.\label{eq:sigm_1_lbnd}
\end{align}
Combining the estimates obtained in (\ref{eq:sigm_n_ubnd})--(\ref{eq:sigm_1_lbnd}),
we arrive at (\ref{eq:sigm_bnds}).
\end{proof}

\subsection{\label{subsec:sol_exist} Solution existence, uniqueness and construction}

Let us set 
\begin{equation}
q\left(x,t\right):=p\left(x,t\right)-\frac{P\left(t\right)}{2a},\hspace{1em}x\in\left(-a,a\right),\hspace{1em}t\geq0,\label{eq:q_def}
\end{equation}
and observe that 
\begin{equation}
\int_{-a}^{a}q\left(x,t\right)\dd x=0,\hspace{1em}t\geq0.\label{eq:q_zeromean}
\end{equation}
Then, equation (\ref{eq:p_eq}) can be equivalently rewritten as
\begin{equation}
\eta q\left(x,t\right)+\mathcal{K}_{2}\left[q\right]\left(x,t\right)-\frac{\nu}{\mu^{1-1/\alpha}}\int_{0}^{t}\mathcal{E}_{\alpha}\left(\mu^{1/\alpha}\left(t-\tau\right)\right)q\left(x,\tau\right)\dd\tau=F\left(x,t\right),\hspace{1em}x\in\left(-a,a\right),\hspace{1em}t>0,\label{eq:q_eq}
\end{equation}
where
\[
\mathcal{K}_{2}\left[q\right]\left(x,t\right):=\int_{-a}^{a}\left[K\left(x-\xi\right)-\frac{1}{2a}K_{1}\left(\xi\right)-\frac{1}{2a}K_{1}\left(x\right)\right]q\left(\xi,t\right)\dd\xi,
\]
\begin{align}
F\left(x,t\right):= & \eta q\left(x,0\right)+\mathcal{K}_{2}\left[q\right]\left(x,0\right)-\frac{P\left(t\right)-P\left(0\right)}{2a}\left(K_{1}\left(x\right)+c_{0}\right),\label{eq:F_def}
\end{align}
and $c_{0}$ is a constant defined in (\ref{eq:c0_def}).

With this preparation, our main results concerning the solution of
the model can be formulated as two theorems below. In their statements
and proofs, we will, according to the aforementioned convention, employ
the simplified notation $\left\langle \cdot,\cdot\right\rangle $,
$\left\Vert \cdot\right\Vert $ to denote the inner product and the
norm in $L^{2}\left(-a,a\right)$, respectively.
\begin{thm}
\label{thm:main_eta_pos} Assume that $\mu\geq0$, $\eta$,
$\nu>0$, $\alpha\in	\left(0,2\right)$, $P\in C_{b}\left(\mathbb{R}_{+}\right)$ and $p\left(\cdot,0\right)\in L^{2}\left(-a,a\right)$.
Suppose that $K$ satisfying Assumptions \ref{assm:K_par}--\ref{assm:K_pos}
is such that the equation
\begin{equation}
\mathcal{K}_{2}\left[\phi\right]\left(x\right):=\int_{-a}^{a}\left[K\left(x-\xi\right)-\frac{1}{2a}\int_{-a}^{a}\left(K\left(\zeta-x\right)+K\left(\zeta-\xi\right)\right)\dd\zeta\right]\phi\left(\xi\right)\dd\xi=0,\hspace{1em}x\in\left(-a,a\right),\label{eq:K_kern_cond}
\end{equation}
has at most one non-zero solution $\phi_{\perp}\in L_{0}^{2}\left(-a,a\right)$
with $\left\Vert \phi_{\perp}\right\Vert =1$, i.e. $\text{dim}\left(\text{Ker }\mathcal{K}_{2}\right)\leq1$.
Then, equations (\ref{eq:p_eq}) and (\ref{eq:q_eq}) have unique
solutions in $C_{b}\left(\mathbb{R}_{+};L^{2}\left(-a,a\right)\right)$
and $C_{b}\left(\mathbb{R}_{+};L_{0}^{2}\left(-a,a\right)\right)$,
respectively. These solutions are given by
\begin{equation}
p\left(x,t\right)=\frac{P\left(t\right)}{2a}+\sum_{k=1}^{\infty}d_{k}\left(t\right)\phi_{k}\left(x\right)+d_{\perp}\left(t\right)\phi_{\perp}\left(x\right),\hspace{1em}x\in\left(-a,a\right),\hspace{1em}t\geq0,\label{eq:p_sol}
\end{equation}
\begin{equation}
q\left(x,t\right)=\sum_{k=1}^{\infty}d_{k}\left(t\right)\phi_{k}\left(x\right)+d_{\perp}\left(t\right)\phi_{\perp}\left(x\right),\hspace{1em}x\in\left(-a,a\right),\hspace{1em}t\geq0,\label{eq:q_sol}
\end{equation}
where $\phi_{k}$, $\sigma_{k}$, $k\geq1$, are eigenfunctions and
eigenvalues of the operator $\mathcal{K}_{2}$ defined as in (\ref{eq:K2_spectral_pbm}),
and
\begin{align}
d_{k}\left(t\right):= & d_{k}^{0}\left[1+\frac{\nu}{\mu\left(\eta+\sigma_{k}\right)+\nu}\left(E_{\alpha}\left(-\left(\mu+\frac{\nu}{\eta+\sigma_{k}}\right)t^{\alpha}\right)-1\right)\right]-\frac{l_{k}}{2a\left(\eta+\sigma_{k}\right)}\left(P\left(t\right)-P\left(0\right)\right)\label{eq:d_k_final}\\
 & -\frac{\nu l_{k}}{2a\left(\eta+\sigma_{k}\right)^{2}}\left(\mu+\frac{\nu}{\eta+\sigma_{k}}\right)^{1/\alpha-1}\int_{0}^{t}\mathcal{E}_{\alpha}\left(\left(\mu+\frac{\nu}{\eta+\sigma_{k}}\right)^{1/\alpha}\left(t-\tau\right)\right)\left[P\left(\tau\right)-P\left(0\right)\right]\dd\tau,\hspace{1em}k\geq1,\nonumber 
\end{align}
\begin{align}
d_{\perp}\left(t\right):= & d_{\perp}^{0}\left[1+\frac{\nu}{\mu\eta+\nu}\left(E_{\alpha}\left(-\left(\mu+\frac{\nu}{\eta}\right)t^{\alpha}\right)-1\right)\right]-\frac{l_{\perp}}{2a\eta}\left(P\left(t\right)-P\left(0\right)\right)\label{eq:d_perp_final}\\
 & -\frac{\nu l_{\perp}}{2a\eta^{2}}\left(\mu+\frac{\nu}{\eta}\right)^{1/\alpha-1}\int_{0}^{t}\mathcal{E}_{\alpha}\left(\left(\mu+\frac{\nu}{\eta}\right)^{1/\alpha}\left(t-\tau\right)\right)\left[P\left(\tau\right)-P\left(0\right)\right]\dd\tau,\nonumber 
\end{align}
with 
\begin{equation}
d_{\perp}^{0}:=\left\langle p\left(\cdot,0\right),\phi_{\perp}\right\rangle ,\hspace{1em}l_{\perp}:=\left\langle K_{1},\phi_{\perp}\right\rangle ,\hspace{1em}d_{k}^{0}:=\left\langle p\left(\cdot,0\right),\phi_{k}\right\rangle ,\hspace{1em}l_{k}:=\left\langle K_{1},\phi_{k}\right\rangle ,\hspace{1em}k\geq1.\label{eq:d_k_0_l_k_def}
\end{equation}
\end{thm}

\begin{proof}
Since equations (\ref{eq:p_eq}) and (\ref{eq:q_eq}) are equivalent
and simply related, we only consider the latter.

The assumption $p\left(\cdot,0\right)\in L^{2}\left(-a,a\right)$
and (\ref{eq:q_def})--(\ref{eq:q_zeromean}) imply that $q\left(\cdot,0\right)\in L_{0}^{2}\left(-a,a\right)$.
Therefore, following (\ref{eq:f_decomp2}) and using that $\text{dim}\left(\text{Ker }\mathcal{K}_{2}\right)\leq1$, we can write
\begin{equation}
q\left(x,0\right):=\sum_{k=1}^{\infty}d_{k}^{0}\phi_{k}\left(x\right)+d_{\perp}^{0}\phi_\perp\left(x\right),\hspace{1em}x\in\left(-a,a\right),\label{eq:q0_decomp}
\end{equation}
with $d_{k}^{0}:=\left\langle q\left(\cdot,0\right),\phi_{k}\right\rangle$,
$k\geq1$, and $d_{\perp}^{0}:=\left\langle q\left(\cdot,0\right),\phi_{\perp}\right\rangle$. Note that these definitions coincide with
those in (\ref{eq:d_k_0_l_k_def}) due to the mean-zero property (\ref{eq:phi_n_zeromean}).

For an arbitrary integer $m\geq1$, let us define 
\begin{equation}
q_{m}\left(x,t\right):=\sum_{k=1}^{m}d_{k}\left(t\right)\phi_{k}\left(x\right)+d_{\perp}\left(t\right)\phi_{\perp}\left(x\right),\label{eq:q_m_def}
\end{equation}
where the functions $d_{k}$, $d_{\perp}$ solve the following integral
equations
\begin{equation}
\left(\eta+\sigma_{k}\right)d_{k}\left(t\right)-\frac{\nu}{\mu^{1-1/\alpha}}\int_{0}^{t}\mathcal{E}_{\alpha}\left(\mu^{1/\alpha}\left(t-\tau\right)\right)d_{k}\left(\tau\right)\dd\tau=\left(\eta+\sigma_{k}\right)d_{k}^{0}-\frac{l_{k}}{2a}\left(P\left(t\right)-P\left(0\right)\right),\hspace{1em}t>0,\hspace{1em}k\geq1,\label{eq:d_k_eq}
\end{equation}
\begin{equation}
\eta d_{\perp}\left(t\right)-\frac{\nu}{\mu^{1-1/\alpha}}\int_{0}^{t}\mathcal{E}_{\alpha}\left(\mu^{1/\alpha}\left(t-\tau\right)\right)d_{\perp}\left(\tau\right)\dd\tau=\eta d_{\perp}^{0}-\frac{l_{\perp}}{2a}\left(P\left(t\right)-P\left(0\right)\right),\hspace{1em}t>0,\label{eq:d_perp_eq}
\end{equation}
with
\[
d_{k}\left(0\right)=\frac{1}{\eta+\sigma_{k}}\left\langle \eta q\left(\cdot,0\right)+\mathcal{K}_{2}\left[q\left(\cdot,0\right)\right],\phi_{k}\right\rangle =\left\langle q\left(\cdot,0\right),\phi_{k}\right\rangle =d_{k}^{0},\hspace{1em}l_{k}:=\left\langle K_{1}+c_{0},\phi_{k}\right\rangle =\left\langle K_{1},\phi_{k}\right\rangle ,\hspace{1em}k\geq1,
\]
\[
d_{\perp}\left(0\right)=d_{\perp}^0,\hspace{1em}l_{\perp}:=\left\langle K_{1}+c_{0},\phi_{\perp}\right\rangle =\left\langle K_{1},\phi_{\perp}\right\rangle .
\]
Here, we performed some simplifications of the above expressions due
to decomposition (\ref{eq:q0_decomp}) and used (\ref{eq:K2_null_sp})--(\ref{eq:K2_spectral_pbm})
as well as (\ref{eq:phi_n_zeromean}).

We observe that equations (\ref{eq:d_k_eq})--(\ref{eq:d_perp_eq})
can be solved in a closed form by application of Corollary \ref{cor:spec_int_eq_2nd}.
Using (\ref{eq:int_cal_E_a}), this yields 
\begin{align}
d_{k}\left(t\right)= & d_{k}^{0}-\frac{l_{k}}{2a\left(\eta+\sigma_{k}\right)}\left(P\left(t\right)-P\left(0\right)\right)\label{eq:d_k_sol}\\
 & +\frac{\nu}{\eta+\sigma_{k}}\left(\mu+\frac{\nu}{\eta+\sigma_{k}}\right)^{1/\alpha-1}\left[\left(d_{k}^{0}+\frac{P\left(0\right)}{2a}\frac{l_{k}}{\eta+\sigma_{k}}\right)\int_{0}^{t}\mathcal{E}_{\alpha}\left(\left(\mu+\frac{\nu}{\eta+\sigma_{k}}\right)^{1/\alpha}\left(t-\tau\right)\right)\dd\tau\right.\nonumber \\
 & \left.-\frac{l_{k}}{2a\left(\eta+\sigma_{k}\right)}\int_{0}^{t}\mathcal{E}_{\alpha}\left(\left(\mu+\frac{\nu}{\eta+\sigma_{k}}\right)^{1/\alpha}\left(t-\tau\right)\right)P\left(\tau\right)\dd\tau\right],\hspace{1em}t\geq0,\hspace{1em}k\geq1,\nonumber 
\end{align}
\begin{align}
d_{\perp}\left(t\right)= & d_{\perp}^{0}-\frac{l_{\perp}}{2a\eta}\left(P\left(t\right)-P\left(0\right)\right)\label{eq:d_perp_sol}\\
 & +\frac{\nu}{\eta}\left(\mu+\frac{\nu}{\eta}\right)^{1/\alpha-1}\left[\left(d_{\perp}^{0}+\frac{P\left(0\right)}{2a}\frac{l_{\perp}}{\eta}\right)\int_{0}^{t}\mathcal{E}_{\alpha}\left(\left(\mu+\frac{\nu}{\eta}\right)^{1/\alpha}\left(t-\tau\right)\right)\dd\tau\right.\nonumber \\
 & \left.-\frac{l_{\perp}}{2a\eta}\int_{0}^{t}\mathcal{E}_{\alpha}\left(\left(\mu+\frac{\nu}{\eta}\right)^{1/\alpha}\left(t-\tau\right)\right)P\left(\tau\right)\dd\tau\right],\hspace{1em}t\geq0,\nonumber 
\end{align}
which can alternatively be rewritten as (\ref{eq:d_k_final})--(\ref{eq:d_perp_final}),
respectively.

Statement of the present theorem is essentially tantamount to showing
that 
\begin{equation}
\sup_{t>0}\left\Vert q\left(\cdot,t\right)-q_{m}\left(\cdot,t\right)\right\Vert \underset{m\rightarrow\infty}{\longrightarrow}0.\label{eq:q_m_conv}
\end{equation}
To prove the convergence (\ref{eq:q_m_conv}), we are going to show
that $\left(q_{m}\right)_{m=1}^{\infty}$, with $q_{m}$ given by
(\ref{eq:q_m_def}), is a Cauchy sequence in $C_{b}\left(\mathbb{R}_{+};L_{0}^{2}\left(-a,a\right)\right)$.
To this effect, let us set
\begin{equation}
S_{mn}\left(x,t\right):=q_{n}\left(x,t\right)-q_{m}\left(x,t\right)=\sum_{k=m+1}^{n}d_{k}\left(t\right)\phi_{k}\left(x\right),\hspace{1em}x\in\left(-a,a\right),\hspace{1em}t\in\mathbb{R}_{+},\hspace{1em}n>m,\label{eq:S_mn_def}
\end{equation}
and use the orthonormality of $\phi_{k}$ to write
\[
\left\Vert S_{mn}\left(\cdot,t\right)\right\Vert =\left(\sum_{k=m+1}^{n}\left|d_{k}\left(t\right)\right|^{2}\right)^{1/2}=\left\Vert \left(d_{k}\left(t\right)\right)_{k=m+1}^{n}\right\Vert _{l^{2}},
\]
where in the last expression we used $\left\Vert \cdot\right\Vert _{l^{2}}$
for designating the Euclidean vector norm. Taking into account that
$\sigma_{k}>0$ (recall Proposition \ref{prop:K2_spec_pos}), we can
estimate 
\[
\frac{\nu}{\mu\left(\eta+\sigma_{k}\right)+\nu}\left|E_{\alpha}\left(-\left(\mu+\frac{\nu}{\eta+\sigma_{k}}\right)t^{\alpha}\right)-1\right|\leq\frac{\nu}{\mu\eta+\nu}\left(C_{1}+1\right),
\]
where $C_{1}:=\sup_{\tau>0}\left|E_{\alpha}\left(-\tau\right)\right|$
is finite since the function $E_{\alpha}\left(t\right)$ is continuous
and decaying for large negative values of the argument $t$ (see (\ref{eq:E_a}),
\eqref{eq:E_1}, (\ref{eq:E_a_small}) and (\ref{eq:E_a_large})). 
Also, due to the
continuity and boundedness of $P\left(t\right)$ on $\mathbb{R}_{+}$,
it is immediate to see that
\[
\frac{1}{2a\left(\eta+\sigma_{k}\right)}\left|P\left(t\right)-P\left(0\right)\right|\leq\frac{1}{a\eta}\left\Vert P\right\Vert _{L^{\infty}\left(\mathbb{R}_{+}\right)},\hspace{1em}t\in\mathbb{R}_{+}.
\]
Let us deal with the term on the second line of \eqref{eq:d_k_final}. To this end, we have, for any $t_0>0$,
\begin{align*}
\left(\mu+\frac{\nu}{\eta+\sigma_{k}}\right)^{1/\alpha-1}\left|\int_{0}^{t}\mathcal{E}_{\alpha}\left(\left(\mu+\frac{\nu}{\eta+\sigma_{k}}\right)^{1/\alpha}\left(t-\tau\right)\right)\left[P\left(\tau\right)-P\left(0\right)\right]\dd\tau\right|\\
=\left(\mu+\frac{\nu}{\eta+\sigma_{k}}\right)^{1/\alpha-1}\left|\int_{0}^{t}\mathcal{E}_{\alpha}\left(\left(\mu+\frac{\nu}{\eta+\sigma_{k}}\right)^{1/\alpha}\tau\right)\left[P\left(t-\tau\right)-P\left(0\right)\right]\dd\tau\right|\leq\left|\int_{0}^{t_{0}}\ldots\right|+\left|\int_{t_{0}}^{t}\ldots\right|
\end{align*}
with
\begin{align*}
\left(\mu+\frac{\nu}{\eta+\sigma_{k}}\right)^{1/\alpha-1}\left|\int_{0}^{t_{0}}\mathcal{E}_{\alpha}\left(\left(\mu+\frac{\nu}{\eta+\sigma_{k}}\right)^{1/\alpha}\tau\right)\left[P\left(t-\tau\right)-P\left(0\right)\right]\dd\tau\right|\\ \leq2C_{1,t_{0}}\left\Vert P\right\Vert _{L^{\infty}\left(\mathbb{R}_{+}\right)}\int_{0}^{t_{0}}\frac{\dd\tau}{\tau^{1-\alpha}}=\frac{2C_{1,t_{0}}\left\Vert P\right\Vert _{L^{\infty}\left(\mathbb{R}_{+}\right)}t_{0}^{\alpha}}{\alpha},
\end{align*}
and
\begin{align*}
\left(\mu+\frac{\nu}{\eta+\sigma_{k}}\right)^{1/\alpha-1}\left|\int_{t_0}^{t}\mathcal{E}_{\alpha}\left(\left(\mu+\frac{\nu}{\eta+\sigma_{k}}\right)^{1/\alpha}\tau\right)\left[P\left(t-\tau\right)-P\left(0\right)\right]\dd\tau\right|\\ \leq2C_{2,t_{0}}\left(\mu+\frac{\nu}{\eta+\sigma_{1}}\right)^{-2}\left\Vert P\right\Vert _{L^{\infty}\left(\mathbb{R}_{+}\right)}\int_{t_{0}}^{\infty}\frac{\dd\tau}{\tau^{1+\alpha}}=\frac{2C_{2,t_{0}}\left\Vert P\right\Vert _{L^{\infty}\left(\mathbb{R}_{+}\right)}}{\left(\mu+\nu/\left(\eta+\sigma_1\right)\right)^{2}\alpha t_{0}^{\alpha}}.
\end{align*}
Here, the constants $C_{1,t_{0}}:=\sup_{\tau\in\left(0,t_{0}\right)}\left|\tau^{1-\alpha}\mathcal{E}_{\alpha}\left(\tau\right)\right|$ and $C_{2,t_{0}}:=\sup_{\tau\in\left(t_{0},\infty\right)}\left|\tau^{1+\alpha}\mathcal{E}_{\alpha}\left(\tau\right)\right|$
are finite due to (\ref{eq:cal_E_a}), (\ref{eq:cal_E_a_small}) and \eqref{eq:cal_E_a_large} since $\alpha\in\left(0,2\right)$.
Consequently, applying the triangle inequality to (\ref{eq:d_k_sol})--(\ref{eq:d_perp_sol}), we estimate
\begin{align}\label{eq:S_mn_estim}
\left\Vert S_{mn}\left(\cdot,t\right)\right\Vert \leq&\left(1+\frac{\nu}{\mu\eta+\nu}\left(C_{1}+1\right)\right)\left\Vert \left(d_{k}^{0}\right)_{k=m+1}^{n}\right\Vert _{l^{2}}\\
&+\left(1+\frac{\nu C_{1,t_0}t_0^{\alpha}}{\eta\alpha}+\frac{\nu C_{2,t_0}}{\eta\left(\mu+\nu/\left(\eta+\sigma_1\right)\right)^{2}\alpha t_0^{\alpha}}\right)\frac{1}{a\eta}\left\Vert P\right\Vert _{L^{\infty}\left(\mathbb{R}_{+}\right)}\left\Vert \left(l_{k}\right)_{k=m+1}^{n}\right\Vert _{l^{2}},\hspace{1em}t\in\mathbb{R}_{+}.\nonumber
\end{align}
Now, recalling (\ref{eq:d_k_0_l_k_def}), we have, by the Bessel's
inequality, 
\[
\sum_{k=1}^{\infty}\left|d_{k}^{0}\right|^{2}\leq\left\Vert q\left(\cdot,0\right)\right\Vert ^{2}<\infty,\hspace{1em}\hspace{1em}\sum_{k=1}^{\infty}\left|l_{k}\right|^{2}\leq\left\Vert K_{1}\right\Vert ^{2}<\infty,
\]
i.e. both series $\sum_{k=1}^{\infty}\left|d_{k}^{0}\right|^{2}$
and $\sum_{k=1}^{\infty}\left|l_{k}\right|^{2}$ converge, and hence
the quantities $\left\Vert \left(d_{k}^{0}\right)_{k=m+1}^{n}\right\Vert _{l^{2}}$,
$\left\Vert \left(l_{k}\right)_{k=m+1}^{n}\right\Vert _{l^{2}}$ can
be made arbitrary small for large $m$. Therefore, we deduce from
(\ref{eq:S_mn_estim}) that $\left\Vert S_{mn}\left(\cdot,t\right)\right\Vert $
is guaranteed to be arbitrary small, uniformly for all $t\in\mathbb{R}_{+}$,
once sufficiently large value of $m$ is chosen. In other words, recalling
(\ref{eq:S_mn_def}), we have obtained that $\left(q_{m}\right)_{m=1}^{\infty}$
is a Cauchy sequence in $C_{b}\left(\mathbb{R}_{+};L_{0}^{2}\left(-a,a\right)\right)$.
Since $C_{b}\left(\mathbb{R}_{+};L_{0}^{2}\left(-a,a\right)\right)$
is a closed subspace of a Banach space (see e.g. \cite[Thm 6.28]{Hunt}
for the standard fact that $L^{\infty}\left(\mathbb{R}_{+};L^{2}\left(-a,a\right)\right)$
is a Banach space), it is also a Banach space, and hence complete.
This implies the desired convergence (\ref{eq:q_m_conv}).
\end{proof}
\begin{rem}
The condition $\text{dim}\left(\text{Ker }\text{\ensuremath{\mathcal{K}_{2}}}\right)\leq1$
imposed in the formulation of Theorem \ref{thm:main_eta_pos} is essential
for the uniqueness of the solution. Otherwise, one could, without changing
the validity of (\ref{eq:p_eq}), (\ref{eq:q_eq}), add to the solution
(\ref{eq:p_sol})--(\ref{eq:q_sol}) the term $\widetilde{d}_{\perp}\left(t\right)\widetilde{\phi}_{\perp}\left(x\right)$
with any $\widetilde{\phi}_{\perp}\in\text{Ker }\text{\ensuremath{\mathcal{K}_{2}}}$,
$\widetilde{\phi}_{\perp}\perp q_{\perp}^{0}$, and $\widetilde{d}_{\perp}$
solving the integral equation
\[
\eta\widetilde{d}_{\perp}\left(t\right)-\frac{\nu}{\mu^{1-1/\alpha}}\int_{0}^{t}\mathcal{E}_{\alpha}\left(\mu^{1/\alpha}\left(t-\tau\right)\right)\widetilde{d}_{\perp}\left(\tau\right)\dd\tau=-\frac{l_{\perp}}{2a}\left(P\left(t\right)-P\left(0\right)\right),\hspace{1em}t>0.
\]
If, on the other hand, $K$ is such that (\ref{eq:K_kern_cond}) admits
only the zero solution, i.e. $\text{Ker }\text{\ensuremath{\mathcal{K}_{2}}}=0$,
then $\phi_{\perp}\equiv0$.
\end{rem}

\begin{cor}
\label{cor:eta_pos} Assume that $\mu\geq0$, $\alpha$, $\eta$,
$\nu>0$, $P\in C_{b}\left(\mathbb{R}_{+}\right)$ and $\Delta\in L^{2}\left(-a,a\right)$.
Then, the problem given by (\ref{eq:p_delta_eq})--(\ref{eq:p0_int_eq}),
with $K$ satisfying Assumptions \ref{assm:K_par}--\ref{assm:K_pos}
and such that $\text{dim}\left(\text{Ker }\mathcal{K}_{2}\right)\leq1$,
and $w\left[p\right]$ defined as in (\ref{eq:w_p_term}), is uniquely
solvable and the solution $p\in C_{b}\left(\mathbb{R}_{+},L^{2}\left(-a,a\right)\right)$
is furnished by (\ref{eq:p_sol}).
\end{cor}

\begin{proof}
By the standard Fredholm theory applied to the positive compact self-adjoint
operator $\mathcal{K}$ defined in (\ref{eq:K_op_def}), the assumptions
$\Delta\in L^{2}\left(-a,a\right)$ and $\eta>0$ entail the existence
of a unique $p\left(\cdot,0\right)\in L^{2}\left(-a,a\right)$ satisfying
(\ref{eq:p0_int_eq}). This makes Theorem \ref{thm:main_eta_pos}
applicable yielding the claim of this Corollary. 
\end{proof}
\begin{thm}
\label{thm:main_eta_0} Assume that $\mu\geq0$, $\alpha\in\left(0,2\right)$,
$\eta=0$, $\nu>0$, $P\left(t\right)\equiv P\left(0\right)=:P_{0}$
for $t\in\left[0,T\right]$ with some $T>0$, and $p\left(\cdot,0\right)\in L^{2}\left(-a,a\right)$.
Suppose that $K$ satisfies Assumptions \ref{assm:K_par}--\ref{assm:K_pos}
and $q\left(\cdot,0\right):=p\left(\cdot,0\right)-\frac{P_{0}}{2a}$
is orthogonal to any $\psi\in\text{Ker }\mathcal{K}_{2}$. Then, equations
(\ref{eq:p_eq}) and (\ref{eq:q_eq}) have unique solutions in $C\left(\left[0,T\right];L^{2}\left(-a,a\right)\right)$
and $C\left(\left[0,T\right];L_{0}^{2}\left(-a,a\right)\right)$,
respectively. These solutions are given by 
\begin{equation}
p\left(x,t\right)=\frac{P_{0}}{2a}+\sum_{k=1}^{\infty}d_{k}\left(t\right)\phi_{k}\left(x\right),\hspace{1em}q\left(x,t\right)=\sum_{k=1}^{\infty}d_{k}\left(t\right)\phi_{k}\left(x\right),\hspace{1em}x\in\left(-a,a\right),\hspace{1em}t\in\left(0,T\right),\label{eq:p_q_sol_eta_0}
\end{equation}
where $\phi_{k}$, $\sigma_{k}$, $k\geq1$, are eigenfunctions and
eigenvalues of the operator $\mathcal{K}_{2}$ defined as in (\ref{eq:K2_spectral_pbm}),
and
\begin{align}
d_{k}\left(t\right):= & d_{k}^{0}\left[1+\frac{\nu}{\mu\sigma_{k}+\nu}\left(E_{\alpha}\left(-\left(\mu+\frac{\nu}{\sigma_{k}}\right)t^{\alpha}\right)-1\right)\right],\hspace{1em}k\geq1,\label{eq:d_k_final_eta_0}
\end{align}
with $d_{k}^{0}$, $k\geq1$, as in (\ref{eq:d_k_0_l_k_def}).
\end{thm}

\begin{proof}
The proof generally goes along the same lines as that for Theorem
\ref{thm:main_eta_pos}. Similarly to (\ref{eq:q_m_def}), we introduce
\begin{equation}
q_{m}\left(x,t\right):=\sum_{k=1}^{m}d_{k}\left(t\right)\phi_{k}\left(x\right)+d_{\perp}\left(t\right)\widetilde{\phi}_{\perp}\left(x\right),\hspace{1em}m\ge1,\label{eq:q_m_def_eta_0}
\end{equation}
where $\widetilde{\phi}_{\perp}\in L_{0}^{2}\left(-a,a\right)$ denotes
any solution of the equation (\ref{eq:K_kern_cond}).

By substitution of (\ref{eq:q_m_def_eta_0}) in (\ref{eq:q_eq}) and
using orthogonality of $\phi_{k}$, $k\geq1$, and $\widetilde{\phi}_{\perp}$,
we deduce that $d_{k}$, $d_{\perp}$ are the functions solving the
following integral equations
\begin{equation}
\sigma_{k}d_{k}\left(t\right)-\frac{\nu}{\mu^{1-1/\alpha}}\int_{0}^{t}\mathcal{E}_{\alpha}\left(\mu^{1/\alpha}\left(t-\tau\right)\right)d_{k}\left(\tau\right)\dd\tau=\sigma_{k}d_{k}^{0},\hspace{1em}t\in\left(0,T\right),\hspace{1em}k\geq1,\label{eq:eta_0_d_k_eq}
\end{equation}
\begin{equation}
-\frac{\nu}{\mu^{1-1/\alpha}}\int_{0}^{t}\mathcal{E}_{\alpha}\left(\mu^{1/\alpha}\left(t-\tau\right)\right)d_{\perp}\left(\tau\right)\dd\tau=0,\hspace{1em}t\in\left(0,T\right).\label{eq:eta_0_d_perp_eq}
\end{equation}
We note immediately that application of Corollary \ref{cor:spec_int_eq_1st_mu}
or Lemma \ref{lem:spec_int_eq_1st_mu_alph12} to (\ref{eq:eta_0_d_perp_eq}),
depending on whether $\alpha\in\left(0,1\right)$ or $\alpha\in\left[1,2\right)$,
yields $d_{\perp}\left(t\right)\equiv0$, $t\in\left(0,T\right)$.
For the solution of (\ref{eq:eta_0_d_k_eq}), we invoke Corollary
\ref{cor:spec_int_eq_2nd} which gives
\begin{align}
d_{k}\left(t\right)= & d_{k}^{0}\left[1+\frac{\nu/\sigma_{k}}{\left(\mu+\nu/\sigma_{k}\right)^{1-1/\alpha}}\int_{0}^{t}\mathcal{E}_{\alpha}\left(\left(\mu+\frac{\nu}{\sigma_{k}}\right)^{1/\alpha}\tau\right)\dd\tau\right],\hspace{1em}t\in\left(0,T\right).\label{eq:d_k_eta_0_sol}
\end{align}
Making use of (\ref{eq:int_cal_E_a}), equation (\ref{eq:d_k_eta_0_sol})
simplifies into (\ref{eq:d_k_final_eta_0}). The proof of the convergence
of (\ref{eq:q_m_def_eta_0}) as $m\rightarrow\infty$ is identical
to that for Theorem \ref{thm:main_eta_pos} with an additional simplification
that there is now no necessity to deal with the $l_{k}$ terms. 

Finally, we note that the orthogonality condition $q\left(\cdot,0\right)\perp\psi$
for any $\psi\in\text{Ker }\mathcal{K}_{2}$ imposed in the formulation
of the Theorem is due to the requirement of the continuity of the
solution (\ref{eq:q_sol}). Indeed, if this condition was violated,
we would have
\[
q\left(x,0\right)=\sum_{k=1}^{\infty}d_{k}\left(t\right)\phi_{k}\left(x\right)+\psi_{0}\left(x\right),\hspace{1em}x\in\left(-a,a\right),
\]
for some $\psi_{0}\in\text{Ker }\mathcal{K}_{2}$. This would not
be consistent with the limit of (\ref{eq:q_sol}) as $t\rightarrow0^{+}$
(which itself is a consequence of our conclusion that $d_{\perp}\left(t\right)\equiv0$,
$t\in\left(0,T\right)$, in (\ref{eq:q_m_def_eta_0})).
\end{proof}
\begin{rem}
It is noteworthy that, in case $\eta=0$, the regularity requirement
$p\left(\cdot,0\right)\in L^{2}\left(-a,a\right)$ may be a rather
strong one, when viewed in the context of the entire problem (\ref{eq:p_delta_eq})--(\ref{eq:p0_int_eq}).
As we shall see on an example of a particular kernel function $K$
in Proposition \ref{prop:K0_sol_p_0}, this may restrict $\Delta$
in (\ref{eq:p0_int_eq}) and $P_{0}$. On the other hand, the same
example shows that the orthogonality condition in the statement of
Theorem \ref{thm:main_eta_0} may be trivially satisfied. 
\end{rem}

\subsection{\label{subsec:concr_K} A concrete form of the kernel function}

We now make results of the previous subsection more precise by focussing
on a concrete kernel function that is often used for contact mechanical
problems under consideration, namely, problems with a half-space geometry
or for an elastic layer of a large thickness (see e.g. \cite{AleksKov1,ArgChai1,Shtaer}).
Namely, we consider the kernel function given by
\begin{equation}
K_{0}\left(x\right):=-\log\left|x\right|+C_{K},\hspace{1em}x\in\left(-2a,2a\right),\label{eq:K0_def}
\end{equation}
where $C_{K}>\log a$ is a constant.

It is straightforward to see that $K=K_{0}$ satisfies Assumptions
\ref{assm:K_par}--\ref{assm:K_reg}. Verification of Assumption \ref{assm:K_pos} requires a change of variable $x=a\tilde{x}$, the additive property of logarithms and the positive-definiteness result of \cite{Reade1} valid for the integral operator with purely logarithmic kernel (i.e. $C_K=0$) on the interval $\left(-1,1\right)$.  Here, we have also made use of the assumption $C_K>\log a$.

We now claim that the eigenvalues of the corresponding $\mathcal{K}_{2}$ operator can be characterised as follows.
\begin{prop}
\label{prop:K0_sigma_estims} Let $K=K_{0}$ with $K_{0}$ defined
in (\ref{eq:K0_def}). Then, the eigenvalues of the operator $\mathcal{K}_{2}$
defined in (\ref{eq:K2_op_def}) satisfy the following estimates
\begin{equation}
0<\sigma_{1}\leq a\pi\log2+2a\left(C_{K}+\left|\log a\right|\right),\hspace{1em}\hspace{1em}0<\sigma_{n}=\mathcal{O}\left(\frac{1}{n}\right),\hspace{1em}n\gg1.\label{eq:sigma_1_sigma_n_asympt}
\end{equation}
\end{prop}

\begin{proof}
First, we claim that $\lambda_{n}$, the eigenvalues of the operator
$\mathcal{K}$, defined in (\ref{eq:K_op_def}), decrease to zero
as $\mathcal{O}\left(1/n\right)$ for large $n$. This follows from
the corresponding result of \cite{Reade1} obtained for the case of
purely logarithmic kernel (i.e. (\ref{eq:K0_def}) with $C_{K}=0$).
Indeed, since the asymptotic decrease of the eigenvalues of a positive
integral operator is related to the regularity of the kernel function
(see, in addition to \cite{Reade1}, also \cite{Reade2}), the presence
of an extra constant $C_{K}>0$ does not affect this asymptotic behaviour.
The final asymptotic result given in (\ref{eq:sigma_1_sigma_n_asympt})
now follows from $\lambda_{n}=\mathcal{O}\left(1/n\right)$, $n\gg1$,
by employing Proposition \ref{prop:K2_K_link}.

To deduce the upper bound for $\sigma_{1}$, we shall first get the
one for $\lambda_{1}$. To this effect, we transform (\ref{eq:K_spectral_pbm})
to an equivalent problem on the interval $\left(-1,1\right)$. Namely,
setting $\psi\left(x\right):=\phi\left(ax\right)$, we have
\begin{equation}
\int_{-1}^{1}\left[-\log\left|x-\xi\right|+\log a+C_{K}\right]\psi\left(\xi\right)\dd\xi=\frac{\lambda}{a}\psi\left(x\right),\hspace{1em}x\in\left(-1,1\right).\label{eq:K0_spectr_pbm_scaled}
\end{equation}
Rayleigh's variational characterisation for positive eigenvalues (Lemma
\ref{lem:Rayl-Hilb}) now yields
\begin{align*}
\frac{\lambda_{1}}{a} & \leq\max_{f\in L^{2}\left(-1,1\right)}\frac{\int_{-1}^{1}-\log\left|x-\xi\right|f\left(\xi\right)f\left(x\right)\dd\xi\dd x}{\left\Vert f\right\Vert _{L^{2}\left(-1,1\right)}}+\left(C_{K}+\left|\log a\right|\right)\max_{f\in L^{2}\left(-1,1\right)}\frac{\int_{-1}^{1}\left|f\left(\xi\right)\right|\left|f\left(x\right)\right|\dd\xi\dd x}{\left\Vert f\right\Vert _{L^{2}\left(-1,1\right)}}\\
 & \leq\pi\log2+2\left(C_{K}+\left|\log a\right|\right),
\end{align*}
where we used $\pi\log2$ as an upper bound for the first eigenvalue
of the logarithmic kernel due to \cite{Reade1}, and we used the Cauchy-Schwarz
inequality to trivially estimate the second quotient. Finally, using
Proposition \ref{prop:K2_K_link}, we obtain the bound for $\sigma_{1}$
in (\ref{eq:sigma_1_sigma_n_asympt}).
\end{proof}
When $\eta=0$, the particular form of the kernel function $K_{0}$
makes it possible to work with an explicit form of the initial data
$p\left(x,0\right)$. Indeed, we have the following constructive result.
\begin{prop}
\label{prop:K0_sol_p_0} Let $\eta=0$, $a\neq2$ and $K=K_{0}$ with
$K_{0}$ defined in (\ref{eq:K0_def}). Suppose that $\Delta\in C^{2}\left(\left[-a,a\right]\right)$.
Then, a solution of (\ref{eq:p0_int_eq}) subject to (\ref{eq:equil_eq})
with $t=0$ is given by
\begin{equation}
p\left(x,0\right)=\frac{1}{\left(a^{2}-x^{2}\right)^{1/2}}\left[\fint_{-a}^{a}\frac{\left(a^{2}-\xi^{2}\right)^{1/2}\Delta^{\prime}\left(\xi\right)}{\xi-x}\dd\xi+\frac{1}{\log\left(a/2\right)}\left(\int_{-a}^{a}\frac{\Delta\left(\xi\right)}{\left(a^{2}-\xi^{2}\right)^{1/2}}\dd\xi+\pi\left(C_{K}P\left(0\right)-\delta\left(0\right)\right)\right)\right],\label{eq:p_0_sol}
\end{equation}
and
\begin{equation}
\delta\left(0\right)=\frac{1}{\pi}\int_{-a}^{a}\frac{\Delta\left(\xi\right)}{\left(a^{2}-\xi^{2}\right)^{1/2}}\dd\xi-P\left(0\right)\left(\frac{1}{\pi^{2}}\log\left(\frac{a}{2}\right)-C_{K}\right)+\frac{1}{\pi^{2}}\log\left(\frac{a}{2}\right)\int_{-a}^{a}\frac{1}{\left(a^{2}-x^{2}\right)^{1/2}}\fint_{-a}^{a}\frac{\left(a^{2}-\xi^{2}\right)^{1/2}\Delta^{\prime}\left(\xi\right)}{\xi-x}\dd\xi\dd x.\label{eq:delta_0_sol}
\end{equation}
\end{prop}

\begin{proof}
Inserting (\ref{eq:K0_def}) and taking into account (\ref{eq:equil_eq})
with $t=0$, (\ref{eq:p0_int_eq}) rewrites as
\[
-\int_{-a}^{a}\log\left|x-\xi\right|p\left(\xi,0\right)\dd\xi=-C_{K}P\left(0\right)+\delta\left(0\right)-\Delta\left(x\right),\hspace{1em}x\in\left(-a,a\right).
\]
Application of Lemma \ref{lem:log_int_eq} using the elementary identity
\begin{equation}
\int_{-a}^{a}\frac{\dd x}{\left(a^{2}-x^{2}\right)^{1/2}}=\int_{-1}^{1}\frac{\dd x}{\left(1-x^{2}\right)^{1/2}}=\pi,\label{eq:elem_integr}
\end{equation}
now yields
\begin{align}
p\left(x,0\right)= & \frac{1}{\left(a^{2}-x^{2}\right)^{1/2}}\fint_{-a}^{a}\frac{\left(a^{2}-\xi^{2}\right)^{1/2}\Delta^{\prime}\left(\xi\right)}{\xi-x}\dd\xi\label{eq:p_0_sol_prelim}\\
 & +\frac{1}{\left(a^{2}-x^{2}\right)^{1/2}}\frac{1}{\log\left(a/2\right)}\left(\int_{-a}^{a}\frac{\Delta\left(\xi\right)}{\left(a^{2}-\xi^{2}\right)^{1/2}}\dd\xi+\pi\left(C_{K}P\left(0\right)-\delta\left(0\right)\right)\right),\nonumber 
\end{align}
where $\fint_{-a}^{a}$ stands for the Cauchy principal value integral.

Integrating (\ref{eq:p_0_sol_prelim}) on $\left(-a,a\right)$, we
reuse, in the left-hand side, (\ref{eq:equil_eq}) with $t=0$. Rearranging
the terms, we obtain (\ref{eq:delta_0_sol}).
\end{proof}
\begin{rem}
We see, from a form of the solution (\ref{eq:p_0_sol}), that, in
general, $p\left(\cdot,0\right)\notin L^{2}\left(-a,a\right)$. However,
the square-integrability condition can be achieved for some particular
profiles $\Delta$ and values $P\left(0\right)$, namely, those that
annihilate the square bracket in (\ref{eq:p_0_sol}) at $x=\pm a$,
see e.g. \cite[pp. 47--48]{Gal} for examples of finite pressure distributions
for quadratic and quartic symmetric shapes of $\Delta$. 
\end{rem}

Finally, by focussing on the concrete kernel function $K_{0}$, we
will show that the auxiliary condition $\text{Ker }\mathcal{K}_{2}\leq1$
in Theorem \ref{thm:main_eta_pos} and the orthogonality condition
in Theorem \ref{thm:main_eta_0} are not difficult to verify. 
\begin{prop}
\label{prop:K0_zero_kern} Assume that $a\neq2$. Let $K=K_{0}$ with
$K_{0}$ defined in (\ref{eq:K0_def}). Then, the only solution of
(\ref{eq:K_kern_cond}) in $L_{0}^{2}\left(-a,a\right)$ is $\phi\equiv0$.
\end{prop}

\begin{proof}
Since we look for the solution in $L_{0}^{2}\left(-a,a\right)$, we
use zero-mean condition (\ref{eq:q_zeromean}) to rewrite (\ref{eq:K_kern_cond})
as 
\[
\int_{-a}^{a}K_{0}\left(x-\xi\right)\phi\left(\xi\right)\dd\xi=\frac{1}{2a}\int_{-a}^{a}K_{0}\left(\zeta-\xi\right)\phi\left(\xi\right)\dd\zeta\dd\xi,\hspace{1em}x\in\left(-a,a\right),
\]
and, furthermore,
\begin{equation}
-\int_{-a}^{a}\log\left|x-\xi\right|\phi\left(\xi\right)\dd\xi=-\frac{1}{2a}\int_{-a}^{a}\int_{-a}^{a}\log\left|\zeta-\xi\right|\phi\left(\xi\right)\dd\zeta\dd\xi,\hspace{1em}x\in\left(-a,a\right).\label{eq:log_eq_RHS_const}
\end{equation}
Since the right-hand side of (\ref{eq:log_eq_RHS_const}) is just
a constant, application of Lemma \ref{lem:log_int_eq} yields a particularly
simple result
\[
\phi\left(x\right)=\frac{1}{\left(a^{2}-x^{2}\right)^{1/2}}\frac{\pi}{2a\log\left(a/2\right)}\int_{-a}^{a}\int_{-a}^{a}\log\left|\zeta-\xi\right|\phi\left(\xi\right)\dd\xi\dd\zeta,\hspace{1em}x\in\left(-a,a\right),
\]
where we used (\ref{eq:elem_integr}). Upon further integration over
$\left(-a,a\right)$ and use of (\ref{eq:q_zeromean}), we conclude
that we must have
\[
\int_{-a}^{a}\int_{-a}^{a}\log\left|\zeta-\xi\right|\phi\left(\xi\right)\dd\zeta\dd\xi=0.
\]
Getting back to (\ref{eq:log_eq_RHS_const}), we see that this condition
entails that
\[
\int_{-a}^{a}\log\left|x-\xi\right|\phi\left(\xi\right)\dd\xi=0,\hspace{1em}x\in\left(-a,a\right).
\]
Hence, by applying Lemma \ref{lem:log_int_eq} again, we conclude
that $\phi\equiv0$.
\end{proof}

\section{Analysis of the solution \label{sec:sol_analys}}

We are going to show that the solution to the proposed model has features
reflecting the expected physical behaviour. In particular, consider
two settings: a constant load and a transitional load (i.e. the one
which stabilises to a constant value after a finite time). We show
that, in both cases, for large times, the solution stabilises to a
stationary pressure distribution that can be found explicitly. Moreover,
in some cases, this pressure distribution is simply constant (uniform),
an aspect which is consistent with previous models \cite{ArgChai2,ArgFad}
but is certainly not a general feature (see e.g. \cite[Ch. 6]{Goryach}).

\subsection{Constant load}

First, let us consider the most commonly investigated case of a constant
load, i.e. where $P\left(t\right)\equiv P\left(0\right)=:P_{0}$ for
$t\geq0$. 
\begin{prop}
\label{prop:stat_regime1} Assume that $\mu$, $\eta\geq0$, $\alpha\in\left(0,2\right)$,
$\nu>0$, $P\left(t\right)\equiv P\left(0\right)=:P_{0}$, $t\geq0$,
and $p\left(\cdot,0\right)\in L^{2}\left(-a,a\right)$. Suppose that
$K$ satisfying Assumptions \ref{assm:K_par}--\ref{assm:K_pos} is
such that equation (\ref{eq:K_kern_cond}) has at most one solution
$\phi_{\perp}\in L_{0}^{2}\left(-a,a\right)$ with $\left\Vert \phi_{\perp}\right\Vert =1$.
Moreover, if $\eta=0$, we additionally assume that $\left\langle p\left(\cdot,0\right)-\frac{P_{0}}{2a},\phi_{\perp}\right\rangle =0$.
Then, for the solution of (\ref{eq:p_eq}), we have the following
asymptotic results
\begin{equation}
\left\Vert p\left(\cdot,t\right)-p_{\infty}^{\left(1\right)}\right\Vert =\mathcal{O}\left(\exp\left(-\left(\mu+\frac{\nu}{\eta+\sigma_{1}}\right)t\right)\right),\hspace{1em}t\gg1,\hspace{1em}\alpha=1,\label{eq:p-p_infty_exp}
\end{equation}

\begin{equation}
\left\Vert p\left(\cdot,t\right)-p_{\infty}^{\left(1\right)}\right\Vert =\mathcal{O}\left(\frac{1}{t^{\alpha}}\right),\hspace{1em}t\gg1,\hspace{1em}\alpha\in\left(0,1\right)\cup\left(1,2\right),\label{eq:p-p_infty_alg}
\end{equation}
with
\begin{equation}
p_{\infty}^{\left(1\right)}\left(x\right):=\frac{P_{0}}{2a}+\sum_{k=1}^{\infty}\frac{\mu\left(\eta+\sigma_{k}\right)}{\mu\left(\eta+\sigma_{k}\right)+\nu}d_{k}^{0}\phi_{k}\left(x\right)+\frac{\mu\eta}{\mu\eta+\nu}d_{\perp}^{0}\phi_{\perp}\left(x\right),\hspace{1em}x\in\left(-a,a\right),\label{eq:p_infty_1_def}
\end{equation}
and $d_{k}^{0}$, $k\geq1$, $d_{\perp}^{0}$, as in (\ref{eq:d_k_0_l_k_def}).
\end{prop}

\begin{proof}
First, let us consider $\eta>0$. Application of Theorem \ref{thm:main_eta_pos},
yields 
\[
p\left(x,t\right)-\frac{P_{0}}{2a}=\sum_{k=1}^{\infty}d_{k}\left(t\right)\phi_{k}\left(x\right)+d_{\perp}\left(t\right)\phi_{\perp}\left(x\right),\hspace{1em}x\in\left(-a,a\right),\hspace{1em}t\geq0,
\]
or equivalently, rearranging the terms so that the right-hand side
contains only those proportional to $E_{\alpha}$,
\[
p\left(x,t\right)-p_{\infty}^{\left(1\right)}\left(x\right)=\sum_{k=1}^{\infty}\widetilde{d}_{k}\left(t\right)\phi_{k}\left(x\right)+\widetilde{d}_{\perp}\left(t\right)\phi_{\perp}\left(x\right),\hspace{1em}x\in\left(-a,a\right),\hspace{1em}t\geq0,
\]
with $p_{\infty}^{\left(1\right)}$ defined as in (\ref{eq:p_infty_1_def}),
and 
\begin{equation}
\widetilde{d}_{k}\left(t\right):=d_{k}\left(t\right)-\frac{\mu\left(\eta+\sigma_{k}\right)}{\mu\left(\eta+\sigma_{k}\right)+\nu}d_{k}^{0}=\frac{\nu}{\mu\left(\eta+\sigma_{k}\right)+\nu}d_{k}^{0}E_{\alpha}\left(-\left(\mu+\frac{\nu}{\eta+\sigma_{k}}\right)t^{\alpha}\right),\hspace{1em}t\geq0,\hspace{1em}k\geq1,\label{eq:d_k_tilde_def}
\end{equation}
\begin{equation}
\widetilde{d}_{\perp}\left(t\right):=d_{\perp}\left(t\right)-\frac{\mu\eta}{\mu\eta+\nu}d_{\perp}^{0}=\frac{\nu}{\mu\eta+\nu}d_{\perp}^{0}E_{\alpha}\left(-\left(\mu+\frac{\nu}{\eta}\right)t^{\alpha}\right),\hspace{1em}t\geq0.\label{eq:d_perp_tilde_def}
\end{equation}
Note that the series in (\ref{eq:p_infty_1_def}) converges in $L^{2}\left(-a,a\right)$
due to the Parseval's identity, since 
\begin{align*}
\left\Vert \sum_{k=1}^{\infty}\frac{\mu\left(\eta+\sigma_{k}\right)}{\mu\left(\eta+\sigma_{k}\right)+\nu}d_{k}^{0}\phi_{k}\right\Vert ^{2}=\sum_{k=1}^{\infty}\left|\frac{\eta+\sigma_{k}}{\mu\left(\eta+\sigma_{k}\right)+\nu}d_{k}^{0}\right|^{2} & \leq\sup_{n\geq1}\left(\frac{\mu\left(\eta+\sigma_{n}\right)}{\mu\left(\eta+\sigma_{n}\right)+\nu}\right)^{2}\sum_{k=1}^{\infty}\left|d_{k}^{0}\right|^{2}\\
 & <\sum_{k=1}^{\infty}\left|d_{k}^{0}\right|^{2}\leq\left\Vert q\left(\cdot,0\right)\right\Vert ^{2}<\infty,
\end{align*}
and hence we have $p_{\infty}^{\left(1\right)}\in L^{2}\left(-a,a\right)$.

By the orthonormality of $\phi_{k}$, $k\geq1$, and $\phi_{\perp}$,
we have
\begin{equation}
\left\Vert p\left(\cdot,t\right)-p_{\infty}^{\left(1\right)}\right\Vert =\left(\sum_{k=1}^{\infty}\left|\widetilde{d}_{k}\left(t\right)\right|^{2}+\left|\widetilde{d}_{\perp}\left(t\right)\right|^{2}\right)^{1/2}.\label{eq:p-p_infty}
\end{equation}
We can estimate
\begin{align*}
\left|\widetilde{d}_{k}\left(t\right)\right|\leq & \left|d_{k}^{0}\right|\left(\sup_{n\geq1}\frac{\nu}{\mu\left(\eta+\sigma_{n}\right)+\nu}\right)\left(\sup_{n\geq1}\left|E_{\alpha}\left(-\left(\mu+\frac{\nu}{\eta+\sigma_{n}}\right)t^{\alpha}\right)\right|\right)\\
\leq & \left|d_{k}^{0}\right|\left|E_{\alpha}\left(-\left(\mu+\frac{\nu}{\eta+\sigma_{1}}\right)t^{\alpha}\right)\right|,\hspace{1em}t\gg1,\hspace{1em}k\geq1,
\end{align*}
where, in the second line, we used $0<\sigma_{n}\leq\sigma_{1}$,
$n\geq1$, together with the fact that $\left|E_{\alpha}\left(-\tau\right)\right|$
is monotonically decreasing for sufficiently large $\tau$ (as evident
from asymptotic expansion (\ref{eq:E_a_large})). Consequently, (\ref{eq:p-p_infty})
implies
\begin{equation}
\left\Vert p\left(\cdot,t\right)-p_{\infty}^{\left(1\right)}\right\Vert ^{2}\leq\left[\left(\sum_{k=1}^{\infty}\left|d_{k}^{0}\right|^{2}\right)\left|E_{\alpha}\left(-\left(\mu+\frac{\nu}{\eta+\sigma_{1}}\right)t^{\alpha}\right)\right|^{2}+\left|d_{\perp}^{0}\right|^{2}\left|E_{\alpha}\left(-\left(\mu+\frac{\nu}{\eta}\right)t^{\alpha}\right)\right|^{2}\right],\hspace{1em}t\gg1.\label{eq:p-p_infty_estim}
\end{equation}
When $\alpha\in\left(0,1\right)\cup\left(1,2\right)$, the use of
(\ref{eq:E_a_large}) in (\ref{eq:p-p_infty_estim}) immediately gives
(\ref{eq:p-p_infty_alg}). When $\alpha=1$, we have $E_{1}\left(-\tau\right)=\exp\left(-t\right)$
(see (\ref{eq:E_1})), and we observe that the first term in the square
bracket of (\ref{eq:p-p_infty_estim}) is dominant as it decays slower
for $t\gg1$. This yields (\ref{eq:p-p_infty_exp}).

Now, when $\eta=0$, we use Theorem \ref{thm:main_eta_0} which gives
\[
p\left(x,t\right)-p_{\infty}^{\left(1\right)}\left(x\right)=\sum_{k=1}^{\infty}\widetilde{d}_{k}\left(t\right)\phi_{k}\left(x\right),\hspace{1em}x\in\left(-a,a\right),\hspace{1em}t\geq0,
\]
with $p_{\infty}^{\left(1\right)}$ and $\widetilde{d}_{k}$ defined
as before, according to (\ref{eq:p_infty_1_def}) and (\ref{eq:d_k_tilde_def}),
respectively. This leads to an analog of (\ref{eq:p-p_infty_estim}),
namely,
\[
\left\Vert p\left(\cdot,t\right)-p_{\infty}^{\left(1\right)}\right\Vert ^{2}\leq\frac{\nu^{2}}{\left(\mu\eta+\nu\right)^{2}}\left(\sum_{k=1}^{\infty}\left|d_{k}^{0}\right|^{2}\right)\left|E_{\alpha}\left(-\left(\mu+\frac{\nu}{\eta+\sigma_{1}}\right)t^{\alpha}\right)\right|^{2},\hspace{1em}t\gg1,
\]
and hence estimate (\ref{eq:p-p_infty_exp}) or (\ref{eq:p-p_infty_alg})
follows depending on the value of $\alpha$.
\end{proof}

\subsection{Transitional load\label{subsec:trans_load}}

We now consider the second scenario, when the load is of transitional
type, i.e. $P$ is taken to be a continuous function on $\mathbb{R}_{+}$
with $P\left(0\right)=:P_{0}$ and such that $P\left(t\right)\equiv P_{1}$
for $t\geq t_{0}$ with some $t_{0}>0$.
\begin{prop}
\label{prop:stat_regime2} Assume that $\mu\geq0$, $\alpha\in\left(0,2\right)$,
$\eta$, $\nu>0$, $P\in C_{b}\left(\mathbb{R}_{+}\right)$ with $P\left(0\right)=:P_{0}$
and such that $P\left(t\right)\equiv P_{1}$, $t\geq t_{0}$, for
some $t_{0}>0$, and $p\left(\cdot,0\right)\in L^{2}\left(-a,a\right)$.
Suppose that $K$ satisfying Assumptions \ref{assm:K_par}--\ref{assm:K_pos}
is such that equation (\ref{eq:K_kern_cond}) has at most one solution
$\phi_{\perp}\in L_{0}^{2}\left(-a,a\right)$ with $\left\Vert \phi_{\perp}\right\Vert =1$.
Then, for the solution of (\ref{eq:p_eq}), we have the following
asymptotic results
\begin{equation}
\left\Vert p\left(\cdot,t\right)-p_{\infty}^{\left(2\right)}\right\Vert =\mathcal{O}\left(\exp\left(-\left(\mu+\frac{\nu}{\eta+\sigma_{1}}\right)t\right)\right),\hspace{1em}t\gg1,\hspace{1em}\alpha=1,\label{eq:p-p_infty_exp_regime2}
\end{equation}

\begin{equation}
\left\Vert p\left(\cdot,t\right)-p_{\infty}^{(2)}\right\Vert =\mathcal{O}\left(\frac{1}{t^{\alpha}}\right),\hspace{1em}t\gg1,\hspace{1em}\alpha\in\left(0,1\right)\cup\left(1,2\right),\label{eq:p-p_infty_alg_regime2}
\end{equation}
with
\begin{equation}
p_{\infty}^{\left(2\right)}\left(x\right):=\frac{P_{1}}{2a}+\sum_{k=1}^{\infty}c_{k}^{\left(2\right)}\phi_{k}\left(x\right)+c_{\perp}^{\left(2\right)}\phi_{\perp}\left(x\right),\hspace{1em}x\in\left(-a,a\right),\label{eq:p_infty_2_def}
\end{equation}
where
\begin{equation}
c_{k}^{\left(2\right)}:=\frac{\mu\left(\eta+\sigma_{k}\right)}{\mu\left(\eta+\sigma_{k}\right)+\nu}d_{k}^{0}+\left(\frac{\nu}{\left(\eta+\sigma_{k}\right)\left(\mu\left(\eta+\sigma_{k}\right)+\nu\right)}-\frac{1}{\eta+\sigma_{k}}\right)\left(P_{1}-P_{0}\right)\frac{l_{k}}{2a},\hspace{1em}k\geq1,\label{eq:c_k_2}
\end{equation}
\begin{equation}
c_{\perp}^{\left(2\right)}:=\frac{\mu\eta}{\mu\eta+\nu}d_{\perp}^{0}+\left(\frac{\nu}{\mu\eta+\nu}-1\right)\left(P_{1}-P_{0}\right)\frac{l_{\perp}}{2a\eta},\label{eq:c_perp_2}
\end{equation}
and $d_{k}^{0}$, $k\geq1$, $d_{\perp}^{0}$ are as in (\ref{eq:d_k_0_l_k_def}).
\end{prop}

\begin{proof}
The proof is very similar to the one of Proposition \ref{prop:stat_regime1}
with the main difference that the separation of terms in (\ref{eq:d_k_final})--(\ref{eq:d_perp_final})
into the constant and the time-decaying components is now slightly
more complicated. Namely, for the present choice of the load function
$P$, from Theorem \ref{thm:main_eta_pos}, we have
\begin{align}
d_{k}\left(t\right)= & d_{k}^{0}\left[1+\frac{\nu}{\mu\left(\eta+\sigma_{k}\right)+\nu}\left(E_{\alpha}\left(-\left(\mu+\frac{\nu}{\eta+\sigma_{k}}\right)t^{\alpha}\right)-1\right)\right]-\frac{l_{k}}{2a\left(\eta+\sigma_{k}\right)}\left(P\left(t\right)-P_{0}\right)\label{eq:d_k_regime2}\\
 & -\frac{\nu l_{k}}{2a\left(\eta+\sigma_{k}\right)^{2}}\left(\mu+\frac{\nu}{\eta+\sigma_{k}}\right)^{1/\alpha-1}\int_{0}^{t_{0}}\mathcal{E}_{\alpha}\left(\left(\mu+\frac{\nu}{\eta+\sigma_{k}}\right)^{1/\alpha}\left(t-\tau\right)\right)P\left(\tau\right)\dd\tau\nonumber \\
 & +\frac{P_{0}}{2a}\frac{\nu l_{k}}{\left(\eta+\sigma_{k}\right)\left(\mu\left(\eta+\sigma_{k}\right)+\nu\right)}\left(E_{\alpha}\left(-\left(\mu+\frac{\nu}{\eta+\sigma_{k}}\right)t^{\alpha}\right)-1\right)\nonumber \\
 & -\frac{P_{1}}{2a}\frac{\nu l_{k}}{\left(\eta+\sigma_{k}\right)\left(\mu\left(\eta+\sigma_{k}\right)+\nu\right)}\left(E_{\alpha}\left(-\left(\mu+\frac{\nu}{\eta+\sigma_{k}}\right)\left(t-t_{0}\right)^{\alpha}\right)-1\right),\hspace{1em}k\geq1,\nonumber 
\end{align}
\begin{align}
d_{\perp}\left(t\right)= & d_{\perp}^{0}\left[1+\frac{\nu}{\mu\eta+\nu}\left(E_{\alpha}\left(-\left(\mu+\frac{\nu}{\eta}\right)t^{\alpha}\right)-1\right)\right]-\frac{l_{\perp}}{2a\eta}\left(P\left(t\right)-P_{0}\right)\label{eq:d_perp_regime2}\\
 & -\frac{\nu l_{\perp}}{2a\eta^{2}}\left(\mu+\frac{\nu}{\eta}\right)^{1/\alpha-1}\int_{0}^{t_{0}}\mathcal{E}_{\alpha}\left(\left(\mu+\frac{\nu}{\eta}\right)^{1/\alpha}\left(t-\tau\right)\right)P\left(\tau\right)\dd\tau\nonumber \\
 & +\frac{P_{0}}{2a}\frac{\nu l_{\perp}}{\eta\left(\mu\eta+\nu\right)}\left(E_{\alpha}\left(-\left(\mu+\frac{\nu}{\eta}\right)t^{\alpha}\right)-1\right)\nonumber \\
 & -\frac{P_{1}}{2a}\frac{\nu l_{\perp}}{\eta\left(\mu\eta+\nu\right)}\left(E_{\alpha}\left(-\left(\mu+\frac{\nu}{\eta}\right)\left(t-t_{0}\right)^{\alpha}\right)-1\right).\nonumber 
\end{align}
Adding and subtracting $\frac{l_{k}}{2a\left(\eta+\sigma_{k}\right)}P_{1}$
and $\frac{l_{\perp}}{2a\eta}P_{1}$ in (\ref{eq:d_k_regime2}) and
(\ref{eq:d_perp_regime2}), respectively, we rearrange the terms to
arrive at
\[
p\left(x,t\right)-p_{\infty}^{\left(2\right)}\left(x\right)=\frac{P\left(t\right)-P_{1}}{2a}+\sum_{k=1}^{\infty}\breve{d}_{k}\left(t\right)\phi_{k}\left(x\right)+\breve{d}_{\perp}\left(t\right)\phi_{\perp}\left(x\right),\hspace{1em}x\in\left(-a,a\right),\hspace{1em}t\geq0,
\]
with $p_{\infty}^{\left(2\right)}$ defined as in (\ref{eq:p_infty_2_def}),
and 
\begin{align}
\breve{d}_{k}\left(t\right):= & \left[d_{k}^{0}+\frac{P_{0}}{2a}\frac{l_{k}}{\left(\eta+\sigma_{k}\right)}\right]\frac{\nu}{\mu\left(\eta+\sigma_{k}\right)+\nu}E_{\alpha}\left(-\left(\mu+\frac{\nu}{\eta+\sigma_{k}}\right)t^{\alpha}\right)\label{eq:d_k_breve_def}\\
 & -\frac{\nu}{\left(\eta+\sigma_{k}\right)\left(\mu\left(\eta+\sigma_{k}\right)+\nu\right)}E_{\alpha}\left(-\left(\mu+\frac{\nu}{\eta+\sigma_{k}}\right)\left(t-t_{0}\right)^{\alpha}\right)\frac{P_{1}l_{k}}{2a}-\frac{\left(P\left(t\right)-P_{1}\right)l_{k}}{2a\left(\eta+\sigma_{k}\right)}\nonumber \\
 & -\frac{\nu l_{k}}{2a\left(\eta+\sigma_{k}\right)^{2}}\left(\mu+\frac{\nu}{\eta+\sigma_{k}}\right)^{1/\alpha-1}\int_{0}^{t_{0}}\mathcal{E}_{\alpha}\left(\left(\mu+\frac{\nu}{\eta+\sigma_{k}}\right)^{1/\alpha}\left(t-\tau\right)\right)P\left(\tau\right)\dd\tau,\hspace{1em}t\geq0,\hspace{1em}k\geq1,\nonumber 
\end{align}
\begin{align}
\breve{d}_{\perp}\left(t\right):= & \left[d_{\perp}^{0}+\frac{P_{0}}{2a\eta}l_{\perp}\right]\frac{\nu}{\mu\eta+\nu}E_{\alpha}\left(-\left(\mu+\frac{\nu}{\eta}\right)t^{\alpha}\right)\label{eq:d_perp_breve_def}\\
 & -\frac{\nu}{\mu\eta+\nu}E_{\alpha}\left(-\left(\mu+\frac{\nu}{\eta}\right)\left(t-t_{0}\right)^{\alpha}\right)\frac{P_{1}l_{\perp}}{2a\eta}-\frac{\left(P\left(t\right)-P_{1}\right)l_{\perp}}{2a\eta}\nonumber \\
 & -\frac{\nu l_{\perp}}{2a\eta^{2}}\left(\mu+\frac{\nu}{\eta}\right)^{1/\alpha-1}\int_{0}^{t_{0}}\mathcal{E}_{\alpha}\left(\left(\mu+\frac{\nu}{\eta}\right)^{1/\alpha}\left(t-\tau\right)\right)P\left(\tau\right)\dd\tau,\hspace{1em}t\geq0.\nonumber 
\end{align}
We note that here again the series in (\ref{eq:p_infty_2_def}) converges
in $L^{2}\left(-a,a\right)$ because of $\sum_{k=1}^{\infty}\left|c_{k}^{\left(2\right)}\right|^{2}<\infty$
which, in turn, follows from $\sum_{k=1}^{\infty}\left|d_{k}^{0}\right|^{2}<\infty$,
$\sum_{k=1}^{\infty}\left|l_{k}\right|^{2}<\infty$, since $\frac{\mu\left(\eta+\sigma_{k}\right)}{\mu\left(\eta+\sigma_{k}\right)+\nu}<1$
and the square-bracketed term in (\ref{eq:c_k_2}) is uniformly bounded
for all $k\geq1$ and $t\geq0$. 

To deduce estimates (\ref{eq:p-p_infty_exp_regime2})--(\ref{eq:p-p_infty_alg_regime2}),
we consider
\begin{align}
\left\Vert p\left(\cdot,t\right)-p_{\infty}^{(2)}-\frac{P\left(t\right)-P_{1}}{2a}\right\Vert ^{2} & =\sum_{k=1}^{\infty}\left|\breve{d}_{k}\left(t\right)\right|^{2}+\left|\breve{d}_{\perp}\left(t\right)\right|^{2}\label{eq:p-p_inf2_prelim}\\
 & \leq2\sum_{k=1}^{\infty}\left(\left|D_{k}\left(t\right)\right|^{2}\left|d_{k}^{0}\right|^{2}+\left|L_{k}\left(t\right)\right|^{2}\left|l_{k}\right|^{2}\right)+2\left|D_{0}\left(t\right)\right|^{2}\left|d_{\perp}^{0}\right|^{2}+2\left|L_{0}\left(t\right)\right|^{2}\left|l_{\perp}\right|^{2},\nonumber 
\end{align}
where we used the elementary inequality $\left(a+b\right)^{2}\leq2\left(a^{2}+b^{2}\right)$,
$a$, $b\in\mathbb{R}$, and we introduced
\begin{equation}
D_{0}\left(t\right):=\frac{\nu}{\mu\eta+\nu}E_{\alpha}\left(-\left(\mu+\frac{\nu}{\eta}\right)t^{\alpha}\right),\hspace{1em}D_{k}\left(t\right):=\frac{\nu}{\mu\left(\eta+\sigma_{k}\right)+\nu}E_{\alpha}\left(-\left(\mu+\frac{\nu}{\eta+\sigma_{k}}\right)t^{\alpha}\right),\hspace{1em}k\geq1,\label{eq:D_0_D_k_def}
\end{equation}
\begin{align}
L_{0}\left(t\right):= & \frac{P_{0}}{2a\eta}\frac{\nu}{\mu\eta+\nu}E_{\alpha}\left(-\left(\mu+\frac{\nu}{\eta}\right)t^{\alpha}\right)-\frac{P_{1}}{2a\eta}\frac{\nu}{\mu\eta+\nu}E_{\alpha}\left(-\left(\mu+\frac{\nu}{\eta}\right)\left(t-t_{0}\right)^{\alpha}\right)\label{eq:L_0_def}\\
 & -\frac{P\left(t\right)-P_{1}}{2a\eta}-\frac{\nu}{2a\eta^{2}}\left(\mu+\frac{\nu}{\eta}\right)^{1/\alpha-1}\int_{0}^{t_{0}}\mathcal{E}_{\alpha}\left(\left(\mu+\frac{\nu}{\eta}\right)^{1/\alpha}\left(t-\tau\right)\right)P\left(\tau\right)\dd\tau,\nonumber 
\end{align}
\begin{align}
L_{k}\left(t\right):= & \frac{P_{0}}{2a}\frac{\nu}{\left(\eta+\sigma_{k}\right)\left(\mu\left(\eta+\sigma_{k}\right)+\nu\right)}E_{\alpha}\left(-\left(\mu+\frac{\nu}{\eta+\sigma_{k}}\right)t^{\alpha}\right)\label{eq:L_k_def}\\
 & -\frac{P_{1}}{2a}\frac{\nu}{\left(\eta+\sigma_{k}\right)\left(\mu\left(\eta+\sigma_{k}\right)+\nu\right)}E_{\alpha}\left(-\left(\mu+\frac{\nu}{\eta+\sigma_{k}}\right)\left(t-t_{0}\right)^{\alpha}\right)\nonumber \\
 & -\frac{P\left(t\right)-P_{1}}{2a\left(\eta+\sigma_{k}\right)}-\frac{\nu}{2a\left(\eta+\sigma_{k}\right)^{2}}\left(\mu+\frac{\nu}{\eta+\sigma_{k}}\right)^{1/\alpha-1}\int_{0}^{t_{0}}\mathcal{E}_{\alpha}\left(\left(\mu+\frac{\nu}{\eta+\sigma_{k}}\right)^{1/\alpha}\left(t-\tau\right)\right)P\left(\tau\right)\dd\tau,\hspace{1em}k\geq1.\nonumber 
\end{align}
From asymptotics (\ref{eq:E_a_large}), (\ref{eq:cal_E_a_large}),
it follows that $\left|E_{\alpha}\left(-\tau\right)\right|$ and $\left|\mathcal{E}_{\alpha}\left(\tau\right)\right|$
are monotonically decreasing functions for sufficiently large $\tau$.
Since $0<\sigma_{k}\leq\sigma_{1}$, $k\geq1$, and $P\left(t\right)\equiv P_{1}$
for $t\geq t_{0}$, we can estimate from (\ref{eq:D_0_D_k_def})--(\ref{eq:L_k_def})
\[
\left|D_{0}\left(t\right)\right|\leq E_{\alpha}\left(-\left(\mu+\frac{\nu}{\eta}\right)t^{\alpha}\right),\hspace{1em}\left|D_{k}\left(t\right)\right|\leq E_{\alpha}\left(-\left(\mu+\frac{\nu}{\eta+\sigma_{1}}\right)t^{\alpha}\right),\hspace{1em}t\gg1,\hspace{1em}k\geq1,
\]
\begin{align*}
\left|L_{0}\left(t\right)\right|\leq & \frac{P_{0}}{2a\eta}\left|E_{\alpha}\left(-\left(\mu+\frac{\nu}{\eta}\right)t^{\alpha}\right)\right|+\frac{P_{1}}{2a\eta}\left|E_{\alpha}\left(-\left(\mu+\frac{\nu}{\eta}\right)\left(t-t_{0}\right)^{\alpha}\right)\right|\\
 & +\frac{\nu}{2a\eta^{2}}\left(\mu+\frac{\nu}{\eta}\right)^{1/\alpha-1}\left|\mathcal{E}_{\alpha}\left(\left(\mu+\frac{\nu}{\eta}\right)^{1/\alpha}\left(t-t_{0}\right)\right)\right|\left\Vert P\right\Vert _{L^{1}\left(0,t_{0}\right)},\hspace{1em}t\gg1,
\end{align*}
\begin{align*}
\left|L_{k}\left(t\right)\right|\leq & \frac{P_{0}}{2a\left(\eta+\sigma_{1}\right)}\left|E_{\alpha}\left(-\left(\mu+\frac{\nu}{\eta+\sigma_{1}}\right)t^{\alpha}\right)\right|+\frac{P_{1}}{2a\left(\eta+\sigma_{1}\right)}\left|E_{\alpha}\left(-\left(\mu+\frac{\nu}{\eta+\sigma_{1}}\right)\left(t-t_{0}\right)^{\alpha}\right)\right|\\
 & +\frac{\nu}{2a\eta\left(\eta+\sigma_{1}\right)}\sup_{n\geq1}\left(\mu+\frac{\nu}{\eta+\sigma_{n}}\right)^{1/\alpha-1}\left|\mathcal{E}_{\alpha}\left(\left(\mu+\frac{\nu}{\eta+\sigma_{1}}\right)^{1/\alpha}\left(t-t_{0}\right)\right)\right|\left\Vert P\right\Vert _{L^{1}\left(0,t_{0}\right)},\hspace{1em}t\gg1,\hspace{1em}k\geq1,
\end{align*}
For $\alpha\in\left(0,1\right)\cup\left(1,2\right)$, we reuse (\ref{eq:E_a_large}),
(\ref{eq:cal_E_a_large}) to get (\ref{eq:p-p_infty_alg_regime2}).
For $\alpha=1$, employing (\ref{eq:e_a})--(\ref{eq:E_1}) which
entails that $E_{1}\left(-\tau\right)=\left|\mathcal{E}_{1}\left(\tau\right)\right|=\exp\left(-\tau\right)$,
$\tau>0$, we deduce (\ref{eq:p-p_infty_exp_regime2}). In both cases,
we again used that $P\left(t\right)\equiv P_{1}$ for $t\geq t_{0}$
which entailed identical vanishing of the third term in the left-hand
side of (\ref{eq:p-p_inf2_prelim}).
\end{proof}

\section{Numerical illustrations\label{sec:numerics}}

To verify and illustrate some of the obtained results numerically,
we consider the particular kernel function $K_{0}$ as given by (\ref{eq:K0_def}),
with $C_{K}=1.6$. We use a specially written MATLAB code which employs
an external function \cite{Podl} for computing $E_{\alpha}$.

\subsection{Verification of the solution for a variable load\label{subsec:num_verif}}

Let us fix the following set of parameters $a=1$, $\nu=2$, $\eta=1$,
$\mu=1.2$, $\alpha=0.6$. Consider the load profile given by 
\begin{equation}
P\left(t\right)=P_{1}\chi_{\left(t_{0},\infty\right)}\left(t\right)+\left[\frac{P_{0}+P_{1}}{2}-\frac{P_{1}-P_{0}}{2}\cos\left(\frac{\pi t}{t_{0}}\right)\right]\chi_{\left(0,t_{0}\right)}\left(t\right),\hspace{1em}t\geq0,\label{eq:P_ext_cos}
\end{equation}
$P_{0}=6$, $P_{1}=10$, $t_{0}=0.5$ consistent with Subsection \ref{subsec:trans_load}.
This describes a smooth load switch from $P_{0}$ to $P_{1}$ occurring
over time $t_{0}$ and remaining constant afterwards. For the sake
of simplicity, let us suppose that the function $\Delta$ in (\ref{eq:p0_int_eq})
is chosen such that 
\begin{equation}
p_{0}\left(x\right)=\frac{2P_{0}}{a^{2}\pi}\sqrt{a^{2}-x^{2}},\hspace{1em}x\in\left(-a,a\right),\label{eq:p0_sqrt}
\end{equation}
and note that, for such a choice, (\ref{eq:equil_eq}) is automatically
satisfied.

We verify the approach presented in Section \ref{sec:sol} by comparing
the solution formula given in Theorem \ref{thm:main_eta_pos}, with
$d_{\perp}\left(t\right)\equiv0$, $d_{\perp}^{0}=l_{\perp}=0$ (according
to Proposition \ref{prop:K0_zero_kern}), to a numerical finite-difference
method for solving integral equation (\ref{eq:q_eq}) in time. In
particular, we truncate the series in (\ref{eq:q_sol}) at $60$ terms
(however, much less would already be sufficient). For the numerical
approach, we combine Nystr{\"o}m collocation method with a finite-difference
scheme in time using singularity subtraction. In Figure \ref{fig:alph06_p_x},
we see that both solutions almost coincide for all shown instances
of time. The solutions have an oscillatory mismatch, especially in
small regions close to the endpoints $x=\pm1$, which is a typical
phenomenon for a spectral approach. 

\begin{figure}
\begin{centering}
\includegraphics[scale=0.6]{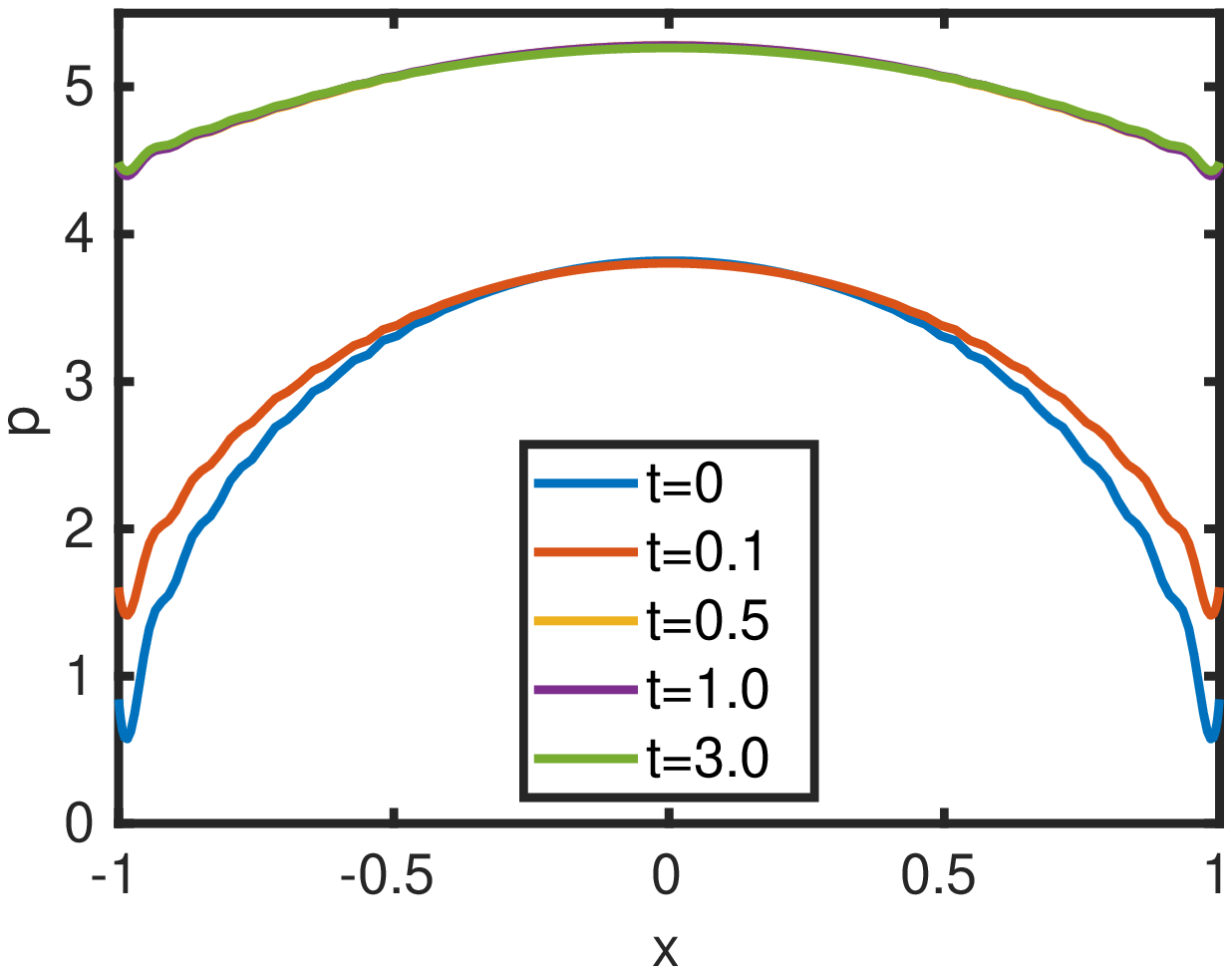}\includegraphics[scale=0.6]{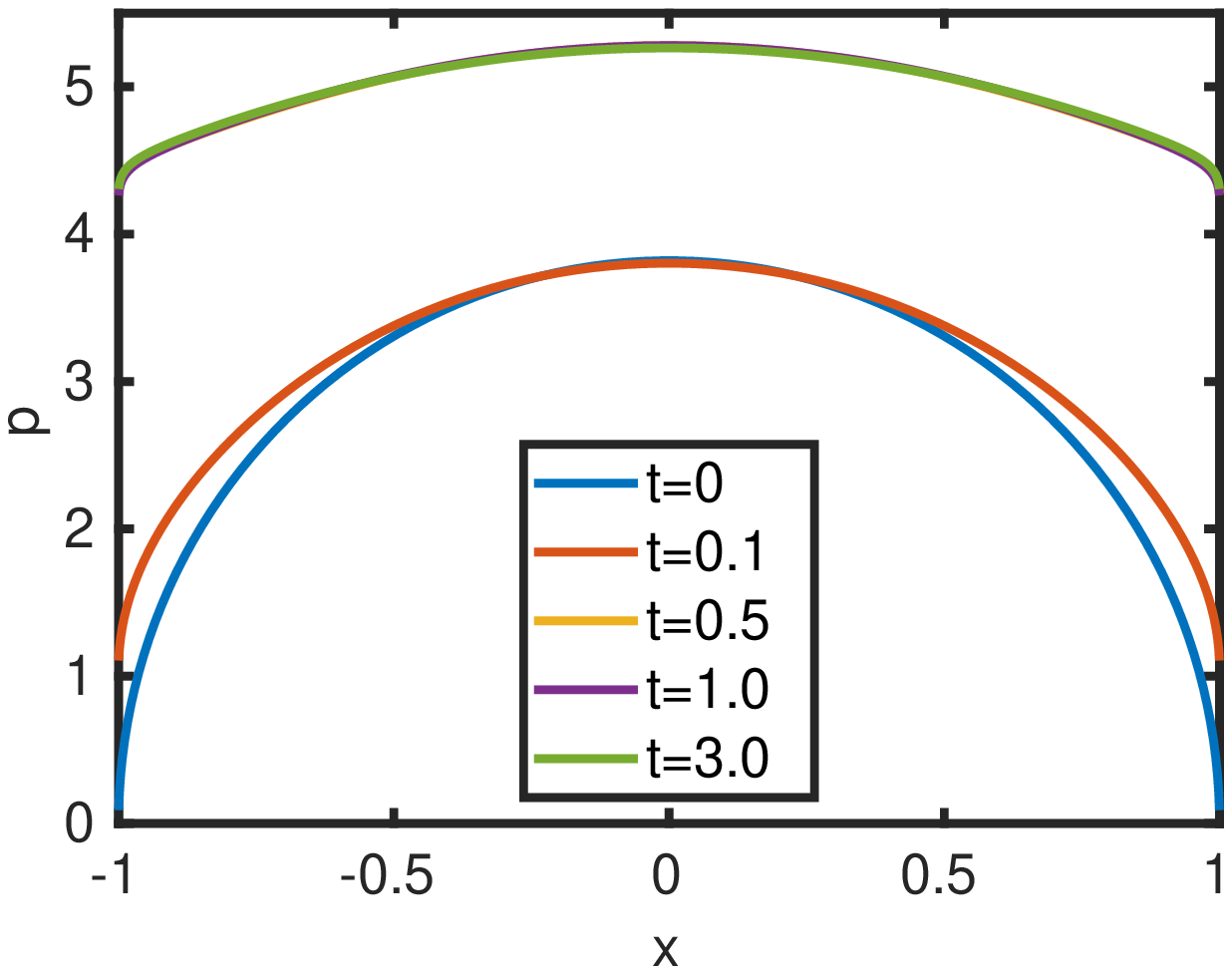}
\par\end{centering}
\caption{\label{fig:alph06_p_x} Time evolution of the pressure profile computed
from equation (\ref{eq:p_sol}) (left) and by a numerical finite-difference
method (right) }
\end{figure}

\subsection{Dependence of the stationary state on the model parameter $\mu$\label{subsec:mu_depend}}

Let us consider the same setup as in Subsection \ref{subsec:num_verif},
but instead of a fixed value of $\mu$, we explore the range of this
model parameter starting from the degenerate case $\mu=0$ (purely
fractional order model for the wear term (\ref{eq:w_p_term}) given
by (\ref{eq:RL_int})) up to $\mu=6$. This is done in order to investigate
the effect of such a parameter variation on an essential output of
the model: the stationary pressure distribution $p_{\infty}^{\left(2\right)}\left(x\right)$
given by (\ref{eq:p_infty_2_def}). As Figure \ref{fig:stat_st_var_mu}
shows, the increase of $\mu$ amounts to steepening of the curve $p_{\infty}^{\left(2\right)}\left(x\right)$
making a deviation from the uniform pressure distribution more pronounced.
Note that, according to (\ref{eq:p_infty_2_def}), the stationary
pressure distribution is independent of the model parameter $\alpha$.

\begin{figure}
\begin{centering}
\includegraphics[scale=0.6]{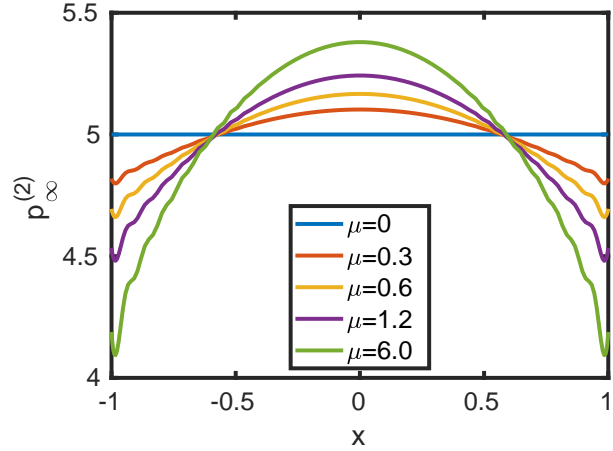}
\par\end{centering}
\caption{\label{fig:stat_st_var_mu} Effect of the model parameter $\mu$ on
the stationary state $p_{\infty}^{\left(2\right)}\left(x\right)$}
\end{figure}

\subsection{Illustration of the convergence under a constant load\label{subsec:t_conv}}

We now consider a simple setting which is more classical for analysis:
the constant load case, i.e. we replace (\ref{eq:P_ext_cos}) with
$P\left(t\right)\equiv P_{0}$, $t\geq0$. We take (\ref{eq:p0_sqrt})
and numerical values of $P_{0}$ and other parameters as in Subsection
\ref{subsec:num_verif}, except for the model ``order'' parameter
$\alpha$ which we now vary. The goal here is to demonstrate qualitatively
different behaviour of the model depending on this parameter. In particular,
we investigate three different values of $\alpha$, by looking at
the pressure evolution at two particular points. For $\alpha=1$,
the convergence to the stationary value occurs fast, and hence Figure
\ref{fig:alph1_p_t} corroborates exponential bound (\ref{eq:p-p_infty_exp}).
For all other values of $\alpha$ inside the interval $\left(0,2\right)$,
the stabilisation occurs only at an algebraic rate. In particular,
this is illustrated in Figures \ref{fig:alph06_p_t} and \ref{fig:alph12_p_t}
where values $\alpha=0.6$ and $\alpha=1.2$ are taken, respectively.
We see that in the first case convergence happens monotonically whereas
in the second one it is accompanied by oscillations. Figure \ref{fig:alph18_p_t}
shows, on extended time interval, that for a larger value ($\alpha=1.8$),
oscillations increase even more.

\begin{figure}
\begin{centering}
\includegraphics[scale=0.6]{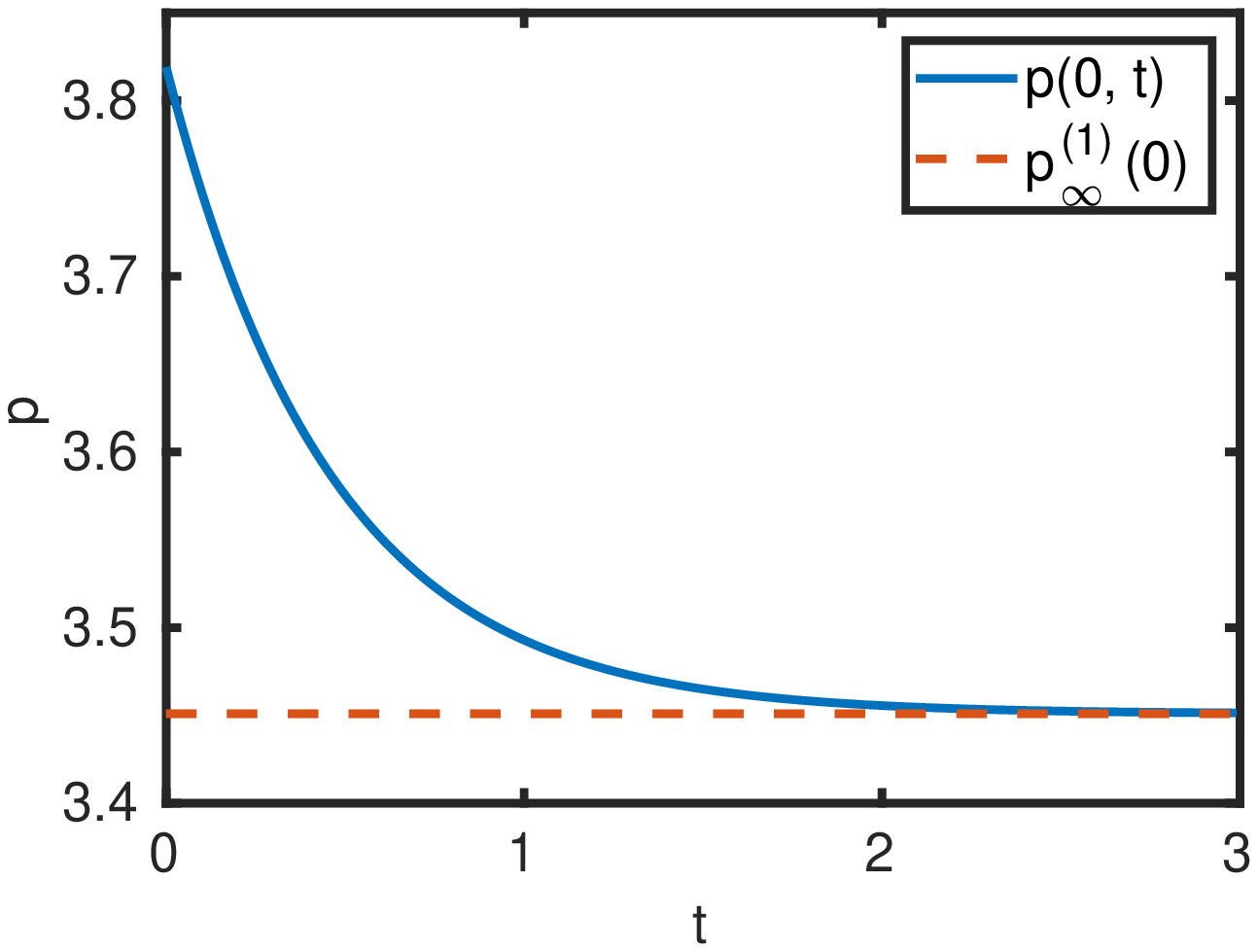}\includegraphics[scale=0.6]{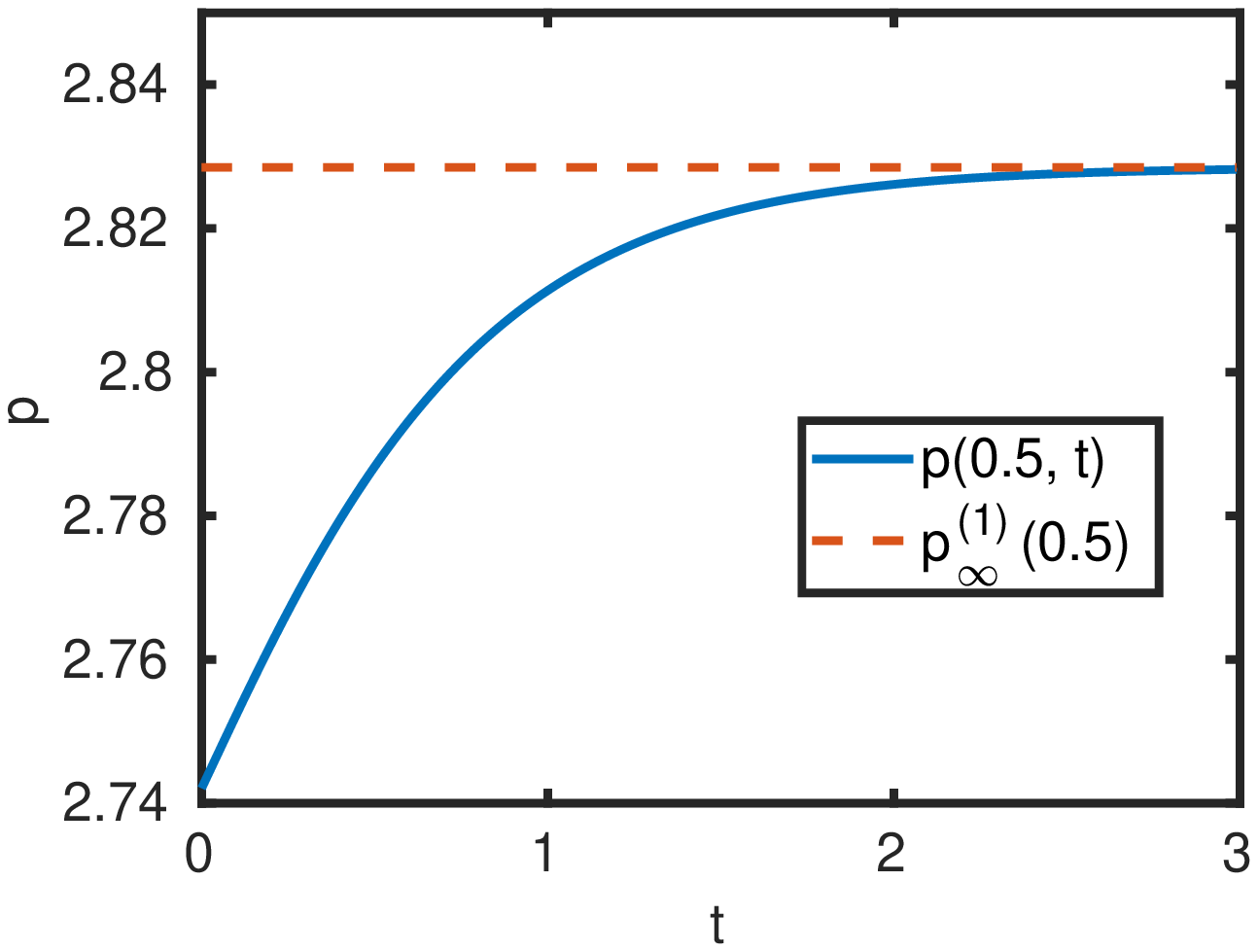}
\par\end{centering}
\caption{\label{fig:alph1_p_t} $\alpha=1$: Time evolution of the pressure
at $x=0$ (left) and $x=0.5$ (right)}
\end{figure}

\begin{figure}
\begin{centering}
\includegraphics[scale=0.6]{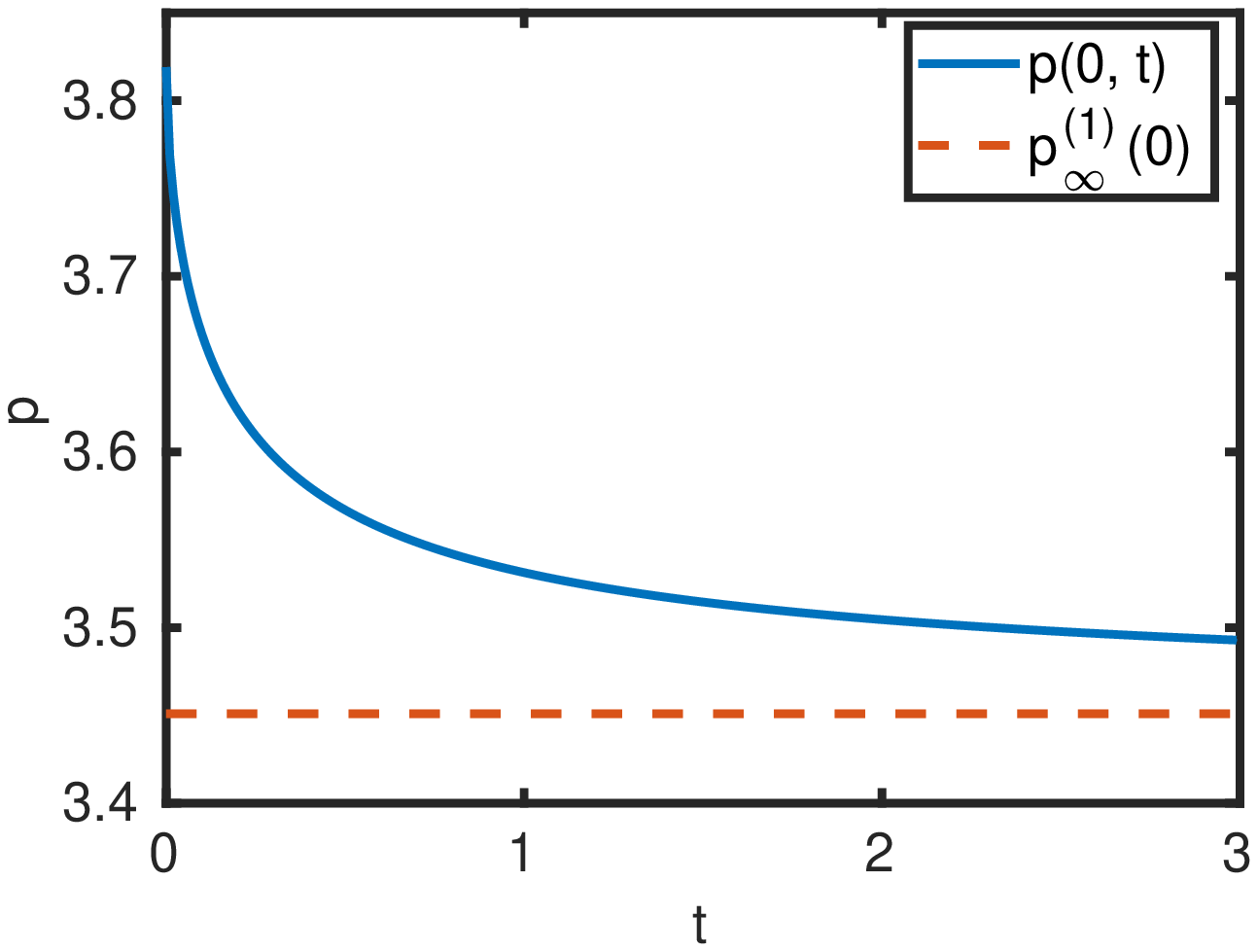}\includegraphics[scale=0.6]{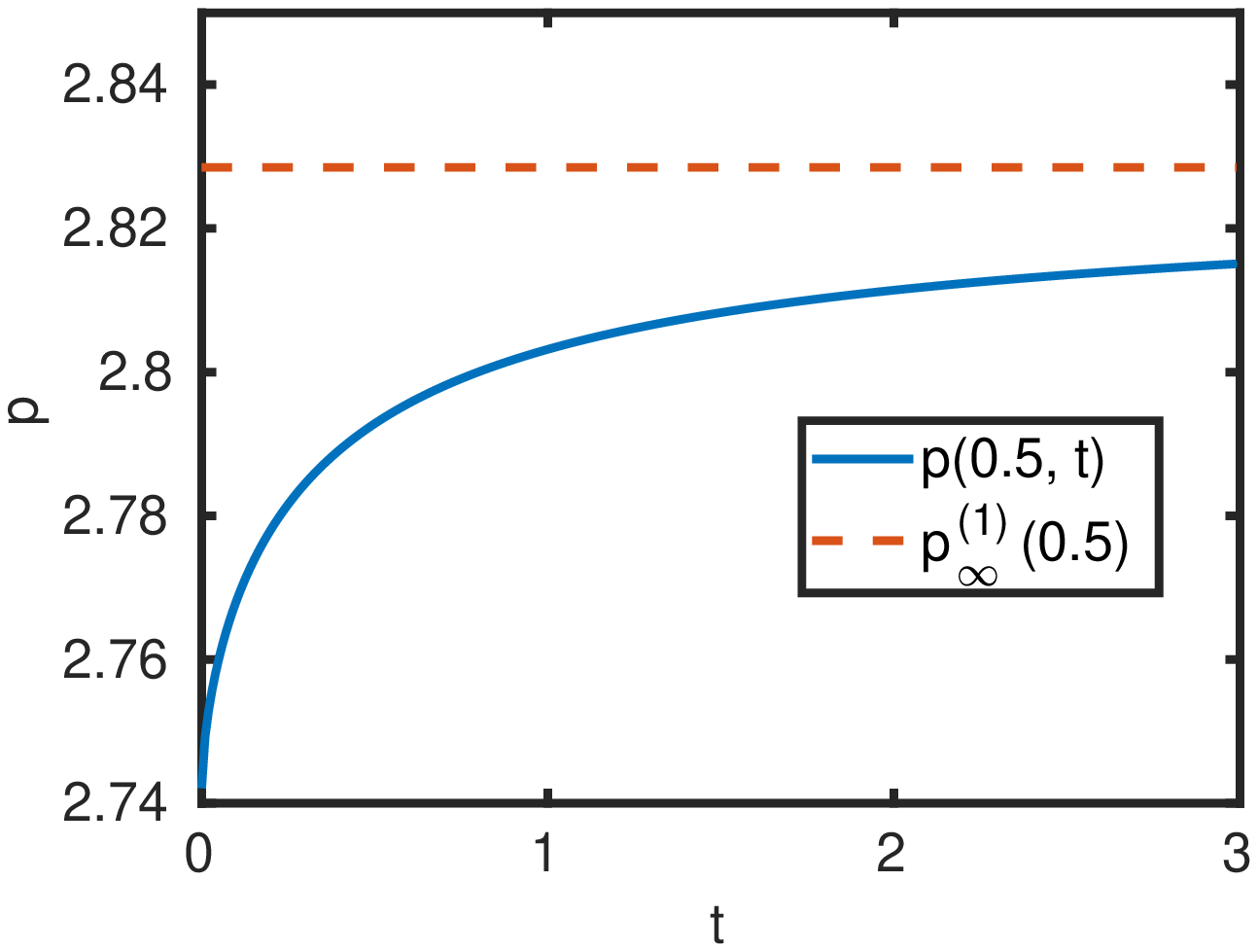}
\par\end{centering}
\caption{\label{fig:alph06_p_t} $\alpha=0.6$: Time evolution of the pressure
at $x=0$ (left) and $x=0.5$ (right)}
\end{figure}
\begin{figure}
\begin{centering}
\includegraphics[scale=0.6]{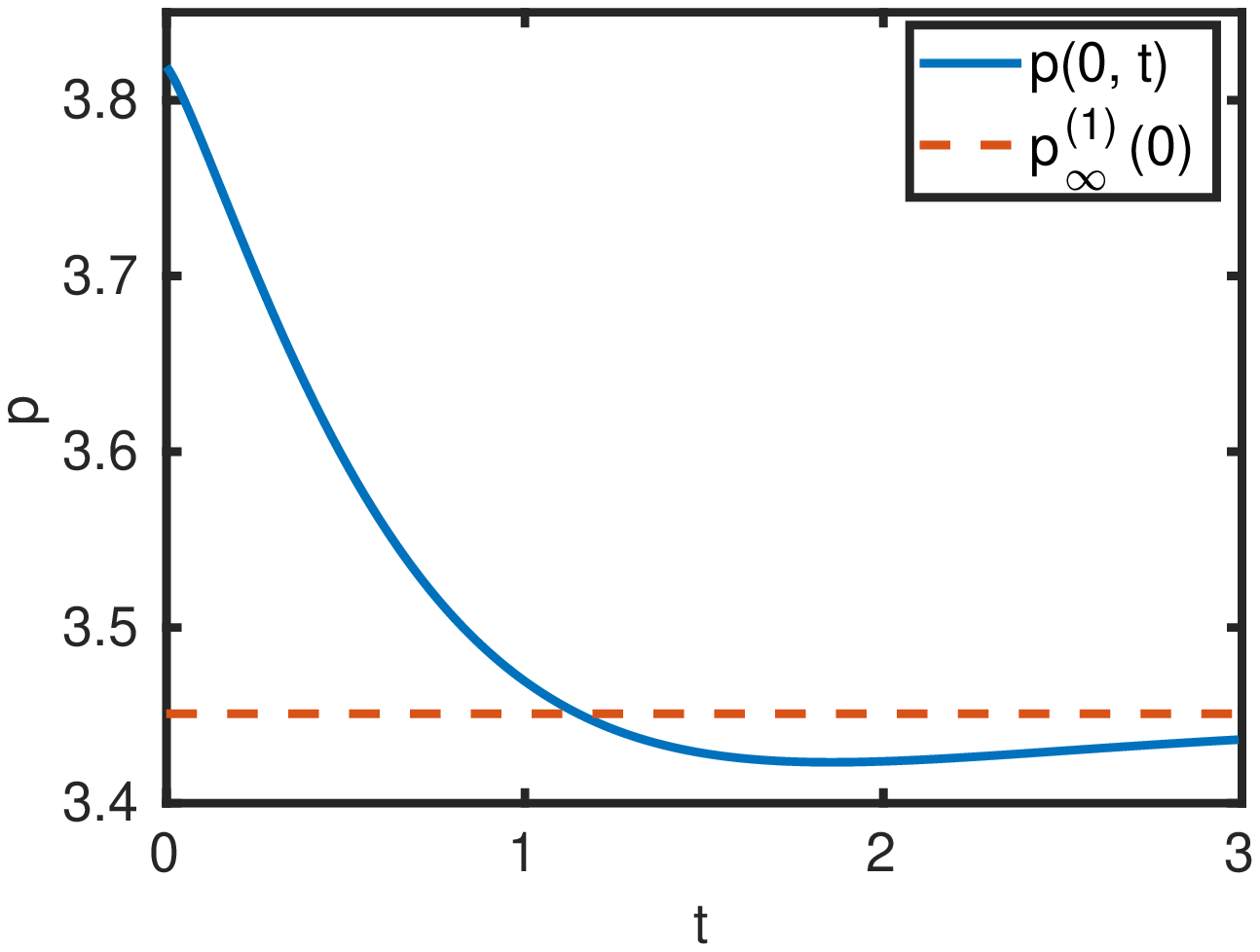}\includegraphics[scale=0.6]{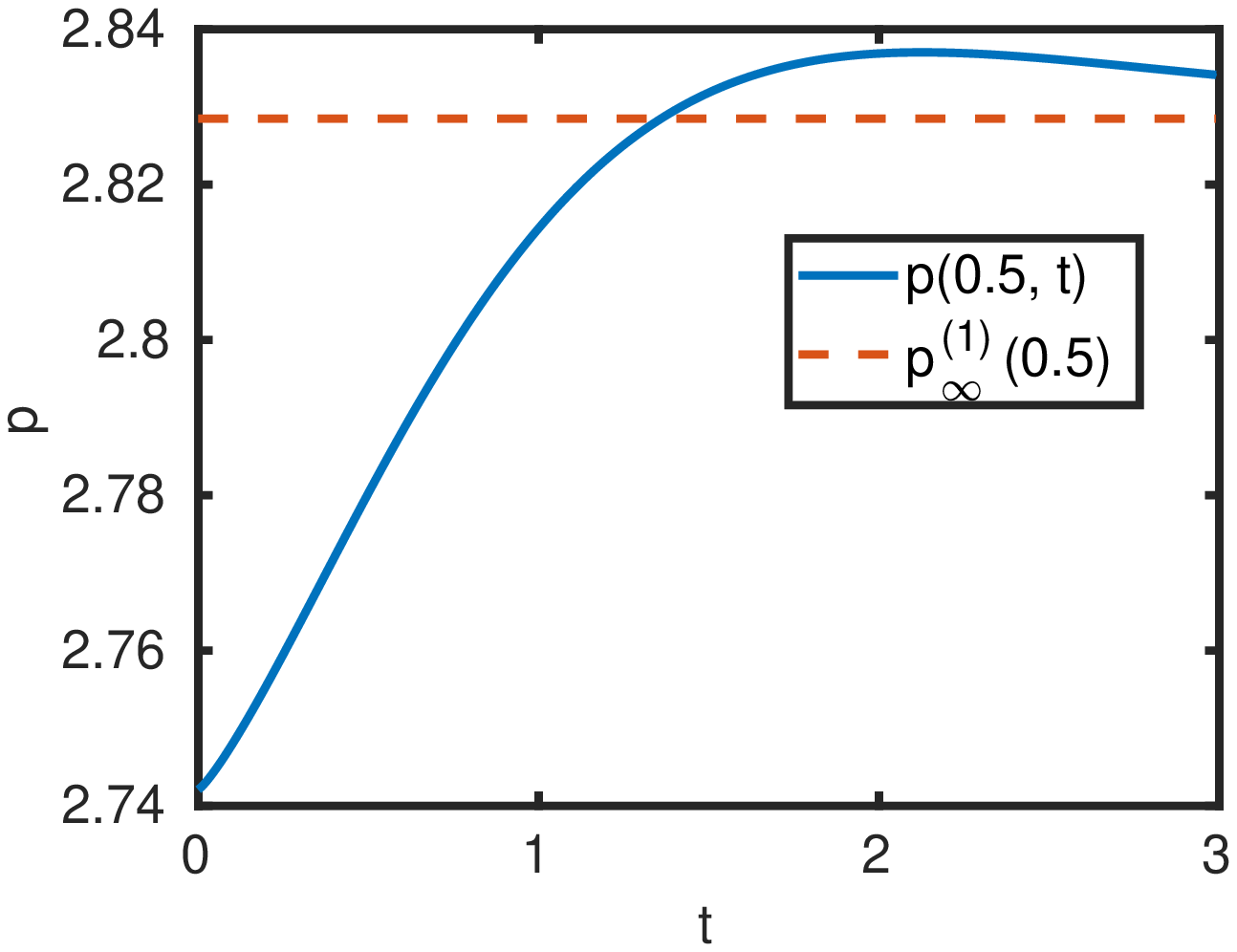}
\par\end{centering}
\caption{\label{fig:alph12_p_t} $\alpha=1.2$: Time evolution of the pressure
at $x=0$ (left) and $x=0.5$ (right)}
\end{figure}
\begin{figure}
\begin{centering}
\includegraphics[scale=0.6]{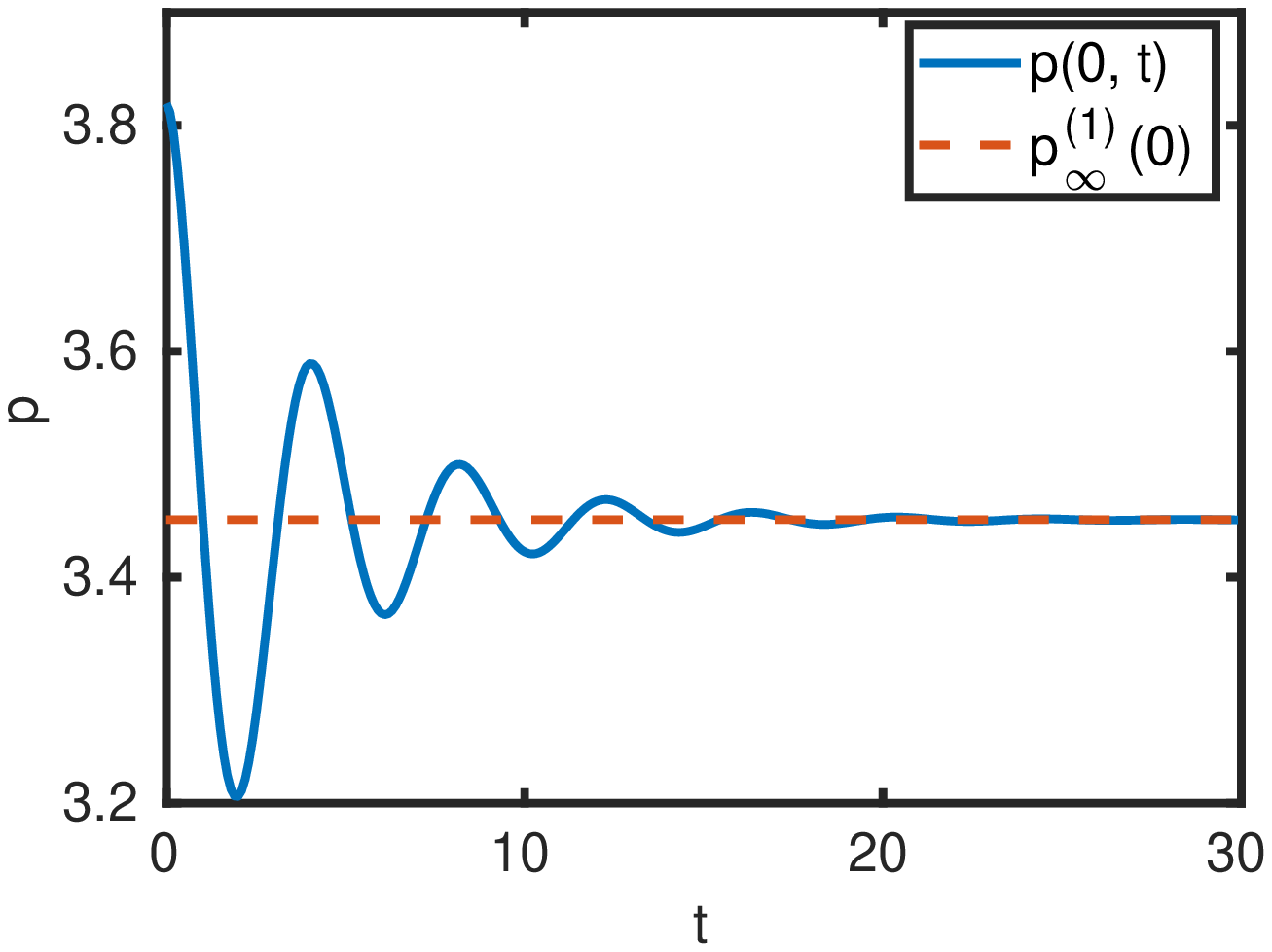}\includegraphics[scale=0.6]{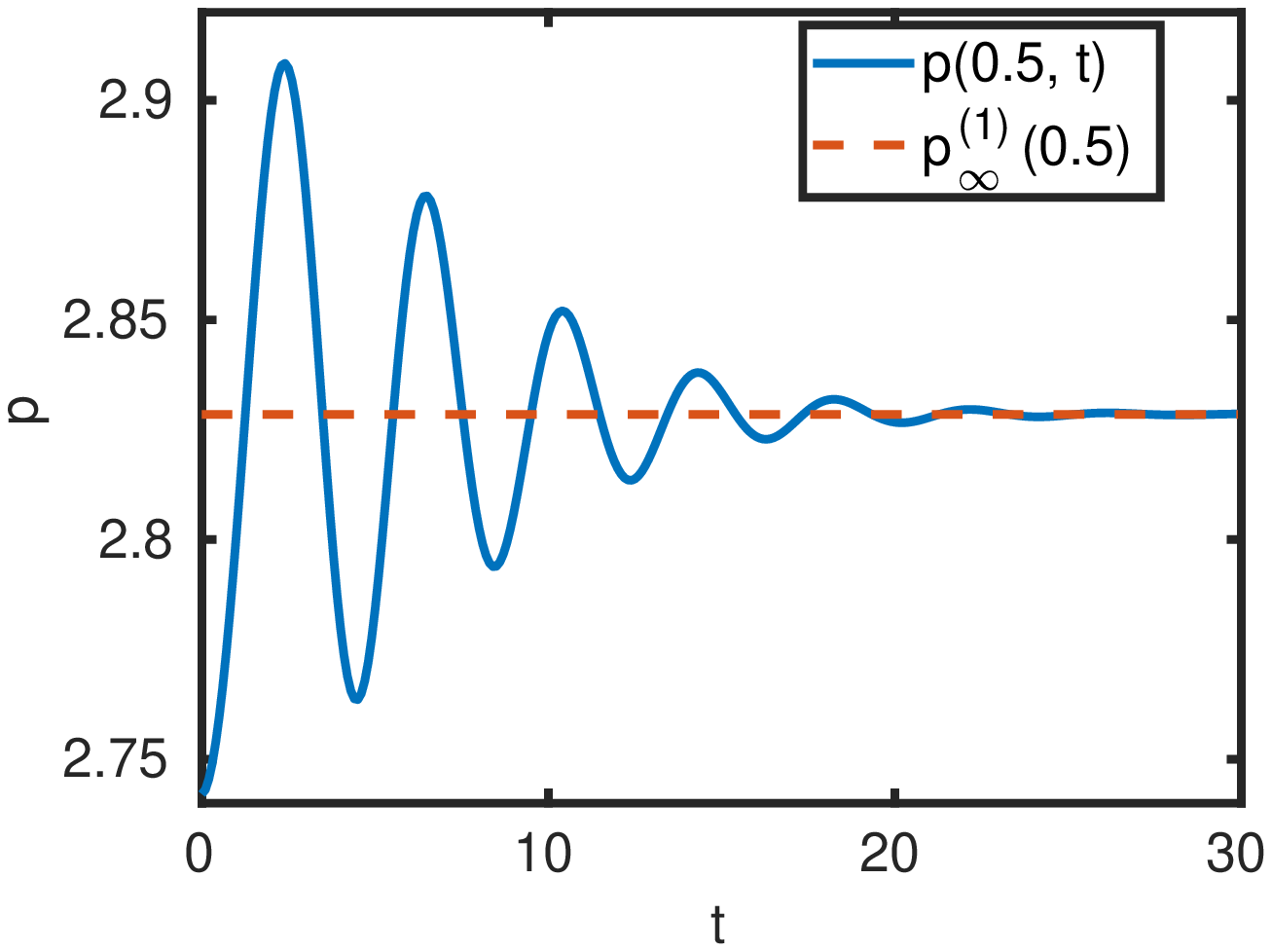}
\par\end{centering}
\caption{\label{fig:alph18_p_t} $\alpha=1.8$: Time evolution of the pressure
at $x=0$ (left) and $x=0.5$ (right)}
\end{figure}

\section{Discussion and conclusion\label{sec:discuss}}

Several formulations of the punch-sliding problem have been considered
at once. These included rather general relation between contact pressure and displacement
generalising the half-space geometry, a
potential presence of coating (or additional surface microstructure),
time-varying load and a linear material wear. In particular, the wear
was modelled through a novel non-local relation between the contact
pressure and the wear generalising the classical one. One advantage
of transient models with such a wear relation is that they admit a
closed-form solution as long as the spatial part of the formulation
can be easily resolved and, most importantly, analysed. This possibility
should not be underestimated in contexts of inverse design of mechanical
materials.

Another relevant advantage is that the solution may exhibit
predictably different transient behavior and the stationary pressure
distribution depending on newly introduced parameters entering the
formulation. As it was shown numerically, these parameters can be
chosen to ensure a slow algebraic stabilisation towards the stationary
distribution whether its monotone or oscillatory. This complements
the fast exponential stabilisation (which happens in the particular
case of the model parameter $\alpha=1$), a typically occuring phenomenon
for such problems, and paves a way for description of new materials
such as polymers within the same framework.

The non-uniform stationary
pressure distributions are also possible and correspond to non-zero
values of the model parameter $\mu$. At the next stage, a practical
validation of the obtained results and a fit of the model to experimental
data is highly desirable. In doing this, the parameter $\mu$ should
be chosen from an observation of the stationary pressure distribution
whereas the parameter $\alpha$ should be determined from temporal
observations (the convergence speed), according to what was described
above.

We thus conclude that the proposed model has a potential to
describe materials whose temporal convergence to a stationary state is slow
as well as those for which the stationary pressure distribution is not uniform.  

As a continuation of this work, analysis of the long-time behavior
under a periodic load $P$ has already been considered in \cite{Pon}. More challenging would be to deal rigorously
with more general kernel functions $K$ such as those that do not
satisfy the parity assumption. However, it seems that the
positivity assumption is not essential, and, in view of that, some
technical results in Appendix are already provided for a slightly
more general setting than the one being considered here. 

\section*{Acknowledgements}

The author is grateful for the support of the bi-national FWF-project
I3538-N32 of the Austrian Science Fund (FWF) used for his employment
at Vienna University of Technology. The paper has benefited from fruitful
discussions with Ivan Argatov on mechanical aspects of the matter. 

\appendix
\section{}
\renewcommand{\thesection}{\Alph{section}}We collect here some auxiliary results that are needed in the paper.

\subsection{Basic theoretical facts about compact linear integral operators}
\begin{lem}
\label{lem:HS_kern} Let $K\in L^{2}\left(\left(-a,a\right)\times\left(-a,a\right)\right)$,
i.e. such that 
\begin{equation}
\int_{-a}^{a}\int_{-a}^{a}\left[K\left(x,\xi\right)\right]^{2}\dd x\dd\xi<\infty.\label{eq:k_HS_cond}
\end{equation}
Then, the integral transformation $f\mapsto\mathfrak{K}\left[f\right]\left(x\right):=\int_{-a}^{a}K\left(x,\xi\right)f\left(\xi\right)\dd\xi$
is a compact linear operator from $L^{2}\left(-a,a\right)$ to $L^{2}\left(-a,a\right)$.
\end{lem}

\begin{proof}
Under assumption of the validity of (\ref{eq:k_HS_cond}), it follows
that $\mathfrak{K}\left[f\right]\in L^{2}\left(-a,a\right)$ by an
application of the Cauchy-Schwarz inequality (see also \cite[Lem 3.2.3]{Hack}).
The fact that the operator $\mathfrak{K}$ is compact can be shown
\cite[Thm 3.2.7]{Hack} using Weierstrass approximation theorem and
the characterisation of a compact operator as the limit of finite
rank operators.
\end{proof}
\begin{lem}
\label{lem:A_spec} Let $X$ be a Hilbert space and $\mathcal{A}:\,X\rightarrow X$
be a compact self-adjoint operator. Then, the eigenvalues of $\mathcal{A}$
form a real non-increasing in absolute value sequence, and every eigenvalue
different from zero has finite multiplicity. Moreover, the eigenfunctions
of $\mathcal{A}$ form an orthonormal basis in $\overline{\text{Ran }\mathcal{A}}$,
the closure of the range of $\mathcal{A}$.
\end{lem}

\begin{proof}
The statement of the lemma is a collection of standard results of
the spectral theory of compact self-adjoint operators. A
close result in the form of a single theorem is given in \cite[Thm 2.8.15]{AtkHan}.
See also \cite[Sect. 97 Thm p.242]{Rie-Sz} and \cite[Thm 5.6]{Teschl}.
\end{proof}
\begin{lem}
\label{lem:Rayl-Hilb}\cite[Sect. 95 Thm p. 237]{Rie-Sz} Let $X$
be a Hilbert space with the inner product $\left\langle \cdot,\cdot\right\rangle _{X}$,
and $\mathcal{A}:\,X\rightarrow X$ is a compact self-adjoint operator.
Then, $\lambda_{n}^{+}$, the $n$-th largest positive eigenvalue
of $\mathcal{A}$ can be characterised as
\[
\lambda_{1}^{+}=\max_{\substack{f\in X}
}\frac{\left\langle \mathcal{A}\left[f\right],f\right\rangle _{X}}{\left\langle f,f\right\rangle _{X}},\hspace{1em}\hspace{1em}\lambda_{n}^{+}=\max_{\substack{f\in X\\
f\perp\left(\varphi_{1}^{+},\ldots\varphi_{n-1}^{+}\right)
}
}\frac{\left\langle \mathcal{A}\left[f\right],f\right\rangle _{X}}{\left\langle f,f\right\rangle _{X}},\hspace{1em}n\geq2,
\]
where $\varphi_{1}^{+},\dots\varphi_{n-1}^{+}$ are the eigenfunctions
corresponding to the $n-1$ largest positive eigenvalues, and the
orthogonality condition $f\perp\left(\varphi_{1}^{+},\ldots\varphi_{n-1}^{+}\right)$
means that $\left\langle f,\varphi_{j}^{+}\right\rangle _{X}=0$,
$j=1$, $\ldots$, $n-1$. Similarly, $\lambda_{n}^{-}$, the $n$-th
smallest negative eigenvalue of $\mathcal{A}$ can be characterised
as
\[
\lambda_{1}^{-}=\min_{\substack{f\in X}
}\frac{\left\langle \mathcal{A}\left[f\right],f\right\rangle _{X}}{\left\langle f,f\right\rangle _{X}},\hspace{1em}\hspace{1em}\lambda_{n}^{-}=\min_{\substack{f\in X\\
f\perp\left(\varphi_{1}^{-},\ldots\varphi_{n-1}^{-}\right)
}
}\frac{\left\langle \mathcal{A}\left[f\right],f\right\rangle _{X}}{\left\langle f,f\right\rangle _{X}},\hspace{1em}n\geq2,
\]
where $\varphi_{1}^{-},\dots\varphi_{n-1}^{-}$ are the eigenfunctions
corresponding to the $n-1$ smallest negative eigenvalues.
\end{lem}

\phantom{.}
\begin{lem}
\label{lem:Cour-Fisch}\cite[Sect. 95 Thm p. 237]{Rie-Sz} Let $X$
be a Hilbert space with the inner product $\left\langle \cdot,\cdot\right\rangle _{X}$,
and $\mathcal{A}:\,X\rightarrow X$ is a compact self-adjoint operator.
Then, $\lambda_{n}$, the $n$-th largest eigenvalue of $\mathcal{A}$
can be characterised as
\[
\lambda_{n}=\min_{f_{1},\ldots,f_{n-1}\in X}\max_{\substack{f\in X\\
f\perp\left(f_{1},\ldots f_{n-1}\right)
}
}\frac{\left\langle \mathcal{A}\left[f\right],f\right\rangle _{X}}{\left\langle f,f\right\rangle _{X}},\hspace{1em}n\geq2,
\]
where the orthogonality condition $f\perp\left(f_{1},\ldots f_{n-1}\right)$
means that $\left\langle f,f_{j}\right\rangle _{X}=0$, $j=1$, $\ldots$,
$n-1$.
\end{lem}

\subsection{Some special functions and their properties }

\phantom{abc}\\\vspace{-5pt}

\uline{Gamma function}:
\begin{equation}
\Gamma\left(z\right):=\int_{0}^{\infty}x^{z-1}e^{-x}\dd x,\hspace{1em}z>0.\label{eq:gamma_fct}
\end{equation}
This function satisfies the following fundamental relation
\begin{equation}
\Gamma\left(z+1\right)=z\Gamma\left(z\right),\hspace{1em}\hspace{1em}\Gamma\left(n+1\right)=n!,\hspace{1em}z>0,\hspace{1em}n\in\mathbb{N}_{0},\label{eq:gamma_prop}
\end{equation}
as well as Euler's reflection formula
\begin{equation}
\Gamma\left(1-z\right)\Gamma\left(z\right)=\frac{\pi}{\sin\left(\pi z\right)},\hspace{1em}z\notin\mathbb{Z}.\label{eq:gamma_Euler}
\end{equation}

\uline{Mittag-Leffler and relevant functions}:
\begin{equation}
E_{\alpha}\left(z\right):=\sum_{k=0}^{\infty}\frac{z^{k}}{\Gamma\left(\alpha k+1\right)},\hspace{1em}\hspace{1em}E_{\alpha,\beta}\left(z\right):=\sum_{k=0}^{\infty}\frac{z^{k}}{\Gamma\left(\alpha k+\beta\right)},\hspace{1em}\hspace{1em}z\in\mathbb{C},\hspace{1em}\alpha,\,\beta>0,\label{eq:E_a}
\end{equation}
\begin{equation}
e_{\alpha}\left(z;\lambda\right):=\frac{\dd}{\dd z}E_{\alpha}\left(-\lambda z^{\alpha}\right)=\frac{\alpha}{z}\sum_{k=1}^{\infty}\frac{\left(-1\right)^{k}k\lambda^{k}z^{\alpha k}}{\Gamma\left(\alpha k+1\right)},\hspace{1em}z>0,\hspace{1em}\alpha>0,\hspace{1em}\lambda\in\mathbb{C},\label{eq:e_a}
\end{equation}
\begin{equation}
\mathcal{E_{\alpha}}\left(z\right):=e_{\alpha}\left(z;1\right)=\frac{\alpha}{z}\sum_{k=1}^{\infty}\frac{\left(-1\right)^{k}kz^{\alpha k}}{\Gamma\left(\alpha k+1\right)},\hspace{1em}z>0,\hspace{1em}\alpha>0.\label{eq:cal_E_a}
\end{equation}
The following two identities are direct consequences of definitions
(\ref{eq:E_a}) and (\ref{eq:gamma_prop})
\begin{equation}
E_{\alpha,1}\left(z\right)=E_{\alpha}\left(z\right),\hspace{1em}\hspace{1em}E_{1}\left(z\right)=\exp\left(z\right),\hspace{1em}z\in\mathbb{C}.\label{eq:E_1}
\end{equation}

Note that when $\lambda>0$, we can write
\begin{equation}
e_{\alpha}\left(z;\lambda\right)=\lambda^{1/\alpha}\mathcal{E_{\alpha}}\left(\lambda^{1/\alpha}z\right),\hspace{1em}z>0,\hspace{1em}\alpha>0,\hspace{1em}\lambda>0.\label{eq:e_a_cal_E_a}
\end{equation}
We have a useful relation
\begin{equation}
E_{\alpha,\alpha}\left(z\right)=\alpha E_{\alpha,1}^{\prime}\left(z\right)=\alpha E_{\alpha}^{\prime}\left(z\right),\hspace{1em}z\in\mathbb{C},\hspace{1em}\alpha>0,\label{eq:E_a_a}
\end{equation}
which is straightforward to obtain using definitions (\ref{eq:E_a})
and (\ref{eq:gamma_prop}). Also, from (\ref{eq:e_a_cal_E_a}) and
(\ref{eq:e_a}), it follows that
\begin{equation}
\int_{0}^{z_{0}}\mathcal{E}_{\alpha}\left(\lambda^{1/\alpha}z\right)\dd z=\lambda^{-1/\alpha}\int_{0}^{z_{0}}e_{\alpha}\left(z;\lambda\right)\dd z=\lambda^{-1/\alpha}\left[E_{\alpha}\left(-\lambda z_{0}^{\alpha}\right)-1\right],\hspace{1em}z_{0}>0,\hspace{1em}\alpha>0,\hspace{1em}\lambda>0.\label{eq:int_cal_E_a}
\end{equation}
Moreover, when $\alpha\in\left(0,1\right)$, we have the following
integral representation 
\begin{equation}
\mathcal{E}_{\alpha}\left(z\right)=-\frac{\sin\left(\alpha\pi\right)}{\pi}{\displaystyle \int_{0}^{\infty}}e^{-rz}\frac{r^{\alpha}}{r^{2\alpha}+2r^{\alpha}\cos\left(\alpha\pi\right)+1}\dd r,\hspace{1em}0<\alpha<1,\hspace{1em}z>0,\label{eq:cal_E_a_alt}
\end{equation}
which can be obtained from \cite[eqs. (3.19)--(3.20), (3.24)--(3.25)]{GorMain}
by differentiation.\\

Small-argument asymptotics of the above functions follow directly
from definitions (\ref{eq:E_a})--(\ref{eq:cal_E_a})
\begin{equation}
E_{\alpha}\left(z\right)=1+\frac{z}{\Gamma\left(1+\alpha\right)}+\mathcal{O}\left(z^{2}\right),\hspace{1em}\hspace{1em}e_{\alpha}\left(z;\lambda\right)=-\frac{\alpha\lambda}{\Gamma\left(1+\alpha\right)}\frac{1}{z^{1-\alpha}}+\mathcal{O}\left(\frac{1}{z^{1-2\alpha}}\right),\hspace{1em}\hspace{1em}\left|z\right|\ll1,\label{eq:E_a_small}
\end{equation}
\begin{equation}
\mathcal{E_{\alpha}}\left(z\right)=-\frac{\alpha}{\Gamma\left(1+\alpha\right)}\frac{1}{z^{1-\alpha}}+\mathcal{O}\left(\frac{1}{z^{1-2\alpha}}\right),\hspace{1em}\hspace{1em}\left|z\right|\ll1.\label{eq:cal_E_a_small}
\end{equation}
Large-argument asymptotics of $\mathcal{E}_{\alpha}\left(z;\lambda\right)$
can be derived from that of $E_{\alpha}\left(z\right)$ given, for
example, in \cite[eqs. (3.4.14)--(3.4.15)]{GorKilMaiRog} (see also
\cite[eq. (1.2)]{Paris} for the same results for $\alpha\in\left(0,1\right)$)
\begin{equation}
E_{\alpha}\left(z\right)=\frac{1}{\alpha}\exp\left(z^{1/\alpha}\right)-\sum_{k=1}^{\infty}\frac{z^{-k}}{\Gamma\left(1-\alpha k\right)},\hspace{1em}\hspace{1em}E_{\alpha}\left(-z\right)=-\sum_{k=1}^{\infty}\frac{\left(-1\right)^{k}z^{-k}}{\Gamma\left(1-\alpha k\right)},\hspace{1em}\hspace{1em}z\gg1,\hspace{1em}\alpha\in\left(0,1\right)\cup\left(1,2\right),\label{eq:E_a_large}
\end{equation}
\begin{equation}
e_{\alpha}\left(z;\lambda\right)=\begin{cases}
-\frac{\alpha}{\lambda\Gamma\left(1-\alpha\right)}\frac{1}{z^{\alpha+1}}+\mathcal{O}\left(\frac{1}{z^{2\alpha+1}}\right), & z\gg1,\hspace{1em}\lambda>0,\hspace{1em}\alpha\in\left(0,1\right)\cup\left(1,2\right),\\
\frac{\lambda^{1/\alpha}}{\alpha}\exp\left(\lambda^{1/\alpha}z\right)+\frac{\alpha}{\lambda\Gamma\left(1-\alpha\right)}\frac{1}{z^{\alpha+1}}+\mathcal{O}\left(\frac{1}{z^{2\alpha+1}}\right), & z\gg1,\hspace{1em}\lambda<0,\hspace{1em}\alpha\in\left(0,1\right)\cup\left(1,2\right),
\end{cases}\label{eq:e_a_large}
\end{equation}
\begin{equation}
\mathcal{E_{\alpha}}\left(z\right)=-\frac{\alpha}{\Gamma\left(1-\alpha\right)}\frac{1}{z^{\alpha+1}}+\mathcal{O}\left(\frac{1}{z^{2\alpha+1}}\right),\hspace{1em}\hspace{1em}z\gg1,\hspace{1em}\alpha\in\left(0,1\right)\cup\left(1,2\right).\label{eq:cal_E_a_large}
\end{equation}

We also need some Laplace transforms of the above functions. We denote
the Laplace transform of $f$ as $\mathcal{L}\left[f\right]\left(s\right):=\int_{0}^{\infty}f\left(t\right)e^{-st}\dd t$.
In particular, we have (see \cite[eq. (3.14)]{GorMain} or \cite[eq. (4)]{Main})
\begin{equation}
\mathcal{L}\left[\mathcal{E}_{\alpha}\right]\left(s\right)=-\frac{1}{s^{\alpha}+1},\hspace{1em}\alpha>0,\label{eq:E_a_Lap}
\end{equation}
which implies that, for $\lambda>0$,
\begin{equation}
\mathcal{L}^{-1}\left[\frac{1}{s^{\alpha}+\lambda}\right]\left(t\right)=-\frac{1}{\lambda^{1-1/\alpha}}\mathcal{E}_{\alpha}\left(\lambda^{1/\alpha}t\right)=-\frac{1}{\lambda}e_{\alpha}\left(t;\lambda\right),\hspace{1em}\alpha>0.\label{eq:E_a_lmb_inv_Lap}
\end{equation}
To obtain similar result for $\lambda<0$, we use the formula \cite[eq. (1.93)]{SamKilMar}
\[
\mathcal{L}\left[\frac{1}{t^{1-\alpha}}E_{\alpha,\alpha}\left(t^{\alpha}\right)\right]\left(s\right)=\frac{1}{s^{\alpha}-1},\hspace{1em}\text{Re }s>1,\hspace{1em}\alpha>0,
\]
which implies that
\begin{align}
\mathcal{L}^{-1}\left[\frac{1}{s^{\alpha}-\left|\lambda\right|}\right]\left(t\right)= & \frac{1}{t^{1-\alpha}}E_{\alpha,\alpha}\left(\left|\lambda\right|t^{\alpha}\right)=\frac{\alpha}{t^{1-\alpha}}E_{\alpha}^{\prime}\left(\left|\lambda\right|t^{\alpha}\right)=\frac{1}{\left|\lambda\right|}\frac{\dd}{\dd t}E_{\alpha}\left(\left|\lambda\right|t^{\alpha}\right)\label{eq:E_a_lmb_inv_Lap_alt}\\
= & -\frac{1}{\lambda}\frac{\dd}{\dd t}E_{\alpha}\left(-\lambda t^{\alpha}\right)=-\frac{1}{\lambda}e_{\alpha}\left(t;\lambda\right).\nonumber 
\end{align}
Here, we used the identity (\ref{eq:E_a_a}) and the definition of
the function $\mathcal{E}_{a}$.\\
Note that (\ref{eq:E_a_lmb_inv_Lap_alt}) shows that (\ref{eq:E_a_lmb_inv_Lap})
is actually valid for all $\lambda\in\mathbb{R}$.

\subsection{Some singular integral equations and their solutions}

We consider here Abel's integral equations of the first and the second
kinds and two other related equations.
\begin{lem}
\label{lem:Abel_2nd}\cite[Thm 4.2]{GorKilMaiRog} Let $\alpha>0$,
$\lambda\in\mathbb{C}$ and $f\in L^{1}\left(0,T\right)$ for any
$T>0$. Then, the integral equation
\begin{equation}
u\left(t\right)+\frac{\lambda}{\Gamma\left(\alpha\right)}\int_{0}^{t}\frac{u\left(\tau\right)}{\left(t-\tau\right)^{1-\alpha}}\dd\tau=f\left(t\right),\hspace{1em}t\in\left(0,T\right),\label{eq:Abel_2nd}
\end{equation}
has a unique solution $u\in L^{1}\left(0,T\right)$ given by
\begin{align}
u\left(t\right)= & f\left(t\right)-\lambda\int_{0}^{t}\frac{E_{\alpha,\alpha}\left(-\lambda\left(t-\tau\right)^{\alpha}\right)}{\left(t-\tau\right)^{1-\alpha}}f\left(\tau\right)\dd\tau\label{eq:Abel_2nd_sol}\\
= & f\left(t\right)+\int_{0}^{t}e_{\alpha}\left(t-\tau;\lambda\right)f\left(\tau\right)\dd\tau,\hspace{1em}t\in\left(0,T\right).\nonumber 
\end{align}
\end{lem}

Note that the integral form of the solution on the second line of
(\ref{eq:Abel_2nd_sol}) can be seen due to the identity (\ref{eq:E_a_a}),
and it is also given elsewhere (see \cite[eq. (8.1.20)]{GorKilMaiRog},
\cite[eq. (2.12)]{GorMain}).
\begin{lem}
\label{lem:Abel_1st}\cite[Thm 2.1 \& Lem. 2.1]{SamKilMar} (see also
\cite[Sect. 8.1.1]{GorKilMaiRog}) Let $\alpha\in\left(0,1\right)$,
and assume that $f$ is such that $\int_{0}^{t}\frac{f\left(\tau\right)}{\left(t-\tau\right)^{\alpha}}\dd\tau$
is an absolutely continuous function on $\left[0,T\right]$ for some
$T>0$ and, moreover, ${\displaystyle \lim_{t\rightarrow0^{+}}}\int_{0}^{t}\frac{f\left(\tau\right)}{\left(t-\tau\right)^{\alpha}}\dd\tau=0$.
Then, the integral equation
\begin{equation}
\frac{1}{\Gamma\left(\alpha\right)}\int_{0}^{t}\frac{u\left(\tau\right)}{\left(t-\tau\right)^{1-\alpha}}\dd\tau=f\left(t\right),\hspace{1em}t\in\left(0,T\right),\label{eq:Abel_1st}
\end{equation}
has a unique solution $u\in L^{1}\left(0,T\right)$ given by
\begin{equation}
u\left(t\right)=\frac{1}{\Gamma\left(1-\alpha\right)}\frac{\dd}{\dd t}\int_{0}^{t}\frac{f\left(\tau\right)}{\left(t-\tau\right)^{\alpha}}\dd\tau,\hspace{1em}t\in\left(0,T\right).\label{eq:Abel_1st_sol}
\end{equation}
In particular, to satisfy the aforementioned conditions on $f$, it
is sufficient that $f$ is an absolutely continuous function on $\left[0,T\right]$.
In this case, the solution (\ref{eq:Abel_1st_sol}) can be written
in the alternative form
\begin{equation}
u\left(t\right)=\frac{1}{\Gamma\left(1-\alpha\right)}\left[\frac{f\left(0\right)}{t^{\alpha}}+\int_{0}^{t}\frac{f^{\prime}\left(\tau\right)}{\left(t-\tau\right)^{\alpha}}\dd\tau\right],\hspace{1em}t\in\left(0,T\right).\label{eq:Abel_1st_sol_alt}
\end{equation}
\end{lem}

\phantom{.}
\begin{cor}
\label{cor:Abel_1st} Let $\alpha\in\left[1,2\right)$, and assume
that $f$ is such that $\int_{0}^{t}\frac{f\left(\tau\right)}{\left(t-\tau\right)^{\alpha-1}}\dd\tau$
is an absolutely continuous function on $\left[0,T\right]$ for some
$T>0$ and, moreover, ${\displaystyle \lim_{t\rightarrow0^{+}}}\int_{0}^{t}\frac{f\left(\tau\right)}{\left(t-\tau\right)^{\alpha-1}}\dd\tau=0$.
Then, integral equation (\ref{eq:Abel_1st}) has a unique solution
$u$ such that $\int_{0}^{\cdot}u\left(\tau\right)\dd\tau\in L^{1}\left(0,T\right)$
and it is given by
\begin{equation}
u\left(t\right)=\frac{\Gamma\left(\alpha\right)\sin\left[\pi\left(\alpha-1\right)\right]}{\pi\left(\alpha-1\right)}\frac{\dd^{2}}{\dd t^{2}}\int_{0}^{t}\frac{f\left(\tau\right)}{\left(t-\tau\right)^{\alpha-1}}\dd\tau,\hspace{1em}t\in\left(0,T\right).\label{eq:Abel_1st_sol_alph12}
\end{equation}
 Moreover, if $f^{\prime}$ is absolutely continuous on $\left[0,T\right]$,
solution formula (\ref{eq:Abel_1st_sol_alph12}) can be rewritten
as
\begin{equation}
u\left(t\right)=\frac{\Gamma\left(\alpha\right)\sin\left[\pi\left(\alpha-1\right)\right]}{\pi\left(\alpha-1\right)}\left[-\frac{\alpha-1}{t^{\alpha}}f\left(0\right)+\frac{f^{\prime}\left(0\right)}{t^{\alpha-1}}+\int_{0}^{t}\frac{f^{\prime\prime}\left(\tau\right)}{\left(t-\tau\right)^{\alpha-1}}\dd\tau\right],\hspace{1em}t\in\left(0,T\right).\label{eq:Abel_1st_sol_alt_alph12}
\end{equation}
\end{cor}

\begin{proof}
Let us first assume that $\alpha\in\left(1,2\right)$. Denoting $\widetilde{u}\left(t\right):=\int_{0}^{t}u\left(\tau\right)\dd\tau$,
integration by parts yields
\[
\int_{0}^{t}\frac{u\left(\tau\right)}{\left(t-\tau\right)^{1-\alpha}}\dd\tau=\left(\alpha-1\right)\int_{0}^{t}\frac{\widetilde{u}\left(\tau\right)}{\left(t-\tau\right)^{2-\alpha}}\dd\tau,
\]
where the boundary terms vanish due to $\alpha>1$ and $\widetilde{u}\left(0\right)=0$.
Lemma \ref{lem:Abel_1st} now applies to the equation
\[
\frac{1}{\Gamma\left(\alpha-1\right)}\int_{0}^{t}\frac{\widetilde{u}\left(\tau\right)}{\left(t-\tau\right)^{2-\alpha}}\dd\tau=\frac{\Gamma\left(\alpha\right)}{\left(\alpha-1\right)\Gamma\left(\alpha-1\right)}f\left(t\right),\hspace{1em}t\in\left(0,T\right),
\]
and gives the existence of a unique solution $\tilde{u}\in L^{1}\left(0,T\right)$:
\begin{equation}
\widetilde{u}\left(t\right)=\frac{\Gamma\left(\alpha\right)}{\left(\alpha-1\right)\Gamma\left(\alpha-1\right)\Gamma\left(2-\alpha\right)}\frac{\dd}{\dd t}\int_{0}^{t}\frac{f\left(\tau\right)}{\left(t-\tau\right)^{\alpha-1}}\dd\tau,\hspace{1em}t\in\left(0,T\right).\label{eq:Abel_1st_sol_tild}
\end{equation}
Hence, differentiating and using the identity $\Gamma\left(\alpha-1\right)\Gamma\left(2-\alpha\right)=\pi/\sin\left[\pi\left(\alpha-1\right)\right]$
(see (\ref{eq:gamma_Euler})), we obtain (\ref{eq:Abel_1st_sol_alph12}).

Suppose now that $f$ is absolutely continuous on $\left[0,T\right]$,
we can then integrate by parts, taking into account that $\alpha<2$,
\begin{align}
\int_{0}^{t}\frac{f\left(\tau\right)}{\left(t-\tau\right)^{\alpha-1}}\dd\tau= & \int_{0}^{t}\frac{f\left(t-\tau\right)}{\tau^{\alpha-1}}\dd\tau=\frac{t^{2-\alpha}}{2-\alpha}f\left(0\right)+\frac{1}{2-\alpha}\int_{0}^{t}f^{\prime}\left(t-\tau\right)\tau^{2-\alpha}\dd\tau\label{eq:f_part_integr}\\
= & \frac{t^{2-\alpha}}{2-\alpha}f\left(0\right)+\frac{1}{2-\alpha}\int_{0}^{t}f^{\prime}\left(\tau\right)\left(t-\tau\right)^{2-\alpha}\dd\tau,\nonumber 
\end{align}
and thus
\begin{equation}
\frac{\dd}{\dd t}\int_{0}^{t}\frac{f\left(\tau\right)}{\left(t-\tau\right)^{\alpha-1}}\dd\tau=\frac{f\left(0\right)}{t^{\alpha-1}}+\int_{0}^{t}\frac{f^{\prime}\left(\tau\right)}{\left(t-\tau\right)^{\alpha-1}}\dd\tau.\label{eq:f_part_integr_der}
\end{equation}
If $f^{\prime}$ is absolutely continuous on $\left[0,T\right]$,
we can apply the same procedure to the integral on the right-hand
side before another differentiation and then we arrive at 
\[
\frac{\dd^{2}}{\dd t^{2}}\int_{0}^{t}\frac{f\left(\tau\right)}{\left(t-\tau\right)^{\alpha-1}}\dd\tau=-\frac{\alpha-1}{t^{\alpha}}f\left(0\right)+\frac{f^{\prime}\left(0\right)}{t^{\alpha-1}}+\int_{0}^{t}\frac{f^{\prime\prime}\left(\tau\right)}{\left(t-\tau\right)^{\alpha-1}}\dd\tau,
\]
which, upon insertion into (\ref{eq:Abel_1st_sol_alph12}), gives
(\ref{eq:Abel_1st_sol_alt_alph12}).

Now, we consider the situation when $\alpha=1$. Equation (\ref{eq:Abel_1st})
degenerates and its unique solution $u\left(t\right)=f^{\prime}\left(t\right)$
follows immediately by differentiation of both sides of the equation.
Taking into account that $\frac{\sin\left[\pi\left(\alpha-1\right)\right]}{\pi\left(\alpha-1\right)}\rightarrow1$
as $\alpha\rightarrow1$, upon an elementary simplification, we see
that this solution coincides precisely with (\ref{eq:Abel_1st_sol_alt_alph12})
for $\alpha=1$.
\end{proof}
\begin{lem}
\label{lem:spec_int_eq_2nd} Let $\alpha>0$, $\lambda\in\mathbb{R}$.
Then, the integral equation
\begin{equation}
u\left(t\right)-\lambda\int_{0}^{t}\mathcal{E}_{\alpha}\left(t-\tau\right)u\left(\tau\right)\dd\tau=f\left(t\right),\hspace{1em}t>0,\label{eq:cal_E_a_int_eq_2nd}
\end{equation}
has a unique solution given by
\begin{equation}
u\left(t\right)=f\left(t\right)+\frac{\lambda}{1+\lambda}\int_{0}^{t}e_{\alpha}\left(t-\tau;1+\lambda\right)f\left(\tau\right)\dd\tau,\hspace{1em}t>0.\label{eq:cal_E_a_int_eq_2nd_sol}
\end{equation}
\end{lem}

\begin{proof}
Let us denote $U\left(s\right):=\mathcal{L}\left[u\right]\left(s\right)$,
$F\left(s\right):=\mathcal{L}\left[f\right]\left(s\right)$ the Laplace
transforms of $u\left(t\right)$ and $f\left(t\right)$, respectively.
Application of the Laplace transformation to the both sides of (\ref{eq:cal_E_a_int_eq_2nd})
yields, upon using the convolution theorem and (\ref{eq:E_a_Lap}),
\begin{equation}
U\left(s\right)+\frac{\lambda}{s^{\alpha}+1}U\left(s\right)=F\left(s\right)\hspace{1em}\Longrightarrow\hspace{1em}U\left(s\right)=\frac{s^{\alpha}+1}{s^{\alpha}+1+\lambda}F\left(s\right)=\left[1-\frac{\lambda}{s^{\alpha}+1+\lambda}\right]F\left(s\right).\label{eq:cal_E_int_eq_Lap}
\end{equation}
Using (\ref{eq:E_a_lmb_inv_Lap}), inversion of Laplace transformation
of the rightmost equation above leads to (\ref{eq:cal_E_a_int_eq_2nd_sol}).
\end{proof}
\begin{cor}
\label{cor:spec_int_eq_2nd} Let $\alpha>0$, $\mu\geq0$, $\lambda\in\mathbb{R}$.
Then, the integral equation
\begin{equation}
u\left(t\right)-\frac{\lambda}{\mu^{1-1/\alpha}}\int_{0}^{t}\mathcal{E}_{\alpha}\left(\mu^{1/\alpha}\left(t-\tau\right)\right)u\left(\tau\right)\dd\tau=f\left(t\right),\hspace{1em}t>0,\label{eq:cal_E_a_int_eq_mu}
\end{equation}
has a unique solution given by
\begin{align}
u\left(t\right) & =f\left(t\right)+\frac{\lambda}{\mu+\lambda}\int_{0}^{t}e_{\alpha}\left(t-\tau;\mu+\lambda\right)f\left(\tau\right)\dd\tau\label{eq:cal_E_a_int_eq_2nd_mu_sol}\\
 & =f\left(t\right)+\frac{\lambda}{\left(\mu+\lambda\right)^{1-1/\alpha}}\int_{0}^{t}\mathcal{E}_{\alpha}\left(\left(\mu+\lambda\right)^{1/\alpha}\left(t-\tau\right)\right)f\left(\tau\right)\dd\tau,\hspace{1em}t>0.\nonumber 
\end{align}
\end{cor}

\begin{proof}
For $\mu>0$, the result follows from Lemma \ref{lem:spec_int_eq_2nd}
applied to the integral equation for $u\left(t/\mu^{1/\alpha}\right)$
with $\lambda/\mu$ and $f\left(t/\mu^{1/\alpha}\right)$ in place
of $\lambda$ and $f\left(t\right)$, respectively.

To treat the case $\mu=0$, we observe, using (\ref{eq:cal_E_a_small})
and (\ref{eq:gamma_prop}), that in the limit $\mu\rightarrow0$,
equation (\ref{eq:cal_E_a_int_eq_mu}) becomes (\ref{eq:Abel_2nd})
whose solution, given by Lemma \ref{lem:Abel_2nd}, coincides exactly
with (\ref{eq:cal_E_a_int_eq_2nd_mu_sol}) for $\mu=0$.
\end{proof}
\begin{lem}
\label{lem:spec_int_eq_1st} Let $\alpha\in\left(0,1\right)$, and
assume that $f$ is an absolutely continuous function on every subinterval
of $\mathbb{R}_{+}$. Then, the integral equation 
\begin{equation}
\int_{0}^{t}\mathcal{E}_{\alpha}\left(t-\tau\right)u\left(\tau\right)\dd\tau=f\left(t\right),\hspace{1em}t>0,\label{eq:cal_E_a_int_eq_1st}
\end{equation}
has a unique solution given by
\begin{align}
u\left(t\right) & =-\frac{f\left(0\right)}{\Gamma\left(1-\alpha\right)t^{\alpha}}-\int_{0}^{t}\left(1+\frac{1}{\Gamma\left(1-\alpha\right)}\frac{1}{\left(t-\tau\right)^{\alpha}}\right)f^{\prime}\left(\tau\right)\dd\tau\label{eq:cal_E_a_int_eq_1st_sol}\\
 & =-f\left(t\right)-\text{\ensuremath{\left(\frac{1}{\Gamma\left(1-\alpha\right)t^{\alpha}}-1\right)}}f\left(0\right)-\frac{1}{\Gamma\left(1-\alpha\right)}\int_{0}^{t}\frac{f^{\prime}\left(\tau\right)}{\left(t-\tau\right)^{\alpha}}\dd\tau,\hspace{1em}t>0.\nonumber 
\end{align}
\end{lem}

\begin{proof}
The proof is ideologically similar to that of Lemma \ref{lem:spec_int_eq_2nd}.
Let us denote $U\left(s\right):=\mathcal{L}\left[u\right]\left(s\right)$,
$F\left(s\right):=\mathcal{L}\left[f\right]\left(s\right)$ the Laplace
transforms of $u\left(t\right)$ and $f\left(t\right)$. Upon Laplace
transformation of (\ref{eq:cal_E_a_int_eq_1st}), using the convolution
theorem and (\ref{eq:E_a_Lap}), we have
\[
-U\left(s\right)\frac{1}{s^{\alpha}+1}=F\left(s\right)\hspace{1em}\Longrightarrow\hspace{1em}U\left(s\right)=-F\left(s\right)\left(1+s^{\alpha}\right)=-\left[sF\left(s\right)-f\left(0\right)\right]\left(\frac{1}{s}+\frac{1}{s^{1-\alpha}}\right)-f\left(0\right)\left(\frac{1}{s}+\frac{1}{s^{1-\alpha}}\right).
\]
Hence, employing $\mathcal{L}\left[t^{\beta-1}\right]\left(s\right)=\Gamma\left(\beta\right)/s^{\beta}$,
$\beta>0$ (see e.g. \cite[pp. 318--319]{Doet}), and, in particular,
$\mathcal{L}\left[1\right]\left(s\right)=1/s$, we invert the transform
to obtain (\ref{eq:cal_E_a_int_eq_1st_sol}). 
\end{proof}
\begin{cor}
\label{cor:spec_int_eq_1st_mu} Let $\alpha\in\left(0,1\right)$,
$\mu\geq0$ and assume that $f$ is an absolutely continuous function
on every bounded subinterval of $\mathbb{R}_{+}$. Then, the integral
equation 
\begin{equation}
-\frac{1}{\mu^{1-1/\alpha}}\int_{0}^{t}\mathcal{E}_{\alpha}\left(\mu^{1/\alpha}\left(t-\tau\right)\right)u\left(\tau\right)\dd\tau=f\left(t\right),\hspace{1em}t>0,\label{eq:cal_E_a_int_eq_1st_mu}
\end{equation}
has a unique solution given by
\begin{align}
u\left(t\right) & =\frac{f\left(0\right)}{\Gamma\left(1-\alpha\right)t^{\alpha}}+\int_{0}^{t}\left(\mu+\frac{1}{\Gamma\left(1-\alpha\right)}\frac{1}{\left(t-\tau\right)^{\alpha}}\right)f^{\prime}\left(\tau\right)\dd\tau\label{eq:cal_E_a_int_eq_1st_sol_mu}\\
 & =\mu f\left(t\right)+\text{\ensuremath{\left(\frac{1}{\Gamma\left(1-\alpha\right)t^{\alpha}}-\mu\right)}}f\left(0\right)+\frac{1}{\Gamma\left(1-\alpha\right)}\int_{0}^{t}\frac{f^{\prime}\left(\tau\right)}{\left(t-\tau\right)^{\alpha}}\dd\tau,\hspace{1em}t>0.\nonumber 
\end{align}
\end{cor}

\begin{proof}
For $\mu>0$, the result follows from Lemma \ref{lem:spec_int_eq_1st}
applied to the integral equation for $u\left(t/\mu^{1/\alpha}\right)$
with $-\mu f\left(t/\mu^{1/\alpha}\right)$ in place of $f\left(t\right)$.

To treat the case $\mu=0$, we observe, using (\ref{eq:cal_E_a_small})
and (\ref{eq:gamma_prop}), that in the limit $\mu\rightarrow0$,
equation (\ref{eq:cal_E_a_int_eq_1st_mu}) becomes (\ref{eq:Abel_1st})
whose solution, given by Lemma \ref{lem:Abel_1st}, coincides exactly
with (\ref{eq:cal_E_a_int_eq_1st_sol_mu}) for $\mu=0$.
\end{proof}
\begin{lem}
\label{lem:spec_int_eq_1st_mu_alph12} Let $\alpha\in\left[1,2\right)$,
$\mu\geq0$, and assume that $f^{\prime}$ is an absolutely continuous
function on every bounded subinterval of $\mathbb{R}_{+}$. Then,
integral equation (\ref{eq:cal_E_a_int_eq_1st_mu}) has a unique solution
given by
\begin{align}
u\left(t\right)=\mu f\left(t\right)-\frac{\alpha-1}{\Gamma\left(2-\alpha\right)t^{\alpha}}f\left(0\right)+\frac{1}{\Gamma\left(2-\alpha\right)t^{\alpha-1}}f^{\prime}\left(0\right)+\frac{1}{\Gamma\left(2-\alpha\right)}\int_{0}^{t}\frac{f^{\prime\prime}\left(\tau\right)}{\left(t-\tau\right)^{\alpha-1}}\dd\tau,\hspace{1em}t>0.\label{eq:cal_E_a_int_eq_1st_sol_mu_alph12}
\end{align}
\end{lem}

\begin{proof}
\uline{Case \mbox{$\mu>0$}}:

First, let us consider $\alpha\in\left(1,2\right)$. Let $\widetilde{u}\left(t\right):=\int_{0}^{t}u\left(\tau\right)\dd\tau$.
Since $\mathcal{E}_{\alpha}\left(0\right)=0$ for $\alpha>1$ and
$\widetilde{u}\left(0\right)=0$, integration by parts of (\ref{eq:cal_E_a_int_eq_1st_mu})
leads to
\[
-\frac{1}{\mu^{1-2/\alpha}}\int_{0}^{t}\mathcal{E}_{\alpha}^{\prime}\left(\mu^{1/\alpha}\left(t-\tau\right)\right)\widetilde{u}\left(\tau\right)\dd\tau=f\left(t\right),\hspace{1em}t>0.
\]
Denoting $\widetilde{U}\left(s\right):=\mathcal{L}\left[\widetilde{u}\right]\left(s\right)$,
$F\left(s\right):=\mathcal{L}\left[f\right]\left(s\right)$ the Laplace
transforms of $\widetilde{u}\left(t\right)$ and $f\left(t\right)$,
respectively, we apply Laplace transformation to the above equation
and use the convolution theorem and $\mathcal{L}\left[\mathcal{E}_{\alpha}^{\prime}\left(\mu^{1/\alpha}\cdot\right)\right]$$\left(s\right)=-\frac{\mu^{1-2/\alpha}s}{s^{\alpha}+\mu}$
(which easily follows from (\ref{eq:E_a_Lap})) to arrive at
\[
\widetilde{U}\left(s\right)\frac{s}{s^{\alpha}+\mu}=F\left(s\right)\hspace{1em}\Longrightarrow\hspace{1em}\widetilde{U}\left(s\right)=F\left(s\right)\left(\frac{\mu}{s}+s^{\alpha-1}\right)=\left[sF\left(s\right)-f\left(0\right)\right]\left(\frac{\mu}{s^{2}}+\frac{1}{s^{2-\alpha}}\right)+f\left(0\right)\left(\frac{\mu}{s^{2}}+\frac{1}{s^{2-\alpha}}\right).
\]
Therefore, upon inversion of Laplace transformation, using that $\mathcal{L}\left[t^{1-\alpha}\right]\left(s\right)=\Gamma\left(2-\alpha\right)/s^{2-\alpha}$,
$\alpha<2$, (and, in particular, $\mathcal{L}\left[t\right]\left(s\right)=1/s^{2}$),
we obtain
\begin{equation}
\widetilde{u}\left(t\right)=\int_{0}^{t}f^{\prime}\left(\tau\right)\left[\mu\left(t-\tau\right)+\frac{1}{\Gamma\left(2-\alpha\right)\left(t-\tau\right)^{\alpha-1}}\right]\dd\tau+f\left(0\right)\left(\mu t+\frac{1}{\Gamma\left(2-\alpha\right)t^{\alpha-1}}\right),\hspace{1em}t>0.\label{eq:spec_int_eq_1st_sol_tild}
\end{equation}
Note that, similarly to (\ref{eq:f_part_integr})--(\ref{eq:f_part_integr_der}),
we have, for $\alpha\in\left(1,2\right)$, 
\[
\int_{0}^{t}\frac{f^{\prime}\left(\tau\right)}{\left(t-\tau\right)^{\alpha-1}}\dd\tau=\frac{t^{2-\alpha}}{2-\alpha}f^{\prime}\left(0\right)+\frac{1}{2-\alpha}\int_{0}^{t}f^{\prime\prime}\left(\tau\right)\left(t-\tau\right)^{2-\alpha}\dd\tau,
\]
\[
\frac{\dd}{\dd t}\int_{0}^{t}\frac{f^{\prime}\left(\tau\right)}{\left(t-\tau\right)^{\alpha-1}}\dd\tau=\frac{1}{t^{\alpha-1}}f^{\prime}\left(0\right)+\int_{0}^{t}\frac{f^{\prime\prime}\left(\tau\right)}{\left(t-\tau\right)^{\alpha-1}}\dd\tau,
\]
and hence, by differentiation of (\ref{eq:spec_int_eq_1st_sol_tild}),
we deduce (\ref{eq:cal_E_a_int_eq_1st_sol_mu_alph12}).

It remains to consider the situation when $\alpha=1$. To this effect,
we recall (\ref{eq:e_a_cal_E_a}) and (\ref{eq:e_a}) which imply
that $\mathcal{E}_{1}\left(\mu t\right)=-e^{-\mu t}$, and hence equation
(\ref{eq:cal_E_a_int_eq_1st_mu}) can be recast as
\[
\int_{0}^{t}e^{\mu\tau}u\left(\tau\right)\dd\tau=e^{\mu t}f\left(t\right),\hspace{1em}t>0,
\]
which is then immediately solved by the differentiation to give $u\left(t\right)=\mu f\left(t\right)+f^{\prime}\left(t\right)$.
This solution coincides with that given by (\ref{eq:cal_E_a_int_eq_1st_sol_mu_alph12}),
upon substitution $\alpha=1$ and relevant simplifications (integration
and cancellations).

\uline{Case \mbox{$\mu=0$}}:

Recalling (\ref{eq:cal_E_a_small}) and (\ref{eq:gamma_prop}), we
observe that in the limit $\mu\rightarrow0$, equation (\ref{eq:cal_E_a_int_eq_1st_mu})
becomes (\ref{eq:Abel_1st}). Solution of the latter equation, for
$\alpha\in\left[1,2\right)$, is given by Corollary \ref{cor:Abel_1st}.
Using the identities $\Gamma\left(\alpha-1\right)\Gamma\left(2-\alpha\right)=\pi/\sin\left[\pi\left(\alpha-1\right)\right]$
and $\Gamma\left(\alpha\right)=\left(\alpha-1\right)\Gamma\left(\alpha-1\right)$,
this solution can be seen to coincide with (\ref{eq:cal_E_a_int_eq_1st_sol_mu_alph12})
for $\mu=0$.
\end{proof}
\begin{lem}
\label{lem:log_int_eq} Suppose that $f\in C^{2}\left(\left[-a,a\right]\right)$,
then, if $a\neq2$, the integral equation
\[
-\int_{-a}^{a}\log\left|x-\xi\right|u\left(\xi\right)\dd\xi=f\left(x\right),\hspace{1em}x\in\left(-a,a\right),
\]
admits a solution given by
\[
u\left(x\right)=-\frac{1}{\left(a^{2}-x^{2}\right)^{1/2}}\left[\fint_{-a}^{a}\frac{\left(a^{2}-\xi^{2}\right)^{1/2}f^{\prime}\left(\xi\right)}{\xi-x}\dd\xi+\frac{1}{\log\left(a/2\right)}\int_{-a}^{a}\frac{f\left(\xi\right)}{\left(a^{2}-\xi^{2}\right)^{1/2}}\dd\xi\right],\hspace{1em}x\in\left(-a,a\right).
\]
Moreover, this solution is unique in the class of H{\"o}lder continuous
functions with possible integrable singularities at the endpoints
of the interval $\left[-a,a\right]$. 
\end{lem}

\begin{proof}
See e.g. \cite[p. 428]{Carr} for the case $a=1$ and \cite[p. 591]{Gakh}
for arbitrary $a>0$, $a\neq2$.
\end{proof}

\end{document}